\documentclass[10pts]{IEEEtran}
\usepackage{dsfont}
\usepackage{times}
\usepackage{graphicx}
\usepackage{graphics,color,oplotsymbl,bbding} % for pdf, bitmapped graphics files
\usepackage{rotating}
\usepackage{algorithmic}
\usepackage{comment}
\usepackage{algorithm,amssymb}
\usepackage{epsfig}
\usepackage{xspace}
\usepackage{flushend}
\usepackage[cmex10]{amsmath}
\usepackage{url}
\usepackage{hyperref}
\usepackage{kpfonts}
\usepackage{bm} 
\usepackage{cite}
\usepackage{textcomp}
\usepackage{grffile}
\usepackage{amsmath,amssymb,amsthm}
\usepackage{algorithm}
\usepackage{algorithmic}
\usepackage{nicefrac}
\usepackage{colortbl}
\usepackage{overpic}
\usepackage{tabularx}

\newcommand {\R} {{\mathbb R}}

\newcommand {\diag} {\mbox{diag}}

\newcommand{\cR}{\mathcal{R}}
\newcommand{\cA}{\mathcal{A}}
\newcommand{\cI}{\mathcal{I}}
\newcommand{\cO}{\mathcal{O}}
\newcommand{\cG}{\mathcal{G}}

\newcommand{\OGD}{\texttt{OGD}\xspace} 
\newcommand{\DA}{\texttt{DA}\xspace} 
\newcommand{\DAQ}{\texttt{DAQ}\xspace} 
 
\newcommand{\RRM}{\texttt{RRM}\xspace}
\newcommand{\BR}{\texttt{BR} \xspace}

\DeclareMathOperator*{\argmax}{arg\,max}

\newcommand {\SDSC}{{\bf SDSC}\xspace}

\newtheorem{lemma}{Lemma}

\newtheorem{corollary}{Corollary}
\newtheorem{theorem}{Theorem}
\newtheorem{definition}{Definition}

\newtheorem{remark}{Remark}
\newtheorem{assumption}{Assumption}

\usepackage{subcaption}
\usepackage{float}
\usepackage{booktabs}

\usepackage{array,multirow}

\usepackage{xcolor}

\definecolor{lightgray}{gray}{0.9}
\definecolor{color1}{RGB}{255, 204, 204}
\definecolor{color2}{RGB}{255, 247, 188} % Jaune pâle
\definecolor{color3}{RGB}{210, 235, 255} % Bleu clair
\definecolor{color4}{RGB}{235, 215, 250} % Violet pastel
\definecolor{color5}{RGB}{255, 220, 220} % Rouge pastel
\definecolor{color6}{RGB}{220, 255, 220} % Vert pastel
\definecolor{color7}{RGB}{220, 220, 255} % Bleu pastel
\definecolor{color8}{RGB}{255, 240, 200} % Orange pastel
\definecolor{color9}{RGB}{245, 222, 179} % Beige pastel
\definecolor{color10}{RGB}{230, 230, 250} % Lavande pastel

\definecolor{propositiongray}{gray}{0.85} % Un peu plus foncé
\definecolor{lemmagray}{gray}{0.8} % Moyen clair
\definecolor{corollarygray}{gray}{0.75} % Le plus 

\usepackage{mdframed}

\newmdtheoremenv[
  linecolor=black,
  backgroundcolor=color5,
  linewidth=2pt,
  innerleftmargin=5pt,
  innerrightmargin=5pt,
  innertopmargin=5pt,
  innerbottommargin=5pt
]{theoremSp}{Theorem}

\newmdtheoremenv[
  linecolor=black,
  backgroundcolor=color3,
  linewidth=2pt,
  innerleftmargin=5pt,
  innerrightmargin=5pt,
  innertopmargin=5pt,
  innerbottommargin=5pt
]{lemmaSp}{Lemma}

\newmdtheoremenv[
  linecolor=black,
  backgroundcolor=color2,
  linewidth=2pt,
  innerleftmargin=5pt,
  innerrightmargin=5pt,
  innertopmargin=5pt,
  innerbottommargin=5pt
]{corollarySp}{Corollary}

\newmdtheoremenv[
  linecolor=black,
  backgroundcolor=color3,
  linewidth=2pt,
  innerleftmargin=5pt,
  innerrightmargin=5pt,
  innertopmargin=5pt,
  innerbottommargin=5pt
]{propositionSp}{Proposition}

\newmdtheoremenv[
  linecolor=black,
  backgroundcolor=color6,
  linewidth=2pt,
  innerleftmargin=5pt,
  innerrightmargin=5pt,
  innertopmargin=5pt,
  innerbottommargin=5pt
]{assumptionSp}{Assumption}

\newmdtheoremenv[
  linecolor=black,
  backgroundcolor=color9,
  linewidth=2pt,
  innerleftmargin=5pt,
  innerrightmargin=5pt,
  innertopmargin=5pt,
  innerbottommargin=5pt
]{definitionSp}{Definition}

\IEEEoverridecommandlockouts 

\begin{document}
\title{Learning in Proportional Allocation\\
       Auctions Games}

\author{Younes Ben Mazziane$^{1}$, Cleque-Marlain Mboulou Moutoubi$^{1}$, \\ Eitan Altman$^{1,2}$ and Francesco De Pellegrini$^{1}$ \thanks{$^{1}$LIA, Avignon university, Avignon, France; $^{2}$INRIA, Sophia Antipolis, France.}}

% \author{Y. Ben Mazziane$^{1}$, C. Mboulou$^{1}$, F. De Pellegrini$^{1}$ and E. Altman$^{1,2}$\thanks{$^{1}$LIA, Avignon university, Avignon, France; $^{2}$INRIA, Sophia Antipolis, France.}}
\maketitle  
% typeset the header of the contribution
\thispagestyle{empty}
\begin{abstract}
The \textit{Kelly} or \textit{proportional allocation} mechanism is a simple and efficient auction-based scheme that distributes an infinitely divisible resource proportionally to the agents' bids. %When agents are aware of the allocation rule, their interactions form a game, that has been extensively studied. 
When agents are aware of the allocation rule, their interactions form a game extensively studied in the literature. This paper examines the less explored repeated Kelly game, focusing mainly on utilities that are logarithmic in the allocated resource fraction. We first derive this logarithmic form from fairness–throughput trade-offs in wireless network slicing, and then prove that the induced stage game admits a unique Nash equilibrium (NE). For the repeated play, we prove convergence to this NE under three behavioral models: (i) all agents use Online Gradient Descent (\texttt{OGD}), (ii) all agents use Dual Averaging with a quadratic regularizer (\texttt{DAQ}) (a variant of the Follow-the-Regularized leader algorithm), and (iii) all agents play myopic best responses (\texttt{BR}). Our convergence results hold even when agents use personalized learning rates in \texttt{OGD} and \texttt{DAQ} (e.g., tuned to optimize individual regret bounds), and they extend to a broader class of utilities that meet a certain sufficient condition. Finally, we complement our theoretical results with extensive simulations of the repeated Kelly game under several behavioral models, comparing them in terms of convergence speed to the NE, and per-agent time-average utility. The results suggest that \texttt{BR} achieves the fastest convergence and the highest time-average utility, and that convergence to the stage-game NE may fail under heterogeneous update rules. 
\end{abstract} 

%%%%%%%%%%%%%%%%%%%%%%%%%%%%%%%%%%%%%%%%%%%%%%%%%%%%%%%%%%%%%%%%%%%%%%%%%%%%%%%%%%%%%%%%%%%%%%%%%%%%
%%%%%%%%%%%%%%%%%%%%%%%%%%%%%%%%%%%%%%%%%%%%%%%%%%%%%%%%%%%%

\section{Introduction}\label{sec:intro}
% Context of auction based decentralized allocation

% \YBnote{Perhaps, we can write a small paragraph citing these papers.\begin{itemize}
%     \item This is a TON paper studying best-response dyanmics in repeated Kelly games: https://arxiv.org/pdf/1612.08446

%     \item Another TON paper for network slicing and proportional allocation rule: https://arxiv.org/pdf/1705.00582 

%     \item Another TON paper, for mean-field games for resource sharing in cloud based networks: https://ieeexplore.ieee.org/document/7015619

%     \item TON paper for network slicing games: https://ieeexplore.ieee.org/document/10107421
% \end{itemize}}
Decentralized resource allocation in large-scale systems is a fundamental problem extensively studied in network economics~\cite{johari_thesis2004}. In this context, a resource owner seeks to distribute resources among multiple agents to optimize an objective, such as maximizing \textit{social welfare}, namely, the aggregate net benefit of the agents, or their own revenue. It is standard to assume that the resource owner may have partial or lack information about the agents' utilities or preferences. Instead, they depend on signals~\cite{basar_JSAC_PoA1} provided by the agents, such as declared valuations, willingness to pay, or other indirect indicators of agents' preferences. Moreover, agents often act selfishly and strategically in order to maximize their benefits. This problem is prevalent in various technological domains, including bandwidth allocation in communication networks~\cite{TuffinInfocom04BandwidthAuction}, task scheduling in cloud computing~\cite{Trans_SC_auction_schedulingTasksCloud}, energy distribution in smart grids~\cite{SaadW_auction_smartGrid}, and pricing mechanisms in shared transportation systems~\cite{Auction_RideSharingICDE19}.
\begin{figure}[t]
    \centering
    \includegraphics[width=1\linewidth]{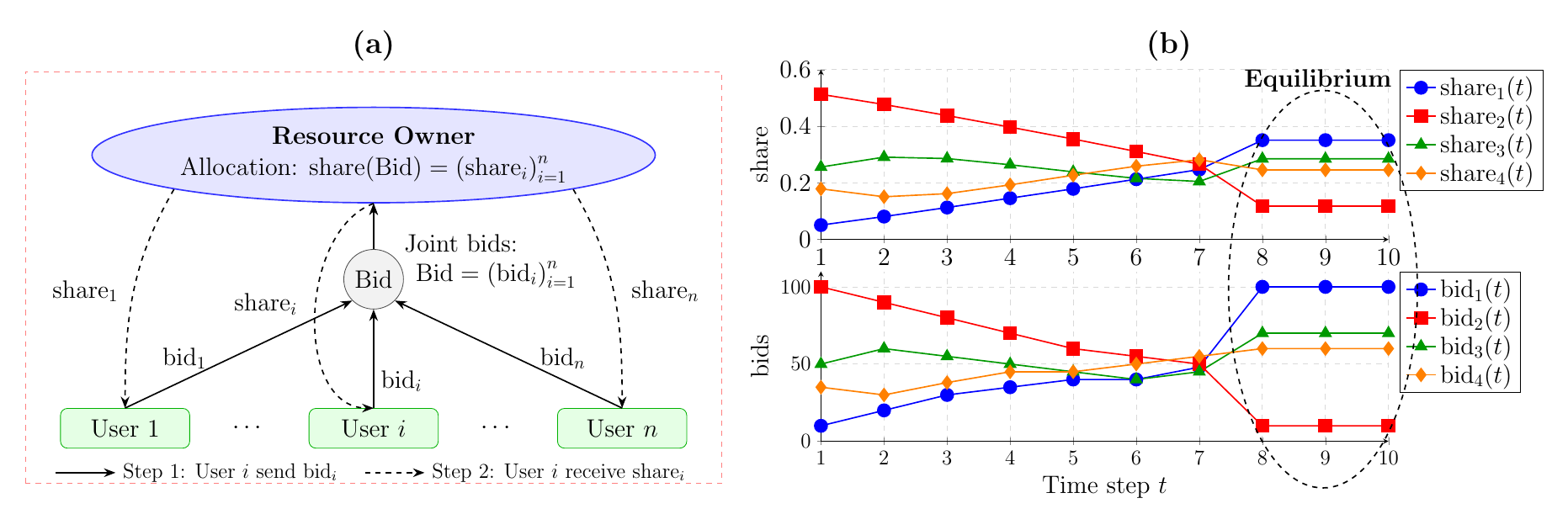}
{\captionsetup{aboveskip=2pt,belowskip=0pt}
    \caption{Repeated resource allocation game.}
    \label{fig:example_repeated_Kelly}}
\end{figure}

% Description of the Kelly
The \textit{Kelly} or \textit{proportional allocation} mechanism stands out among decentralized resource allocation mechanisms for its simplicity and efficiency~\cite{Kelly_Original97,johari_thesis2004}. In its basic form, agents submit bids to secure shares of a finite, infinitely divisible resource, with allocations distributed proportionally to their bids. In a generalized formulation \cite{basar_JSAC_PoA1}, each user’s allocation is determined by a \textit{weighting function} of their bid, enabling diverse allocation strategies. In particular, the classic Kelly mechanism arises when this weighting function is simply the identity for all agents. 

Many works have shown that the Kelly mechanism enjoys strong social-welfare optimality guarantees across several settings: (i) agents with unlimited budgets who are either \textit{price takers} (unaware of the allocation rule)~\cite{Kelly_Original97} or \textit{price anticipator} (aware of it)~\cite{JohariEfficiencyLoss_Math_Op_Res04,basar_JSAC_PoA1,VCG_Kelly_Hajek_JSAC_07}, and (ii) environments in which price anticipator agents have budget constraints~\cite{KellyBudgetStoc_SW_LSW13,LPoa_Constant_Caragiannis16, CaragiannisV18_EC_Constraint_Bids}. More specifically, price anticipator agents induce a competitive game with continuous action sets, and the guarantees in terms of social welfare hold exclusively at a Nash Equilibrium (NE) of this game, that we refer to as the \textit{Kelly game} in the sequel. 

In practice, however, agents are not necessarily aware of the utilities of other agents.
A more realistic setting is when agents know only their own utilities but adapt their bids over repeated synchronous rounds based on feedback from previous rounds, e.g., the aggregate bid. This motivates the study of the repeated Kelly game. Figure~\ref{fig:example_repeated_Kelly} illustrates this setting, where at each round, agents compete over a new resource by submitting bids based on outcomes from previous rounds. They then receive a fraction of the resource according to the Kelly mechanism. In this setting, rational agents aim to maximize their time-average utility.

To our knowledge, only a few works have examined the repeated Kelly game~\cite{even2009convergence,Mertikoupoulous_ComNet_17,elkind2024_FP_NE}, and they focus on the case where agents’ utilities are linear in the fraction of the allocated resource. This case coincides with \textit{Tullock} (rent-seeking) contests~\cite{perez1992_rentseeking}, also known as lottery contests~\cite{elkind2024_FP_NE}. In particular, \cite{even2009convergence} shows that if every player uses any \textit{no-regret} bidding algorithm, then each player’s average utility converges to their stage-game NE utility. \cite{Mertikopoulos_potential_MW_2017} proves that under a specific bidding rule used by all agents, the sequence of actions converges to a NE of the stage game. \cite{elkind2024_FP_NE} establishes convergence to this equilibrium when all agents employ fictitious-play updates. On the other hand, we study the repeated Kelly game with a general class of utilities that include logarithmic ones.

%%%%%%%%%%%%%%%%%%%%%%%%%%%%%%%%%%%%%%%%%%%
\subsection{Contributions}
%%%%%%%%%%%%%%%%%%%%%%%%%%%%%%%%%%%%%%%%%%%

Our contributions are summarized as follows:
\begin{enumerate}
    \item  We show that a practical scenario of interest induces a repeated Kelly game with logarithmic utilities.

    \item  We derive a tractable sufficient condition ensuring that the stage game satisfies Rosen’s Diagonal Strict Concavity (\textbf{DSC}) with some vector $\bm{r}\succ \bm{0}$ ($\bm{r}$-DSC)~\cite{rosen}, equivalently, $\bm{r}$-monotonicity, and thus admits a unique Nash equilibrium. The condition reduces to verifying negativity of a scalar function, and we show it holds for logarithmic utilities.

    \item For repeated Kelly games satisfying our sufficient $\bm r$-DSC condition, and with utilities that differ only by multiplicative factors, we prove convergence when all agents use either Online Gradient Descent ($\OGD$) or Dual Averaging with a quadratic regularizer ($\DAQ$).
    
    % whose stage utilities satisfy~$\bm{r}$-DSC and differ only by multiplicative factors, In particular, this yields convergence for the logarithmic-utility instance.

    \item We establish convergence of best-response dynamics in the repeated Kelly game under logarithmic utilities. 

    \item We conduct extensive numerical simulations to validate our theoretical results, and complement them with additional scenarios in which agents run heterogeneous learning dynamics.

\end{enumerate}

\medskip 

% \noindent \textbf{Log-utility model from practice.} While prior work on the repeated Kelly game has largely focused on utilities that are linear in the allocated share, we primarily study 
% this game when this dependence is logarithmic. 
% Section~\ref{ss:Example_log_utility}
% shows that in the problem of bandwidth allocation for network slicing, 
% the Kelly mechanism naturally induces utilities which are logarithmic in the allocated fraction. There, the objective of an agent, e.g., a slice tenant, is to balance throughput and fairness across a group of users.

% More broadly, concave utilities such as those of the logarithmic type are meant to express a diminishing marginal utility with additional resources.

We provide more details about our contributions. 

\medskip 

\noindent\textbf{Rosen's $\bm r$-DSC and uniqueness of the Nash equilibrium.}
A standard way to establish \textbf{DSC} is to show that a certain $n\times n$ matrix (with $n$ the number of agents) is negative definite over the action set. 
In general, checking negative definiteness requires  $\cO(n^{2})$ memory and  $\cO(n^{3})$ time. In contrast, Theorem~\ref{thm:SDSC_log} exploits the structure of the repeated Kelly game to reduce this verification to $\cO(n)$ time and $\cO(1)$ memory. This tractable condition also enables proving that there exists a vector $\bm{r}$ for which $\bm{r}$-\textbf{DSC} holds under logarithmic utilities. Moreover, proving \textbf{DSC} extends prior uniqueness guarantees of the NE to arbitrary convex action sets, which accommodates budget constraints and Kelly mechanisms with general weighting functions.

% Establishing Rosen’s Diagonal Strict Concavity is a key step toward proving convergence of no-regret learning dynamics. It also implies uniqueness of the Nash equilibrium under convex actions set, which let us accommodate budget constraints and Kelly mechanisms with general weighting functions, which is new with respect to previous results. More generally, DSC is a useful structural property even in games with coupled constraints~\cite{rosen}. 

\medskip 

\noindent \textbf{Convergence of no-regret learning to the NE.} Previous results show that convergence of \texttt{OGD} is guaranteed when the stage game satisfies $\bm r$-\textbf{DSC} for some $\bm r>\bm 0$~\cite{zhou_hal_OGD_convergence}, whereas convergence of \texttt{DAQ} requires the stronger condition $\bm 1$-\textbf{DSC}~\cite{Mertikopoulos_VI_2019_MP}. Under our $\bm r$-\textbf{DSC} sufficient condition---which holds for logarithmic utilities---convergence of \texttt{OGD} follows immediately. However, these previous results impose a common learning rate across agents, and our $\bm 1$-\textbf{DSC} sufficient condition holds only for homogeneous logarithmic utilities. 
We show in Theorem~\ref{thm:OGD_convergence_log} that, under affine heterogeneity in utilities (e.g., utilities share the same logarithmic form but differ by agent-specific multiplicative factors), \texttt{OGD} still converges to the stage-game NE when agents use regret-optimal learning rates. Under the same heterogeneity model, Theorem~\ref{cor:convergence_DAQ} further establishes that, assuming our $\bm r$-\textbf{DSC} sufficient condition (for some $\bm r>\bm 0$), \texttt{DAQ} also converges under personalized regret-optimal learning rates.

% The \texttt{OGD} guarantee therefore applies under our $\bm{r}$-DSC result, but previous results impose a common learning rate across agents. This requirement is generally incompatible with heterogeneous utilities: even when all players share the same underlying logarithmic form, their utilities may differ by an agent-dependent positive scaling, leading to different regret-optimal learning rates. We prove in Theorem~\ref{thm:OGD_convergence_log} that, under this affine heterogeneity, \texttt{OGD} still converges when each agent uses its own utility-tuned learning rate. On the other hand, establishing $\bm{1}$-DSC under heterogeneous utilities is more challenging and is not covered by our condition. 

\medskip 

\noindent \textbf{Convergence of best response dynamics.}
When agents use best response dynamics, we model the iterates of agents as fixed point iterations. We then derive closed form expressions of the Jacobian of the fixed point operator, and we prove that it is a contraction. Leveraging this,  Theorem~\ref{thm:BR_converges_NE} proves convergence of the system to the NE of the stage game and shows that the convergence speed is linear.

\medskip

\noindent \textbf{Numerical simulations.} We simulate the bidding algorithms under both \textit{homogeneous dynamics}, where agents use the same update rule, and \textit{heterogeneous dynamics}, where two update rules coexist in the population. Under \textit{homogeneous dynamics}, the simulations confirm our theoretical convergence results to the stage game NE, and indicate that, in terms of both convergence speed and time-average utility, $\BR$ performs best, followed by $\OGD$, and then $\DAQ$. Under \textit{heterogeneous dynamics}, the results suggest that convergence to the stage game NE may fail. However, the resulting time-average utilities remain similar across algorithms and close to the NE ones, with $\BR$ consistently better in the considered settings.

% We complement our theoretical results with an extensive numerical study of the repeated Kelly game under logarithmic utilities.  The simulations serve two main purposes.  First, they validate the predicted convergence properties of no-regret learning and best-response dynamics, both in utility and action spaces, and illustrate the differences in convergence speed across algorithms. Second, they quantify the impact of heterogeneity both in agents’ valuations and in learning dynamics on efficiency, stability, and equilibrium selection. In particular, we investigate hybrid populations in which two different algorithms coexist, revealing algorithmic externalities and regimes where aggressive dynamics, improve players' satisfaction at the cost of efficiency or stability.
% \color{black}

%\color{black}
%\subsection*{Road map} 
\subsection{Paper outline} 

The rest of the paper is organized as follows. Section~\ref{sec:problem} formally introduces the Kelly mechanism and the induced game. Section~\ref{sec:main_results} derives the repeated Kelly game with $\alpha$-fair utilities for bandwidth allocation in wireless networks, presents the proposed bidding algorithms, and establishes their convergence guarantees. Section~\ref{sec:simulations} complements the theoretical results with numerical simulations. Section~\ref{sec:conclusion} concludes the paper.

%%%%%%%%%%%%%%%%%%%%%%%%%%%%%%%%%%%%%%%%%%%%%%%%%%%%%%%%%%%%%%%%%%%%%%%%%%%%%%%%%%%%%%%%%%%%%%%%%%%%
%%%%%%%%%%%%%%%%%%%%%%%%%%%%%%%%%%%%%%%%%%%%%%%%%%%%%%%%%%%%
\section{Problem Formulation}\label{sec:problem}
We consider a repeated allocation of a unit-sized divisible resource among~$n$ agents over~$T$ rounds according to the general allocation mechanism proposed in~\cite{basar_JSAC_PoA1}, which extends the Kelly mechanism introduced in~\cite{Kelly_Original97}.

\noindent \textbf{Bidding.} At each step $t$, each agent $i$ submits a bid $b_{i,t}$ that must be at least a fixed positive constant $\tilde \epsilon_i$, i.e., $b_{i,t}\geq \tilde \epsilon_i>0$. It must also respect the budget constraint, i.e., $b_{i,t}\leq  \tilde c_i$, where~$\tilde c_i$ is the budget of agent $i$ at each round $t$. This bid is based on previously submitted bids $\bm{b}_1, \ldots, \bm{b}_{t-1}$, where $\bm{b}_s= (b_{i,s})_{i\in \mathcal{I}}$ and $\mathcal{I}$ is the set of agents.    

\noindent \textbf{Allocation:} Based on the bids of each round $t$, the resource owner allocates fractions $x_{i,t}(\bm{b}_t)$ of the resource according to
   \begin{align}
    x_{i,t}(\bm{b}_t) =
    \begin{cases}
    \frac{w_i(b_{i,t})}{\sum_{j=1}^{n} w_j(b_{j,t}) + \delta}, & \text{if } w_i(b_{i,t}) > 0,\\
    0, & \text{otherwise},
    \end{cases}
    \end{align}
where $w_i: \mathbb{R}^{+} \to \mathbb{R}^{+}$ are continuous, increasing functions governing how resources are distributed, and $\delta \geq 0$ is a reservation parameter \cite{basar_JSAC_PoA1}. If $w_i$ is the identity function, this mechanism reduces to the classic Kelly mechanism. 

Each agent~$i$ has a valuation function~$V_i: [0,1] \to \mathbb{R}_{\geq 0}^{n}$, where~$V_{i}(x_i)$ quantifies monetary benefit of acquiring a fraction~$x_i$ of the resource. The utility of agent $i$ at each step $t$ is determined by the function $\varphi_i$, defined as the value the agent derives from the allocated fraction minus the payment, i.e., $\varphi_i(\bm{b}_t)=V_{i}(x_{i,t}(\bm{b}_t)) - b_{i,t}$. The objective of each agent is devise an online bidding strategy~$b_{i,1}, \ldots, b_{i,T}$ to maximize their aggregate utility, i.e., $\sum_{t=1}^{T} \varphi_{i}(\bm{b}_t)$.

Following~\cite{basar_JSAC_PoA1}, define the change of variable $z_{i,t} = w_i(b_{i,t})$, and the function $p_i: \mathbb{R}^{+} \to \mathbb{R}^{+}$ as the inverse of $w_i$, i.e., $p_i(z_i) \triangleq w_i^{-1}(z_i)$. We refer to $p_i$ as the payment function for agent $i$. Under this change of variable, the allocation rule and the utility function become,
\begin{align}\label{e:allocation_mechanism}
    &x_{i,t}(\bm{z}_t) =
    \begin{cases}
    \frac{z_{i,t}}{\sum_{j=1}^{n} z_{j,t} + \delta} & \text{if } z_{i,t} > 0, \\ 
    0 & \text{otherwise}.
    \end{cases} \\ \label{e:utility_auction_game}
   & \varphi_i(\bm{z}) \triangleq V_{i}(x_i(\bm{z})) - p_{i}(z_i).
\end{align}
Note that both formulations are equivalent in the sense that the allocated fraction corresponding to a given payment is the same in each setting. In this paper, we will focus on the second formulation using $\bm{z}$. 

% which can be viewed as the signal submitted by the agents, with the required payment given by the function $p_i: \mathbb{R}^{+} \to \mathbb{R}^{+}$ defined as
% \begin{align}\label{e:cost}
%     p_i(z_i) \triangleq w_i^{-1}(z_i),
% \end{align}
% where $w_i^{-1}$ is the inverse of $w_i$. 

% We study the budgeted Kelly mechanism, where each agent has a budget constraint $c_i$, i.e., $p_i(z_i)\leq c_i$.

When agents are aware that the Kelly mechanism governs resource allocation, the interaction between them forms a competitive repeated game. We define $\mathcal{G}$ as the stage game arising from this competition, where the set of players is $\mathcal{I}$ with utility functions $\boldsymbol{\varphi} = (\varphi_i)_{i \in \mathcal{I}}$ and action space constrained by budgets, denoted $\mathcal{R}$, and given by the cartesian product of $\mathcal{R}_i$ for $i\in \mathcal{I}$, where $\mathcal{R}_{i} \triangleq [\epsilon_i, c_i]$, where $\epsilon_i = p_i^{-1}(\tilde \epsilon_i)$ and $c_i = p_i^{-1}(\tilde c_i)$. 

% with utility functions $\boldsymbol{\varphi} = (\varphi_i)_{i \in [n]}$ and an action space constrained by budgets, denoted $\mathcal{R}=\mathcal{R}_1 \times \ldots \times \mathcal{R}_n$, %  
%\begin{align}\label{e:stateSpace}
%\mathcal{R} \triangleq \varprod_{i \in [n]} \mathcal{R}_i: \; \mathcal{R}_i \triangleq  [\tilde \epsilon_i_i, \tilde c_i],
% \{ z_i: \; \epsilon \leq p_i(z_i) \leq c_i \},
%\end{align}

We make the following assumptions about the functions~$V_i$ and~$p_i$. 
\begin{assumption}\label{assum:V_properties}
Over the domain $[0,1]$, $V_i$ is strictly increasing, concave, and twice continuously differentiable ($V_i \in \mathcal{C}^2([0,1])$). Over the domain $\mathbb{R}^n_{\geq 0}$, $p_i$ is convex, increasing with respect to $z_i$ for any $i$, and twice continuously differentiable, i.e., $\bm{p} \in \mathcal{C}^2(\mathbb{R}^n_{\geq 0})$. 
\end{assumption}

Note that the above assumption is standard~\cite{basar_JSAC_PoA1,johari_thesis2004}. Moreover, it is natural for $V_i$ and $p_i$ to be increasing. Agents gain larger utility from receiving a larger share of the resource.

Under Assumption~\ref{assum:V_properties}, the utility function~$\varphi_i$ is concave with respect to the $i$-th component and thus the \textit{best response operator} is a function, that we denote as $\textbf{BR}:~\mathbb{R}^{n} \mapsto \mathbb{R}^{n}$. %Note that the utility function of each agent in~\eqref{e:utility_auction_game} depends only on their own bid, (i.e., $\mathcal{G}$ is an aggregative game), and thus by abuse of notation, we can write $\varphi_i(z_i,z_{-i})= \varphi_i(z_i,s_{i}(\bm{z}))$ such that $s_i(\bm{z}) \triangleq \sum_{j\neq i} z_j+\delta$. 
Note that $\mathcal{G}$ is an aggregative game because the utility function of each agent in~\eqref{e:utility_auction_game} depends only on their own bid and on the sum of the bids of others. Thus, by abuse of notation, we can write $\varphi_i(z_i,z_{-i})= \varphi_i(z_i,s_{i}(\bm{z}))$ such that $s_i(\bm{z}) \triangleq \sum_{j\neq i} z_j+\delta$ and $(z_i, \bm{z}_{-i})$ denotes the vector where agent $i$ submits a bid $z_i$, while the other agents submit bids $\bm{z}_{-i}$. 

The best response of a player $i$, denoted $\text{BR}_i: \mathbb{R}_+ \mapsto \mathbb{R}_{+}$, is defined as, 
\begin{align}\label{eq:BR_def}
    \text{BR}_i(s)= \arg\max_{z_i\in\mathcal{R}_i}\varphi_i(z_i,s),
\end{align}
and it holds that $\textbf{BR}(\bm{z}) =(\text{BR}_i(s_{i}(\bm{z})))_{i\in [n]}$.
% Let $(z_i, \bm{z}_{-i})$ denote the vector where agent $i$ submits a bid $z_i$, while the other agents submit bids $\bm{z}_{-i}$.
\begin{definition}[Nash Equilibrium] A strategy profile $\bm{z}^* = (z_1^*, z_2^*, \dots, z_n^*) \in \mathcal{R}$ is a \emph{Nash Equilibrium (NE)} of $\mathcal{G}$ if and only if, for every player $i \in \mathcal{I}$,
\begin{align}
\varphi_i(z_i^*, \bm{z}_{-i}^*) \geq \varphi_i(z_i, \bm{z}_{-i}^*) \quad \forall z_i \in \mathcal{R}_i
\end{align}
\end{definition}

As a consequence of Assumption~\ref{assum:V_properties}, the function $\varphi_i$ is concave in its $i$-th component and twice continuously differentiable on $\mathbb{R}_{\geq 0}^n$, and the actions set $\mathcal{R}$ is non empty, closed, bounded, and convex. Existence of a Nash Equilibrium (NE) of the game $\mathcal{G}$ follows by~\cite[Thm. 1]{rosen}. 
\begin{theorem}
    The set of Nash equillibria of $\mathcal{G}$, denoted $NE(\mathcal{G})$ is non-empty, i.e., $NE(\mathcal{G})\neq \emptyset$. 
\end{theorem}

\medskip 

\noindent \textbf{Notation.} We use $ \dot{V}_i(\cdot)$ and  $\ddot{V}_i(\cdot)$ to denote the first and second derivatives of $V_i$, respectively. We use a similar notation for $p_i$. We use $\partial_{j} \varphi_i$ to designate the partial derivative of $\varphi_i$ with respect to the bid of agent $j$, and $\partial_{j,k}^{2} \varphi_i$ to designate the second order mixed derivatives of $\phi$ with respect to the bid of agents $j$ and $k$.

\section{Repeated Kelly Game}
\label{sec:main_results}

\subsection{Motivation: Bandwidth allocation in wireless networks}  
\label{ss:Example_log_utility}
\begin{figure}[H]
    \centering
    \includegraphics[width=\linewidth]{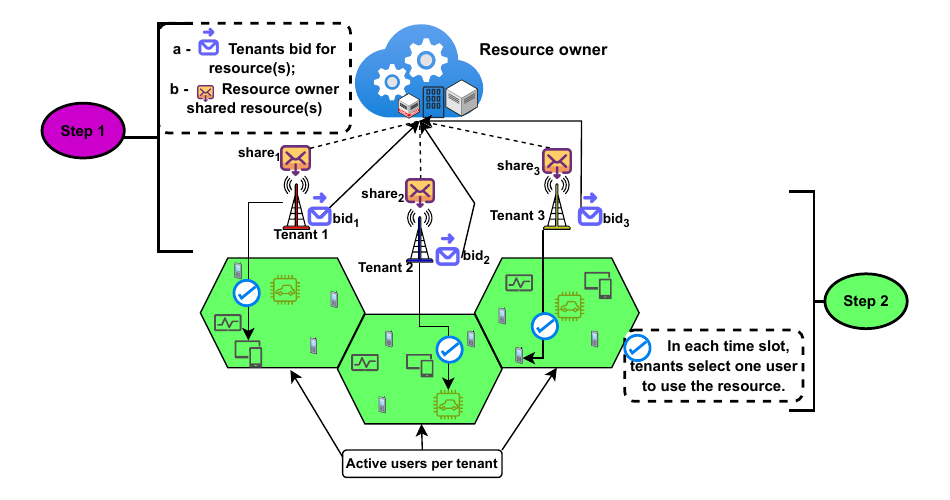}
    \caption{Bandwidth Allocation between tenants and users}
    \label{fig:slicing_drawio}
\end{figure}

In this section, we show how a repeated Kelly game with logarithmic valuations $V_i$ arises in bandwidth allocation among multiple \textit{tenants} (e.g., virtual operators or service providers), each serving its own set of users. 
Over rounds $t\in \{1,\ldots , T\}$, an infrastructure provider allocates a total bandwidth $B$ according to the Kelly mechanism; given bids $z_{j}(t)$, tenant $j$ receives bandwidth:  
   \begin{align}
            B_{j}(t) = \frac{z_{j}(t)}{\sum_{k\in \cI} z_{k}(t) + \delta}, 
     \end{align}
where $\cI$ denotes the set of tenants and $\delta\geq 0$. 

Within round $t$, time is divided into slots. At each slot, tenant $j$ schedules exactly one user from its set $\cI_{j}$. %; unscheduled users are silent. 
Let~$S_j^{\tau}(t)\in \cI_{j}$ denote the scheduled user at slot~$\tau$. When a user $i$ is scheduled, its  transmission rate, denoted $r_{j,i}^{\tau}(t)$, is proportional to the allocated bandwidth,
\begin{align}\label{eq:inst-rate}
        r_{j,i}^{\tau}(t) = \gamma_{j,i}^{\tau}(t) B_{j}(t) \mathds{1}\left(S_{j}^{\tau}(t) = i \right),
\end{align}
where $\gamma_{j,i}^{\tau}(t)>0$. For instance in~\cite{Mandar_Wiopt20}, $\gamma_{j,i}^{\tau}(t)=  \ln \left(1+\frac{p_{j,i}h_{j,i}^{\tau}(t)}{N_0}\right)$, where~$p_{j,i}$ is the transmission power, $h_{j,i}^{\tau}(t)$ is the channel state, and $N_0$ is the noise power\footnote{For the sake of simplicity, we consider a basic AWGN channel model and a single user scheduler; with due modifications, same game extends to more advanced channels and multi-user scheduling.}. 

A standard objective in this setting is the \textit{Proportional-fair} metric~\cite{kim2005proportional}. Optimizing this objective enables balancing the overall throughput and fairness across users. We use $\mathrm{PropFair}_j(t)$ to denote the proportional fair metric of tenant~$j$ at round~$t$, and it is expressed as, 
\begin{align} \label{eq:pf-utility}
  \mathrm{PropFair}_{j}(t) 
&\triangleq
   \sum_{i\in \cI_{j}}\ln\left(\sum_{\tau} r_{j,i}^{\tau}(t)  \right)  \\ 
&\hskip-10mm = N_{j} \ln(B_{j}(t)) +    \sum_{i\in \cI_{j}}\ln\left(\sum_{\tau} \gamma_{j,i}^{\tau}(t) \mathds{1}\left(S_{j}^{\tau}(t) = i \right)  \right),
\end{align}
where $N_j$ is the number of users served by tenant $j$, i.e., $N_{j}=|\cI_{j}|$. This decomposition makes the roles of bidding and scheduling transparent: the term $N_j \ln(B_j(t))$ depends only on the bidding process, while the second term is controlled by the scheduling policy and channel states, and is independent of the bids. 
Using the quasi-linear utility model~\cite{johari_thesis2004}, the %stage 
utility of tenant $j$ at round $t$ writes  
  \begin{align}\nonumber
        \varphi_{j}(z_j(t), z_{-j}(t)) = &   N_{j} \ln \left(\frac{z_j(t)B}{\sum_{k\in \cI_{j}} z_k(t)}\right)\\ \label{eq:log-util}
        & \quad +   \sum_{i\in \cI_{j}}\ln\left(\sum_{\tau} \gamma_{j,i}^{\tau}(t) \mathds{1}\left(S_{j}^{\tau}(t) = i \right)  \right) - z_j.
\end{align}
Therefore, the bidding interaction induced by the scheme just described is a repeated Kelly game with logarithmic $V_i$'s.

\subsection{Single-Agent Formulation of the Online Bidding Problem} 

At round~$t$ of the repeated Kelly game, agent~$i$ faces uncertainty about others’ aggregate bid~$s_{i}(\bm{z}(t))=\delta + \sum_{j\ne i} z_j(t)$. Before bidding, agent $i$ only knows the history $(s_{i}(\bm{z}(1)),\ldots,s_{-i}(\bm{z}(t-1)))$. A bidding algorithm~$\cA_i$ maps this history to a bid $z_i^{\cA_i}(t)\in\cR_i$, and earns the payoff $\varphi_i(z_i^{\cA_i}(t),s_i(\bm{z}(t)))$. Under Assumption\ref{assum:V_properties}, this yields an Online Convex Optimization (OCO)~\cite{hazan2016introduction} problem: in each of the~$T$ rounds, the agent chooses $z_i^{\cA_i}(t)\in \cR_i$, then a concave reward function $u_i^{t}:\cR_i \mapsto \mathbb{R}$, defined as $u_i^{t}(z_i) \triangleq \phi_i(z_i,s_{i}(\bm{z}(t)))$ is revealed, and the agent receives $u_i^{t}\big(z_i^{\cA_i}(t)\big)$. The objective is then to maximize the aggregate reward over rounds. In this framework, the main performance metric of an algorithm $\cA_i$ is the \textit{regret}, denoted as $\mathrm{Reg}_{T}^{(i)}(\mathcal{A}_i)$, and defined as the gap between the cumulative reward of the best fixed bid in hindsight and the agent’s cumulative reward, i.e., 
\begin{align}\label{eq:regret}
\mathrm{Reg}_T^{(i)}(\mathcal{A}_i) \triangleq \max_{z_i\in\mathcal{R}_i}\sum_{t=1}^{T} u_i^{t}(z_i) - \sum_{t=1}^{T} u_i^{t}\big(z_i^{\cA_i}(t)\big).
\end{align}
Define the constants $D_i$ and $G_i$ as upper bounds on the diameter of the decision set $\cR_i$, and the derivatives of $u_i^{t}$'s for any $t$, i.e., 
\begin{align}\label{eq:bounds}
    D_i\geq  c_i-\epsilon_i, \quad G_i \geq \sup_{\bm{z}\in \cR}  \left |  \partial_i \varphi_i(z_i, s_i(\bm{z})) \right|.
\end{align}
Because $\varphi_i$ is continuous over the closed set $\cR$, the constants $D_i$ and $G_i$ exists. Thus, standard OCO methods achieve sublinear regret, $\mathrm{Reg}_T^{(i)} = o(T)$. Consequently, for any sequence of opponent aggregates $s_{i}(\bm{z}(t))$, the agent's time-average reward approaches that of the best fixed bid appearing in~\eqref{eq:regret}.

\subsection{Bidding algorithms}
\label{ss:Bidding_algorithms}

We consider four bidding algorithms. Two of them are adaptations of classical no-regret methods to the repeated Kelly game, namely Online Gradient Descent (\OGD)~\cite{Zinkevich_OGD_03} and Dual Averaging (\DA)~\cite{McMahan_FTRL_JMLR_17}, an instance of the Follow-The-Regularized-Leader (FTRL) family of algorithms. The third algorithm is an instance of Regularized-Robbins–Monro (\RRM) family of algorithms, recently studied in the context of repeated games~\cite{Mertikopoulos_VI_2019_MP,Unified_SA_Games_MP_24}, and encompassing \DA as a special case. The fourth algorithm is a myopic best-response scheme. 

\OGD, \DA, and \RRM are first-order methods: they only require the derivative of the stage utility with respect to the agent’s bid. Specifically, at each step $t$, algorithm~$\cA_i\in\{\OGD,\DA,\RRM\}$ for agent~$i$ uses the gradient of the utility function evaluated at its current bid $z_i^{\cA_i}(t)$, namely,
\begin{align}
    g_t^{(i),\cA_i}\triangleq \partial_i u_i^t(z_i^{\cA_i}(t)) = \varphi_i(z_i^{\cA_i}(t), s_i(\bm{z}(t))).
\end{align}
The algorithm also employs a learning rate (or step-size) $\eta_t^{(i)}>0$ that is tuned at each step~$t$.

\medskip 

\noindent \textbf{Online Gradient Descent \OGD.} When $\cA_i=\OGD$, agent~$i$ updates their bid by moving along the gradient/derivative of the reward function~$u_i^t$ at the current bid, then projects back to the feasible set~$\cR_i$ (minimum bid and budget constraints). Formally, the update at step~$t+1$ is given by, 
\begin{align}\label{e:OGD_update}
z_i^{\texttt{OGD}}(t+1) = \Pi_{\mathcal{R}_i} \left( z_i^{\OGD}(t) + \eta_{t+1}^{(i)} g_t^{(i),\OGD}    \right),
\end{align}
where $\Pi_{\cR_i}$ is the euclidean projection over~$\cR_i$, which reduces to clipping, $\Pi_{\cR_i}(z) =~ \max(\min(z,c_i), \epsilon_i)$. Taking~$\eta_t^{(i)}=D_i/(G_i\sqrt{t})$, yields, $\mathrm{Reg}_T^{(i)}(\texttt{OGD}) \leq~\frac{3}{2} G_i D_i \sqrt{T}$~\cite{hazan2016introduction}[Thm. 3.1]. Similar regret guarantees hold when $\eta_t^{(i)}$ is constant over time; taking $\eta_t^{(i)} \equiv \eta^{(i)} = D_i/(G_i\sqrt{T})$ 
leads to, $\mathrm{Reg}_T^{(i)}(\texttt{OGD}) \le G_i D_i \sqrt{T}$~\cite{Orseau2025_OGDM}. If $u_i^{t}$'s are $\gamma_i$–strongly convex on~$\mathcal{R}_i$, then choosing a more aggressive learning rate~$\eta_t^{(i)}=1/(\gamma_i t)$ yields, $\mathrm{Reg}_T^{(i)}(\texttt{OGD}) \leq  \frac{G_i^2}{\gamma_i}\big(1+\log T\big)$.
% \begin{align}
% \mathrm{Reg}_T^{(i)}(\texttt{OGD}) \leq \frac{3}{2} G_i D_i \sqrt{T}. 
% \end{align}

\medskip 

\noindent \textbf{Dual Averaging (DA)}. This algorithm employs a \textit{regularizer}, i.e., a continuous strongly convex function, $h_i:\mathcal{R}_i\mapsto \mathbb{R}$. Let $g_{1:t}^{(i),\DA}$ designates the sum of the gradients up to time $t$, i.e., $g_{1:t}^{(i),\DA}= \sum_{s=1}^{t} g_s^{(i), \DA}$. \DA's update selects the bid $z$ that maximizes~$\left(z*g_{1:t}^{(i),\DA} - \frac{1}{\eta_t^{(i)}}h_{i}(z)\right)$. In particular, if $h_i(z)= z^{2} /2$, then the update, denoted \DAQ, is given by,  
\begin{align}\label{e:DA_update}
z_{i}^{\DAQ}(t+1) &= \Pi_{\cR_i} \left( \eta_{t+1}^{(i)} g_{1:t}^{(i), \DAQ} \right).
\end{align}
This update is also known as the lazy version of \texttt{OGD}, while~\eqref{e:OGD_update} is known as the agile one. Taking an adaptive learning rate~$\eta_t^{(i)}=D_i/(2G_i\sqrt{t})$ yields, $\mathrm{Reg}_T^{(i)}(\DAQ)\leq~\sqrt{2}G_iD_i \sqrt{T}$~\cite{McMahan_FTRL_JMLR_17}[Sec. 3.1]. If the number of rounds~$T$ is apriori known, then taking $\eta_t^{(i)}=D_i/(G_i\sqrt{T})$ yields, $\mathrm{Reg}_T^{(i)}(\DAQ)\leq G_iD_i \sqrt{T}$~\cite{McMahan_FTRL_JMLR_17}[Sec. 3.2].  

\begin{remark}\label{remark:OCO_constants_log}
In general, \OGD and \DAQ require an upper bound $G_i$ on the gradient of the agent's utility (see~\eqref{eq:bounds}), for tuning $\eta_t^{(i)}$. The constant $G_i$ depends on the budgets of the other players and may therefore be unknown in practice. In the case of logarithmic utilities, i.e., $V_i(\cdot) = a_i \ln(\cdot) + d_i $, and $p_i$ is the identity function, we can write,
\begin{align}\label{e:gradient_bound}
   \left| \partial_i \varphi_i(\bm{z}) \right| = \left|a_i \left( \frac{1}{z_i} - \frac{1}{\sum_{j} z_j + \delta} \right) - 1\right|  \leq \frac{a_i}{z_i} + 1 \leq \frac{a_i}{\epsilon_i} +1. 
\end{align}
for every $\bm{z}\in \cR$. Thus taking~$G_i = \frac{a_i}{\epsilon_i}+1$ provides a bound that is independent of the other agents' budgets, which simplifies the use of these bidding algorithms in practice. Indeed, the learning rate for each agent~$i$, using either $\OGD$ or $\DAQ$, can be tuned as $\eta_{t}^{(i)}=\cO\left(\frac{c_i \epsilon_i}{a_{i} \sqrt{T}}\right)$, leading to a regret $\cO\left(\frac{c_i\, a_i}{\epsilon_i}\sqrt{T}\right)$. 

\begin{comment}
where $\mathcal{S}_i$ is the set of possible values for $s_{-i}(t)$ for any $t$, i.e., 
\begin{align}
\mathcal{S}_i \triangleq \left [\delta + \sum_{j\neq i } \epsilon_j, \delta + \sum_{j\neq i} c_j  \right].
\end{align}
\end{comment}

\end{remark}

\medskip

\noindent \textbf{Regularized-Robbins Monro (\texttt{RRM}).} In the context of the repeated Kelly, an agent $i$ using \texttt{RRM} maintains a cumulative weighted sum of gradients at each step $t$, denoted $y_{i}^{\RRM}(t)$, which is then converted to the bid for that iteration, denoted as $z_i^{\RRM}(t)$. Initially, $y_{i}^{\RRM}(0)=0$. At any step $t\geq 1$, 
\begin{align}\label{eq:RRM_update}
\begin{cases}
 &y_i^{\RRM}(t) = y_i^{\RRM}(t-1) + \eta_t^{(i)} g_t^{(i),\RRM} , \\
 &z_i^{\RRM}(t) = Q_i(y_i^{\RRM}(t-1)): \\
 &Q_i(y) \triangleq \argmax_{z_i\in \mathcal{R}_i} \left( z_iy - h_i(z_i) \right), 
\end{cases}
\end{align}
\noindent where $h_i:\mathcal{R}_i\mapsto \mathbb{R}$ is a continuous and $K_i$-strongly convex function for some $K_i>0$. 
In particular, if $h_i(z)=\frac{z^2}{2\lambda_i}$, then the update, denoted $\texttt{RMQ}$, is given by, 
   \begin{align}\label{e:update_RMQ}
      &z_i^{\texttt{RMQ}}(t) = \Pi_{\cR_i}\left(  \lambda_i y_i^{\texttt{RMQ}}(t)\right), 
    \end{align}
In particular, if $\eta^{(i)}_t$ is constant over time and $\lambda_i=1$ for all agents, then \texttt{RMQ} and \texttt{DAQ} yield the same update. Thus \texttt{RMQ} in this case has sublinear regret guarantees. 

\medskip

\noindent \textbf{Best-response ($\BR$).} When $\cA_i= \BR$, at each round $t$, agent~$i$ selects the bid that maximizes their payoff function~$\varphi_i$, when the aggregate bid of the other agents is equal to its value in the previous round, $s_{i}(\bm{z}(t-1)) \triangleq \sum_{j\neq i} z_j$. Formally, 
\begin{align}\label{e:BR_update}
z_i^{\text{BR}}(t)=\text{BR}_i \left(s_{i}\left(\bm{z}^{\BR}(t-1)\right)\right).
\end{align}
While $\BR$ lacks in general the no-regret guarantees of~$\DAQ$ and~$\OGD$, it is simpler to implement: it uses the observed aggregate $s_{i}\left(\bm{z}^{\BR}(t-1)\right)$ and the agent’s own constraints, and requires no knowledge or estimation of other agents’ budgets. 

\begin{remark}\label{remark:BR_log}

When the $V_i$'s are logarithmic, i.e., $V_i(\cdot) = a_i \ln(\cdot) + d_i $ with $a_i>0$, and $p_i(z)=z$, straightforward calculations yield a closed-form expression for the best response operator, 
% In general, when each agent’s payoff $\varphi_i(\cdot,s_{-i})$ is concave in $z_i$, the best-response problem
% \eqref{eq:BR_def} admits a unique solution, which can be computed efficiently using standard convex optimization methods.
\begin{equation}\label{eq:BR_log}
    \text{BR}_i(s)
=\Pi_{\mathcal{R}_i}\left(\frac{-s+\sqrt{s^2 + 4a_i\, s}}{2}\right). 
\end{equation}
More generally, \cite{mboulou2025best} derives closed form expressions for the best-response operator for $V_i$'s of the $\alpha$-fair type with~$\alpha\in~\{ 0,1,2\}$.

% In the present setting, for logarithmic valuations $V_i(\cdot) = a_i \ln(\cdot) + d_i $, KKT conditions yield a closed-form expression,

% Closed-form best responses can also be derived for other commonly used concave utilities, such as linear valuations, i.e., $V_i(\cdot) = a_i (\cdot) + d_i $ or inverse utilities, i.e., $V_i(\cdot) =  -\frac{a_i}{(\cdot)} + d_i $ (see,~\cite{mboulou2025best}).
\end{remark}

% While $\BR$ lacks in general the no-regret guarantees of~$\DAQ$ and~$\OGD$, it is simpler to implement: it uses the observed aggregate $s_{-i}(t)$ and the agent’s own constraints, and requires no knowledge or estimation of other agents’ budgets. 

\medskip 

% We first focus on the case where the~$V_i$'s are logarithmic in the allocated resource. This model of utilities is motivated by the technological context in network slicing of Section~\ref{ss:Example_log_utility}. In the following section, we prove a series of properties of the stage game~$\mathcal{G}$, which will be then useful for studying the evolution of the system under the different bidding strategies considered. 

\subsection{Convergence guarantees}

% Establishing the uniqueness of the NE of the game and the convergence of DA dynamics to this point builds on a notion that we call \emph{Strong Diagonal Strict Concavity} (\SDSC). The rest of this section is organized as follows: Section~\ref{ss:SDSC_log} derives sufficient conditions for the \SDSC, Section~\ref{ss:Uniq_NE_log} shows the uniqueness of the NE, Section~\ref{ss:convg_DA_log} provides sufficient conditions on the heterogeneity of the agents and the weighting functions to guarantee the convergence of DA dynamics to the NE, Section~\ref{ss:convg_BR_log} presents a sufficient condition on the minimum bid for the convergence of the BR dynamics, Section~\ref{ss:extensions_general_V} extends the convergence results, for both  and BR, to more general utility functions, and Section~\ref{ss:comp_BR_DA} discusses the choice of DA or BR from a selfish agent's perspective. 

In this section, we focus utilities of the form, namely $V_i(\cdot)=a_iV(\cdot)+d_i$ with $a_i>0$. Our results hold in particular when $V(\cdot)=\ln(\cdot)$. This model is motivated by the network-slicing setting described in Section~\ref{ss:Example_log_utility}. We first provide a sufficient condition for the stage game~$\mathcal{G}$ to satisfy \textit{Strong Diagonal Strict Concavity} (\SDSC), which implies Rosen’s \textit{Diagonal Strict Concavity}~\cite{rosen}, or equivalently \textit{monotonicity}~\cite{Zhou_OGD_NE_OR_21}.
As a consequence, the Nash equilibrium is unique. This property also serves as a key ingredient to establish convergence to the equilibrium under \OGD and \DAQ dynamics. Finally, we prove convergence of $\BR$ via a contraction argument.

For a vector $\bm{r} \in \mathbb{R}_{> 0}^n $, \SDSC is defined in terms of the $ n \times n $ matrix $ \bm{H}_{\bm{r}}(\bm{z}) $, whose $(i,j)$-entry is given by,
\begin{align}
(\bm{H}_{\bm{r}}(\bm{z}))_{i,j} \triangleq r_i \, \partial^2_{i,j} \varphi_i(\bm{z}) + r_j \, \partial^2_{j,i} \varphi_j(\bm{z}),
\end{align}
where the partial derivative is taken with respect to the actions of agent $i$ and $j$, i.e., $z_i$ and $z_j$, respectively. 
\begin{definition}[Strong Diagonal Strict Concavity]\label{def:SDSC}
The game $\mathcal{G}$ satisfies \textit{Strong Diagonal Strict Concavity} in $\bm{r}$ if and only if the matrix $\bm{H}_{\bm{r}}(\bm{z})$ is negative definite for all $\bm{z} \in \mathcal{R}$, %; that is,
\begin{align}
%\forall \bm{z} \in \mathcal{R}, \quad 
\max_{\bm{v} : \| \bm{v} \| = 1} \left\{ \bm{v}^{\intercal} \bm{H}_{\bm{r}}(\bm{z}) \bm{v} \right\} < 0.
\end{align}
In this case, we write that $\mathcal{G}$ satisfies $\SDSC(\bm{r})$.
\end{definition}

\SDSC appears in Rosen’s paper~\cite{rosen} as a sufficient condition for \textit{diagonal strict concavity}, or equivalently for the game to be \textit{monotone}~\cite{Zhou_OGD_NE_OR_21}. This assumption is particularly useful for analyzing more general Kelly-type games with coupled action sets and for proving convergence of continuous-time dynamics. \SDSC is also one of the main conditions for the convergence of \RRM updates in a multi-agent setting to a NE~\cite{Mertikopoulos_VI_2019_MP,Unified_SA_Games_MP_24}.

We prove in Theorem~\ref{thm:SDSC_log} that the Kelly game satisfies \SDSC when utilities scale logarithmically with the allocated resource. First, we introduce the necessary notation. Define the functions $f_i$, $g_i$, and $\psi_{\bm{r},\bm{V}}$ as
\begin{align}\label{e:f_expression}
        &f_i(x) = (1 - x)^2 \ddot{V}_i(x) - 2(1 - x)\dot{V}_i(x) , \\ \label{e:g_expression}
        &g_i(x) = -  x(1 - x) \ddot{V}_i(x)+  (2x - 1) \dot{V}_i(x),  \\ \label{eq:sufCond} 
        & \psi_{\bm{r},\bm{V}}(\bm{x}) \triangleq \left( \sum_{i\in \cI}\frac{r_ig_i(x_i)^2}{k_i(x_i)} \right) \left(\sum_{i\in \cI}\frac{1}{r_ik_i(x_i)} \right),
\end{align}
where $k_i(x) =  g_i(x) - f_i(x) + \delta^2 L_i$ and $L_i= \min_{z_i \in \mathcal{R}_i} \ddot{p}_i(z_i)$. Further define the set $\Delta = \{ \bm{x}>\bm{0}: \;\sum_{i\in \cI} x_{i}\leq \frac{\sum_{k\in \cI} c_k}{\sum_{k\in \cI} c_k + \delta}  \}$.

\begin{theorem}\label{thm:SDSC_log}
    The following holds, 
    \begin{enumerate}
        \item If there exists a vector $\bm{r}> \bm{0}$ such that $\psi_{\bm{r},\bm{V}}(\bm{x})$ is strictly smaller than $1$, then $\cG$ is $\SDSC(\bm{r})$. Formally, 
        
        \begin{align}\label{e:suff_condition_SDSC}
               \exists \bm{r} > \bm{0} : \; \sup_{\bm{x}\in \Delta} \psi_{\bm{r},\bm{V}}(\bm{x}) < 1 \implies \mathcal{G} \text{ is } \SDSC(\bm{r}).
        \end{align}
        % then $\mathcal{G}$ is $\SDSC(\bm{r})$. 

        \item If $V_i(\cdot)= a_i \ln (\cdot) + d_i$, then the condition~\eqref{e:suff_condition_SDSC} is satisfied for $r_i = 1/a_i$.
    \end{enumerate}
\end{theorem}
The proof of Theorem~\ref{thm:SDSC_log} is presented in the supplementary material. In general, a negative definiteness numerical test for an~$n\times n$ matrix  requires~$\cO(n^{3})$ time and~$\cO(n^{2})$ memory. Theorem~\ref{thm:SDSC_log} exploits the structure of the matrix~$\bm{H}_{\bm{r}}(\bm{z})$ in the Kelly game to significantly reduce the verification, for a fixed~$\bm{z}$, to~$\cO(n)$ time and~$\cO(1)$ memory via the condition~\eqref{e:suff_condition_SDSC}. Moreover, this reduction turns \SDSC test into bounding the maximum of a closed-form function, which is easier to deal with analytically; in particular, it enables the proof of \SDSC when the $V_i$'s are logarithmic. 

\begin{comment}
\begin{IEEEproof}[Sketch of the proof]
The proof is presented in Section~\ref{sec:SDSC}. 
Lemma~\ref{lem:H_computation} simplifies the quadratic form of the matrix $\bm{H}_{\bm{r}}(\bm{z})$, Lemma~\ref{thm:SDSC} leverages this simplification to show that \eqref{e:suff_condition_SDSC} is a sufficient condition for $\mathcal{G}$ to be \SDSC$(\bm{r})$, and Lemma~\ref{lem:logSDSC} shows that the aforementioned condition holds for logarithmic $V_i$'s for $r_i=1/a_i$. 

    % Section~\ref{sec:SDSC} develops a more general analysis for arbitrary~$V_i$ in the game $\mathcal{G}$. Specifically, it presents 3 lemmas: Lemma~\ref{lem:H_computation} simplifies the quadratic form of the matrix $\bm{H}_{\bm{r}}(\bm{z})$, Lemma~\ref{thm:SDSC} leverages this simplification to provide a sufficient condition for $\mathcal{G}$ to be \SDSC$(\bm{r})$ under general $V_i$'s, and 
    % and provides a sufficient condition under which~$\mathcal{G}$ is~\SDSC(1); the logarithmic case is recovered as a particular case, yielding~
\end{IEEEproof}
\end{comment}

% \subsection{Uniqueness of Nash equilibrium} 
% \label{ss:Uniq_NE_log} 

% The condition~\eqref{e:suff_condition_SDSC} is also sufficient for t
% Theorem~\ref{thm:SDSC_log} with 

\begin{corollary}\label{thm:unique_NE}

The condition~\eqref{e:suff_condition_SDSC} is sufficient for the uniqueness of the Nash equilibrium of $\cG$. We denote this unique equilibrium as $\bm{z}^{*}$.   
\end{corollary}

Corollary~\ref{thm:unique_NE} follows directly from Theorem~\ref{thm:SDSC_log} using~\cite{rosen}. 
% \begin{IEEEproof}
%   Owing to Theorem~\ref{thm:SDSC_log}, there exists~$\bm{r}$ such that~$\mathcal{G}$ satisfies~$\textbf{SDSC}(\bm{r})$. Thus~$\mathcal{G}$ is \textit{diagonally strictly concave} with respect to $\bm{r}$~\cite[Thm. 6]{rosen}, and this guarantees a unique Nash equilibrium~\cite[Thm. 2]{rosen}.
% \end{IEEEproof} 
Uniqueness of the Nash equilibrium has been established under various conditions: when the $V_i$'s satisfy Assumption~\ref{assum:V_properties} with identity payment function and no budget constraints~\cite[Thm.~2.2]{johari_thesis2004}, or more generally when the $p_i$'s satisfy Assumption~\ref{assum:V_properties} but are identical across agents~\cite[Prop.~2]{basar_JSAC_PoA1}. These results do not account for constraints. Under budget constraints, uniqueness was shown in~\cite[Thm.~1]{NetGCoop16} for a common linear weighting function, while~\cite{R_Ma_Kelly_price_differentiation} allows heterogeneous linear coefficients but no budget constraints. By contrast, Theorem~\ref{thm:unique_NE} and Corollary~\ref{thm:unique_NE} establish uniqueness for logarithmic $V_i$, admits any $p_i$ satisfying Assumption~\ref{assum:V_properties} (possibly heterogeneous), and incorporates budget constraints, thereby unifying and extending the above results for logarithmic~$V_i$. 

Leveraging the \SDSC property of the game~$\mathcal{G}$, Theorem~\ref{thm:OGD_convergence_log} establishes the convergence of \OGD to the unique NE when the utilities scale logarithmically in the allocated resource.

\begin{theorem}[Convergence of \texttt{OGD}]\label{thm:OGD_convergence_log}
Assume that the condition~\eqref{e:suff_condition_SDSC} holds and that the $V_i$'s are of the form $V_i(\cdot)= a_i V (\cdot) + d_i$, with $a_i>0$ and $V$ satisfying Assumption~\ref{assum:V_properties}. If each agent updates their bid using \texttt{OGD}, i.e., $\cA_i = \texttt{OGD}$, $\forall i \in \cI$, with $\eta_t^{(i)} = \alpha_i \eta_t^{(0)}$, $\alpha_i>0$, $\sum_{t=1}^{\infty} \eta_t^{(0)} = \infty $, and $\sum_{t=1}^{\infty} \left(\eta_t^{(0)}\right)^{2} < \infty$, then the sequence of play~$\bm{z}^{\texttt{OGD}}(t)=(z_i^{\OGD}(t))_{i\in \cI}$ converges to the unique NE of the stage game, i.e., $\lim_{t\to \infty}\bm{z}^{\texttt{OGD}}(t) = \bm{z}^{*}$.
\end{theorem}

\begin{IEEEproof}
Let $\bm{r}^{*}$ be the vector for which the condition~\eqref{e:suff_condition_SDSC} is satisfied. 

   If $\eta_t^{(i)} = \eta^{(0)}_t$, then by Theorem~\ref{thm:SDSC_log} the game~$\mathcal{G}$ is $\bm{r}^{*}$-\textit{monotone}, so the convergence of \texttt{OGD} follows directly from~\cite{Zhou_OGD_NE_OR_21}[Thm.~4]. 
   
   To extend this result to heterogeneous learning rates $\eta_t^{(i)} = \alpha_i \eta_t^{(0)}$, consider the auxiliary game $\tilde{\mathcal{G}}$ with modified utilities $\tilde{\phi}_i = \alpha_i \phi_i$. The following holds, 

   \begin{itemize}
       \item The \texttt{OGD} updates in the game~$\tilde{\mathcal{G}}$ with $\tilde \eta_t^{(i)}= \eta_t^{(0)}$ coincide with the \texttt{OGD} updates in the game~$\mathcal{G}$ with~$\eta_t^{(i)} = \alpha_i \eta_t^{(0)}$. 

       \item The games~$\tilde{\mathcal{G}}$ and~$\mathcal{G}$ share the same set of Nash equilibria.  

       \item If $\mathcal{G}$ is $\bm{r}^{*}$-monotone, then the game $\tilde{\mathcal{G}}$ is $\bm{r}$-monotone with $r_i = r_{i}^{*}/ \alpha_i$. 
   \end{itemize}

   Combining the above statements with Theorem~\ref{thm:SDSC_log}, we deduce that $\tilde{\mathcal{G}}$ is $\bm{r}$-monotone with $r_i = r_{i}^{*}/ \alpha_i$. Thus, applying~\cite{Zhou_OGD_NE_OR_21}[Thm.~4] to~$\tilde{\mathcal{G}}$ yields the desired convergence result, which completes the proof.
   % Since $\partial_i \tilde{\phi}_i = \alpha_i \partial_i \phi_i$, the \texttt{OGD} updates in the repeated game~$\tilde{\mathcal{G}}$ with homogeneous step size $\eta_t^{(0)}$ coincide with the \texttt{OGD} updates in the repeated game~$\mathcal{G}$ with heterogeneous step sizes $\eta_t^{(i)} = \alpha_i \eta_t^{(0)}$. Moreover, Theorem~\ref{thm:SDSC_log} implies that $\tilde{\mathcal{G}}$ is $\bm{r}$-\textit{monotone} with $r_i = r_{i}^{*}/ \alpha_i$. 
\end{IEEEproof}

The step-size condition in Theorem~\ref{thm:OGD_convergence_log} is satisfied whenever~$\eta_t \propto t^{-\beta}$ with~$\beta \in (1/2,1]$. For logarithmic utilities, the induced optimization problem is convex, and the standard choice is $\eta_t \propto t^{-1/2}$, which yields the usual $O(\sqrt{T})$ regret but does not satisfy this condition. A simple workaround in the merely convex case is to take $\beta$ close to $1/2$, which preserves sublinear regret while complying with $\beta>1/2$. Moreover, since each agent’s action set is a compact interval $[\epsilon_i,c_i]$ with $\epsilon_i>0$, the logarithmic utilities are in fact strongly concave on this domain, with parameter $\gamma_i>0$ that goes to $0$ when $\epsilon_i\to 0$. In this strongly concave setting, one may take $\eta_t \propto 1/t$ (corresponding to $\beta=1$), which both satisfies the corollary’s condition and yields the standard logarithmic regret guarantees.

Now we turn our attention to the analysis of \DAQ and \texttt{RMQ}. Theorem~\ref{cor:convergence_DAQ} shows that \texttt{RMQ} converges to the NE of the stage game. Moreover, when the utilities scale logarithmically in the allocated resource, the theorem quantifies the convergence gap of \DAQ to the NE, via the metric~$\overline{\mathrm{Gap}}_T(\mathcal{A})$, used in~\cite{Mertikopoulos_VI_2019_MP}, and defined as, 
\begin{align}
&\overline{\mathrm{Gap}}_T(\mathcal{A})
= \frac{1}{T} \sum_{t=1}^{T} \mathrm{Gap}\left(\bm{z}^{\mathcal{A}}(t)\right):  \\ 
&\mathrm{Gap}(\bm{z})
= \sum_{i\in \cI} \frac{ \partial_i \varphi_i(\bm{z})}{a_i} \left(z^{*}_i - z_i\right).  \label{e:gap_metric}
\end{align}
Indeed, when the $V_i$'s are logarithmic, the game~$\mathcal{G}$ is~$\SDSC(\bm{r}^{*})$ with $r_i^{*}=1/a_i$, as shown in Theorem~\ref{thm:SDSC_log}. Thus, $\mathrm{Gap}(\bm{z}) > 0$ for all $\bm{z} \neq \bm{z}^{*}$, with equality if and only if $\bm{z} = \bm{z}^{*}$~\cite{Mertikopoulos_VI_2019_MP}. 

\begin{comment}
    \textcolor{blue}{
Cleque:  I believe the correct uniform gradient bound is $G_i=\frac{a_i}{\varepsilon_i}+1$ (linear payment) or $G_i=\frac{a_i}{\varepsilon_i}+\dot p_i(c_i)$ (in general) . The inequality used in~\eqref{e:gradient_bound} is not preserved under absolute values, since the gradient may be negative.
This may affect the resulting $\mathrm{Gap}$ bound, eq 33, by explicitly incorporating heterogeneity through the $a_i$'s; in particular, if $\min_i a_i$ is small, the gap bound can increase.
}
\end{comment}

\begin{theorem}[Convergence of \texttt{RMQ} and \texttt{DAQ}] \label{cor:convergence_DAQ} Assume that the condition~\eqref{e:suff_condition_SDSC} holds and that the $V_i$'s are of the form $V_i(\cdot)= a_i V (\cdot) + d_i$, with $a_i>0$ and $V$ is function that satisfies Assumption~\ref{assum:V_properties}. Further assume that~$\epsilon_i=\epsilon$ for all agents. The following holds, 
\begin{enumerate}
    \item If all agents update their bids according to \texttt{RMQ}, i.e., $\forall i$, $\cA_i=\texttt{RMQ}$, $\eta_t^{(i)} =  \alpha_i \eta_t^{(0)}$, $\alpha_i>0$, and $\sum_{t=1}^{\infty} (\eta_t^{(0)})^{2}/ \sum_{t=1}^{\infty} \eta_t^{(0)}\to 0$, then the vector of bids $\bm{z}^{\texttt{RMQ}}(t)$ converges to the unique NE of the stage game, i.e., $\lim_{t\to \infty} \bm{z}^{\texttt{RMQ}}(t)= \bm{z}^{*}$.  

    \item If all agents update their bids according to \texttt{DAQ} and $V(\cdot) = \ln(\cdot)$, with $\eta_t^{(i)}=\frac{\epsilon c_i}{a_i\sqrt{T}}$, $p_i$ is the identity function, and $a_i\geq \epsilon$, then 
        \begin{align}
                \overline{\mathrm{Gap}}_T(\DAQ) \leq \frac{1}{2\epsilon\sqrt{T}}
                 \left( \sum_{i\in \cI} c_i + 4|\cI| \, \max_{i\in \cI} c_i    \right).
        \end{align}
\end{enumerate}
 
\end{theorem}

\begin{IEEEproof}
Let $\bm{r}^{*}$ be the vector for which the condition~\eqref{e:suff_condition_SDSC}. 

We first prove the convergence of \texttt{RMQ} in~\eqref{e:update_RMQ}. 

When for all agents, $\eta_t^{(i)}= \eta_t^{(0)}$, and $\sum_{t=1}^{\infty} (\eta_t^{(0)})^{2}/ \sum_{t=1}^{\infty} \eta_t^{(0)}\to 0$, convergence of the bids vector, when all agents use $\RRM$---with $\texttt{RMQ}$ as particular case---is guaranteed by~\cite{Mertikopoulos_VI_2019_MP}[Thm. 4.6] under two conditions: 1) The game $\mathcal{G}$ is $\SDSC(\bm{1})$, and 2)  For any player $i$ and for any sequence $\{y_n\}\subset \mathbb{R}$ such that $Q_i(y_n)\to \ell$,  $F_i(\ell,y_n)\to 0$, where $F_i$ is what is called the \textit{Fenchel conjugate} in~\cite{Mertikopoulos_VI_2019_MP}, and defined as, $F_i(\ell, y)=   h^{*}_{i}(y) - \tilde h_i(y,\ell)$, where $h^{*}_{i}(y)= \max_{z\in \mathcal{R}_i} \left\{ \tilde h_i(y,z)\right\}$, and $h_{i}^{*}$ is the convex conjugate of $h_i$. It is easy to prove that the second condition holds for quadratic regularizers, i.e., any $h$ such that $h(z)=\frac{z^{2}}{2\lambda}$, and $\lambda>0$~\cite{Wiopt25_BMDA}. Thus, when $\bm{r}^{*}=\bm{1}$ convergence follows directly from \cite{Mertikopoulos_VI_2019_MP}[Thm. 4.6]. 

To extend this result to arbitrary $\bm{r}^{*}>\bm{0}$ and when $\eta_t^{(i)}= \alpha_i \eta_t^{(0)}$, we define the payoff functions $\tilde \phi_i = r_i^{*} \phi_i$, and the corresponding stage game $\tilde{\cG}$. The following holds, 
\begin{itemize}
    \item The games~$\tilde{\mathcal{G}}$ and~$\mathcal{G}$ share the same set of Nash equilibria.  

    \item If the game $\cG$ is $\SDSC(\bm{r}^{*})$, then the game $\tilde{\cG}$ is $\SDSC(\bm{1})$. 

    \item  \texttt{RMQ} updates with $\tilde \eta_t^{(i)}=\eta_t^{(0)}$ and regularizer $\tilde h_i(z)= \frac{z^{2}}{2\tilde{\lambda}_i}$, with $\tilde{\lambda}_i= \frac{\lambda_i \alpha_i}{r_i^{*}}$ in the repeated $\tilde{\cG}$, coincides with \texttt{RMQ} updates in the repeated $\cG$ with $\eta_t^{(i)}=\alpha_i \eta_t^{(0)}$, and $h_i(z)=\frac{z^{2}}{2\lambda_i}$.  
\end{itemize}

The statements above combined with Theorem~\ref{thm:SDSC_log} yields the convergence of the \texttt{RMQ} updates with homogeneous $\tilde \eta_t^{(i)}$ regularizers $\tilde h_i$ in the repeated $\tilde{\cG}$. This implies the convergence of \texttt{RMQ} updates in $\cG$ with heterogeneous $\eta_t^{(i)}$'s.

\medskip 

We now prove the gap bound for \texttt{DAQ} when $V(\cdot)= \ln(\cdot)$, when~$\eta_t^{(i)}=\frac{\epsilon c_i}{a_i\sqrt{T}}$. By Theorem~\ref{thm:SDSC_log}, $r_i^{*}$ is equal to~$1/a_i$ for logarithmic $V$. Moreover, because $\eta_t^{(i)}$ is constant over time, \texttt{RMQ}, with $\lambda_i=1$, and \DAQ updates coincide. Similarly to the convergence proof of \texttt{RMQ}, we consider the proxy \texttt{RMQ} updates in $\tilde{\cG}$ with $\tilde \eta_t^{(i)}=\eta_t^{(0)}=\frac{1}{\sqrt{T}}$, for any agent $i$, $\tilde h_i(z) = \frac{z^{2}}{2 \tilde \lambda_i}$, $\tilde \lambda_i = \alpha_i a_i$, and $\alpha_i= \frac{\epsilon \,c_i}{a_i}$. Applying~\cite[Thm.~6.2]{Mertikopoulos_VI_2019_MP} and~\cite[Cor.~6.3]{Mertikopoulos_VI_2019_MP} to these \texttt{RMQ} updates in $\tilde{\cG}$ yields,        

\begin{align}\label{e:gap_expression}
    \overline{\mathrm{Gap}}_{T}(\texttt{DAQ}) 
    \leq \frac{1}{\sqrt{T}} \left(  \tilde \Omega +\frac{\tilde G^{2}}{2\tilde K}  \right),
\end{align}
where
\begin{align}\label{e:constants_gap}
    &\tilde \Omega\triangleq \max_{\bm{z}\in \mathcal{R}} \sum_{i} \tilde h_i(z_i) - \min_{\bm{z}\in \mathcal{R}} \sum_{i} \tilde h_i(z_i), \\
     &\tilde K\triangleq \min_{i\in \cI} \frac{1}{\tilde \lambda_i},
    \quad
    \tilde G \triangleq \sup_{\bm{z}\in \mathcal{R}} \left\| \left(\partial_i \tilde \varphi_i(\bm{z})\right)_{i\in \cI} \right\|_2^{2}.
\end{align}

\noindent We bound these quantities as follows, 
\begin{align}
  &  \tilde G^{2}\leq \sum_{i\in \cI} \frac{1}{a_i^{2}} G_i^{2} \leq \sum_{i\in \cI}   \frac{4a_i^{2}}{a_i^{2}\epsilon^{2}}= \frac{4|\cI|}{\epsilon^{2}},  \\
  &\frac{1}{\tilde K}=\max_{i} \tilde \lambda_i = \epsilon \max_{i} c_i, \qquad \text{and } \tilde \Omega\leq \frac{1}{2\epsilon} \sum_{i\in \cI} c_i^{2}.
\end{align}
Plugging these bounds in~\eqref{e:gap_expression} yields the target result, which finishes the proof.\end{IEEEproof}

While \texttt{DAQ} comes with standard no-regret guarantees and can therefore be viewed as a plausible behavioral model for repeated bidding, \texttt{RMQ} with adaptive step-sizes $(\eta_t^{(i)})_{t}$ does not generally enjoy regret guarantees. Nevertheless, \texttt{RMQ} acts as proxy to analyze the multi-agent behavior of \texttt{DAQ} when the step-sizes $\eta_t^{(i)}$ are time-invariant; in addition, \texttt{RMQ} can be interpreted as a distributed procedure for computing the unique Nash equilibrium of the stage game, as shown in Theorem~\ref{cor:convergence_DAQ}.

The choice of~$\eta_t^{(i)}$ optimizes the asymptotic dependency of the regret in the parameters problem when the budgets of other agents is unknown (see Remark~\ref{remark:OCO_constants_log}). Under this choice of~$\eta_t^{(i)}$, Theorem~\ref{cor:convergence_DAQ} provides an explicit finite-horizon bound of the deviation from equilibrium through the averaged gap
$\overline{\mathrm{Gap}}_T(\texttt{DAQ})$.  The established bound highlights that convergence may deteriorate when the minimum admissible bid $\epsilon$ is small, since logarithmic utilities induce large gradients near~$0$. Finally, budget constraints introduce additional variability in the updates, which also contributes to slower convergence.

\medskip

Now we study the case where all agents employ a myopic best response. This is a \textit{simultaneous} best-response update: all agents revise in parallel from the last observed profile. Such parallel BR does not necessarily converge to a Nash equilibrium in general~\cite{hart2003uncoupled}. By contrast, in \textit{unilateral} best-response updates, agents revise one at a time (cyclically or at random); in finite potential games, these dynamics converge to a pure Nash equilibrium~\cite{monderer1996potential}.

\begin{theorem}\label{thm:BR_converges_NE}
    If $V_i(\cdot)= a_i \ln (\cdot) + d_i$, $p_i(z) = z $, $\cA_{i}=\BR$ for all agents, and $\epsilon_i=\epsilon$ such that, 
    \begin{align}\label{e:Condition_CV_BR}
           \epsilon>\frac{1}{n-1}\left(\frac{(\sqrt n-1)^2}{\sqrt n}\,\max_{i\in \cI} a_i -\delta\right).
    \end{align} 
then $(\bm{z}^{\BR}(t))_t$ converges to the unique Nash equilibrium $\bm{z}^{*}$ linearly fast, i.e., $\exists \rho \in (0,1)$: $\| \bm{z}^{\BR}(t) - \bm{z}^{*}\| \leq \rho^{t} \| \bm{z}^{\BR}(0) - \bm{z}^{*}\|$.
\end{theorem}

\begin{IEEEproof} 
The best-response updates are fixed point iterations with the best-response operator $\bm{BR}$. We prove that this operator is a \textit{contraction}, which yields convergence to the unique fixed point---which is also the NE of the stage game---at linear speed.

To prove that $\textbf{BR}$ is a contraction, define $\tilde{\text{BR}}_i(s) = \left(-s+\sqrt{s^2+4a_i s}\right)/2$ and $\widetilde{\mathbf{BR}}(\bm z)\triangleq\left(\widetilde{\text{BR}}_i(s_i(\bm{z}))\right)_{i\in\mathcal I}$, so that $\mathbf{BR}(\bm{z}) = \Pi_{\mathcal R}\left(\widetilde{\mathbf{BR}}(\bm{z}) \right)$ (see~\eqref{eq:BR_log}). Given that the projection map $\Pi_{\mathcal R}$ is $1$-Lipschitz, $\tilde{\text{BR}}_i$ is smooth, and using the generalized mean-value theorem~\cite[Cor.~3.2]{coleman}, a sufficient condition for the best-response operator $\textbf{BR}$ to be a contraction, is given by, 
\begin{align}
\sup_{\bm{z}\in \cR} \left\|\mathcal J_{\widetilde{\textbf{BR}}}(\bm z)\right\|_\infty<1. 
\end{align}
Direct calculations yield, 
\begin{align}\label{e:jacobi_small_1}
\left(\mathcal J_{\widetilde{\mathbf{BR}}}(\bm z)\right)_{i,j}
=
\begin{cases}
\zeta_i(s_i(\bm{z})),
& j\neq i,\\
0, & j=i.
\end{cases}
\end{align}
where $\zeta_i(s) = -\frac{1}{2} +\frac{s+2a_i}{2\sqrt{s^2+4a_i s}}$, and $s_i(\bm{z})=\sum_{j\neq i} z_j + \delta$. The function $\zeta_i$ is decreasing over $\mathbb{R}^{+}$, and thus,
\begin{align}
\left\|\mathcal J_{\widetilde{\mathbf{BR}}}(\bm z)\right\|_\infty
&=\max_{i\in\cI}\sum_{j\neq i}\zeta_i\!\left(s_i(\bm z)\right)
=\max_{i\in\cI}(n-1) \zeta_i \left(s_{\min}\right).
\end{align}
where $s_{\min}\triangleq (n-1)\epsilon + \delta$. Combining this with the fact that $(n-1)\zeta_i(s)<1$ is satisfied whenever, $s>\frac{(\sqrt n-1)^2}{\sqrt n}\, a_i$ (see~\cite{mboulou2025best}), proves that~\eqref{e:Condition_CV_BR} is indeed a sufficient condition for $\textbf{BR}$ to be a contraction. This finishes the proof. \end{IEEEproof}

\color{black}
\begin{comment}
\begin{IEEEproof}[Sketch of the proof]
We cast the dynamics as a fixed-point iteration governed by the best-response operator: at each step, each agent best-responds to the opponents’ previous actions. We derive a closed-form expression for this operator and its Jacobian, and we show that the Jacobian’s norm-one operator is strictly less than one as long as the minimum bid satisfies the condition~\eqref{e:Condition_CV_BR}. Consequently, the map is a contraction and, by Banach’s fixed-point theorem, the iterates exhibit global linear convergence to a unique fixed point, which also corresponds to the equilibrium of the game. 
The detailed proof will be presented in the appendix. 
\end{IEEEproof}
\end{comment}

Note that the lower bound on the minimum bid in~\eqref{e:Condition_CV_BR} scales inversely with the number of agents, i.e., $\epsilon = \mathcal{O}(1/n)$, and thus the convergence of best response dynamics is guaranteed with arbitrarily small minimum bid $\epsilon$ for large number of agents. Moreover, only~$\mathcal{O}(\ln(1/p))$ iterations are needed to converge to a point whose distance from the NE is smaller than~$p$.

% Finally, Theorem~\ref{thm:BR_converges_NE} shows that 

% Under the minimum-bid condition in~\eqref{e:Condition_CV_BR}, synchronous best-response dynamics converge to the Nash equilibrium. However, this theorem requires linear weighting functions (in contrast to \texttt{DAQ}). 

%%%%%%%%%%%%%%%%%%%%%%%%%%%%%%%%%%%%%%%%%%%%%%%%%%%%%%%%%%%%%%%%%%%%%%%%%%%%%%%%%%%%%%%%%%%%%%%%%%%%
%%%%%%%%%%%%%%%%%%%%%%%%%%%%%%%%%%%%%%%%%%%%%%%%%%%%%%%%%%%%

\section{Numerical Simulations}
\label{sec:simulations}

\begin{table}
\centering
\begin{minipage}{0.48\textwidth}
\centering
\captionof{table}{Homogeneous dynamics: number of convergence iterations in terms of the 
fixed-point residual $r_t$ (threshold $<10^{-5}$) under varying $\gamma$ and $n$.
}
\label{tab:convergence_residual}
\setlength{\tabcolsep}{4pt}
\begin{tabular}{|c|c|c|c|c|c|c|c|}
\hline
\textbf{$\gamma$} & \textbf{$n$} 
& \textbf{\texttt{BR}} 
& \textbf{\texttt{OGD}$_V$} 
& \textbf{\texttt{OGD}$_F$} 
& \textbf{\texttt{DAQ}$_F$} 
& \textbf{\texttt{DAQ}$_V$} 
& \textbf{\texttt{RRM}$_V$} \\
\hline
\multirow{3}{*}{\textbf{0}}
& 2   & 15 & 37 & 122 & 195 & $(r_T=3.7{\times}10^{-2})$ & 1682  \\
& 10  & 7  & 19 & 240 & 311 & $(r_T=1.05)$              & 2291 \\
& 20  & 6  & 20 & 252 & 330 & $(r_T=1.73)$              & 2361 \\
\hline

\multirow{3}{*}{\textbf{5}}
& 2   & 15 & 111 & 127 & 184 & $(r_T=1.0{\times}10^{-1})$ & 504  \\
& 10  & 7  & 40 & 206 & 274 & $(r_T=6.8{\times}10^{-1})$ & 2182 \\
& 20  & 6  & 1811 & 221 & 289 & $(r_T=7.5{\times}10^{-1})$ & 2220 \\
\hline

\multirow{3}{*}{\textbf{10}}
& 2  & 15 & 116  & 117 & 184 & $(r_T=9.7{\times}10^{-2})$ & 498   \\
& 10 & 8  & 533 & 185 & 253 & $(r_T=4.5{\times}10^{-1})$ & 1604 \\
& 20 & 8  & 533 & 194 & 259 & $(r_T=4.5{\times}10^{-1})$ & 2062 \\
\hline
\end{tabular}
\end{minipage}
\end{table}

\begin{figure}[t]
    \centering
    % Legend (optional separate graphic)
    \includegraphics[width=\linewidth]{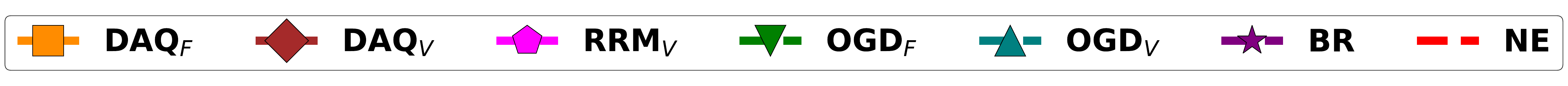}

    % First row of subfigures
    \begin{subfigure}{0.32\linewidth}
        \centering
        \includegraphics[width=\textwidth]{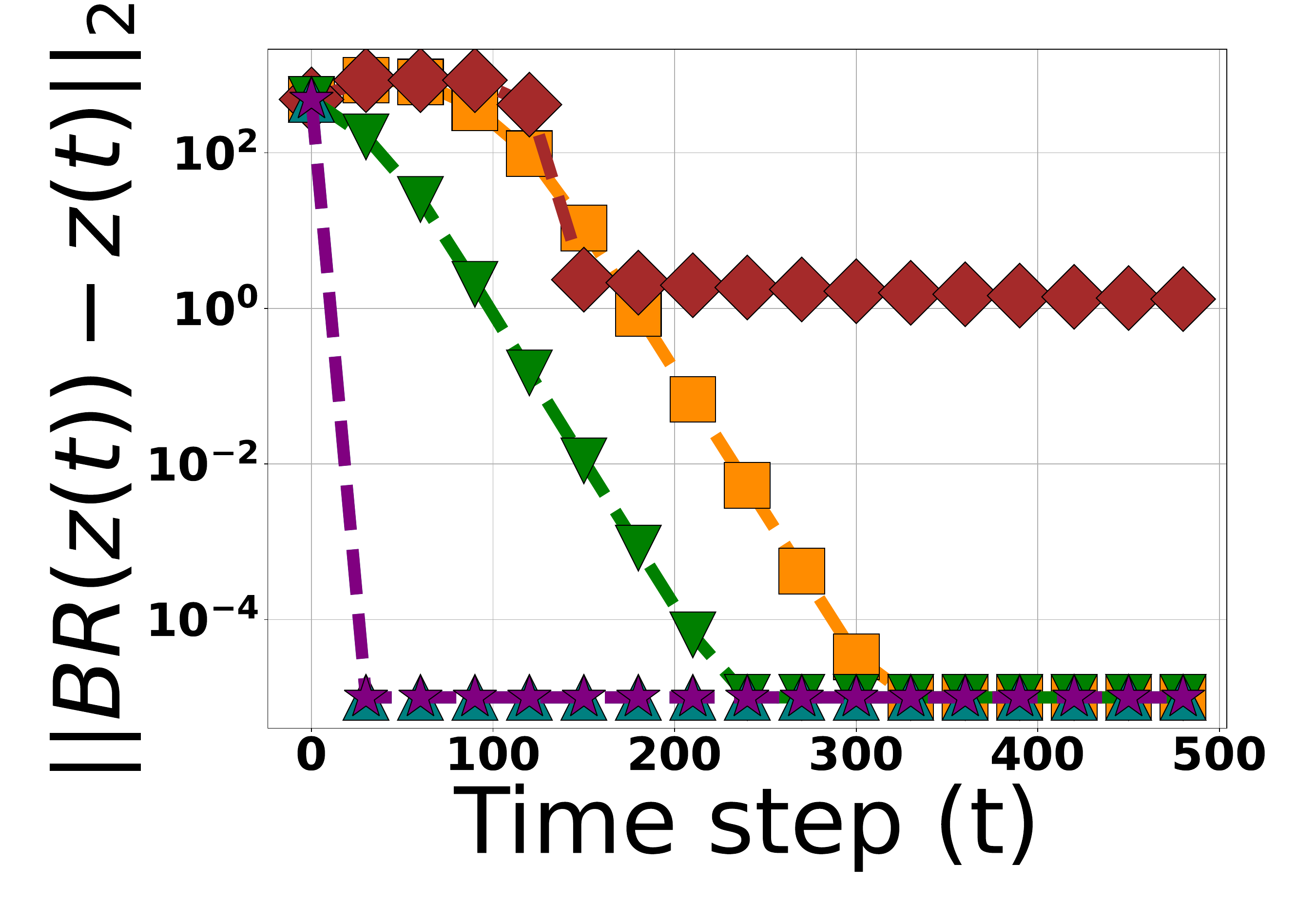}
        \caption{$\gamma = 0$}
        \label{fig:Speed_alpha1_gamma0.0_n_10}
    \end{subfigure}
    \hfill
    \begin{subfigure}{0.32\linewidth}
        \centering
        \includegraphics[width=\textwidth]{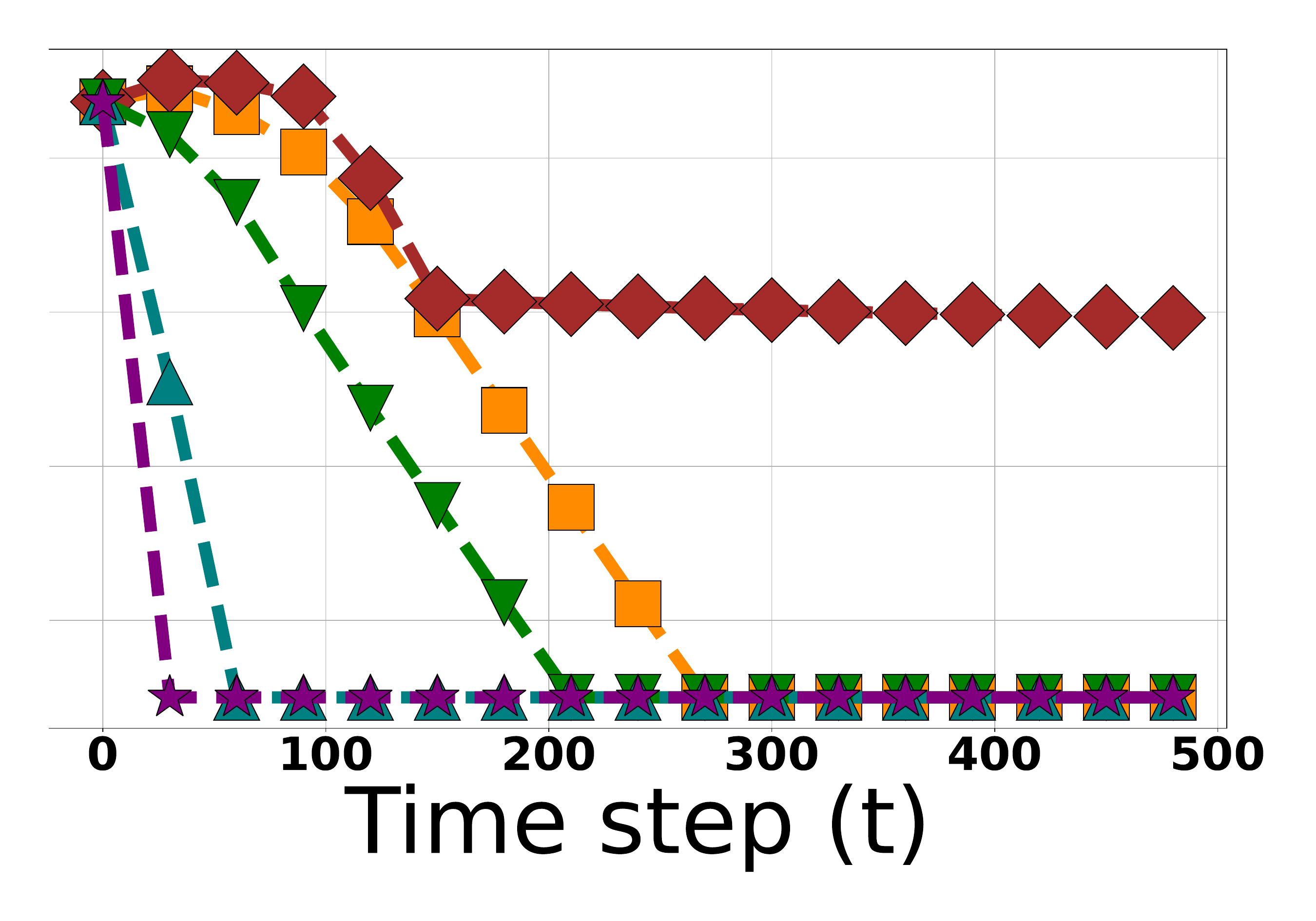}
        \caption{$\gamma = 5$}
        \label{fig:Speed_alpha1_gamma5.0_n_10}
    \end{subfigure}
    \hfill
    \begin{subfigure}{0.32\linewidth}
        \centering
        \includegraphics[width=\textwidth]{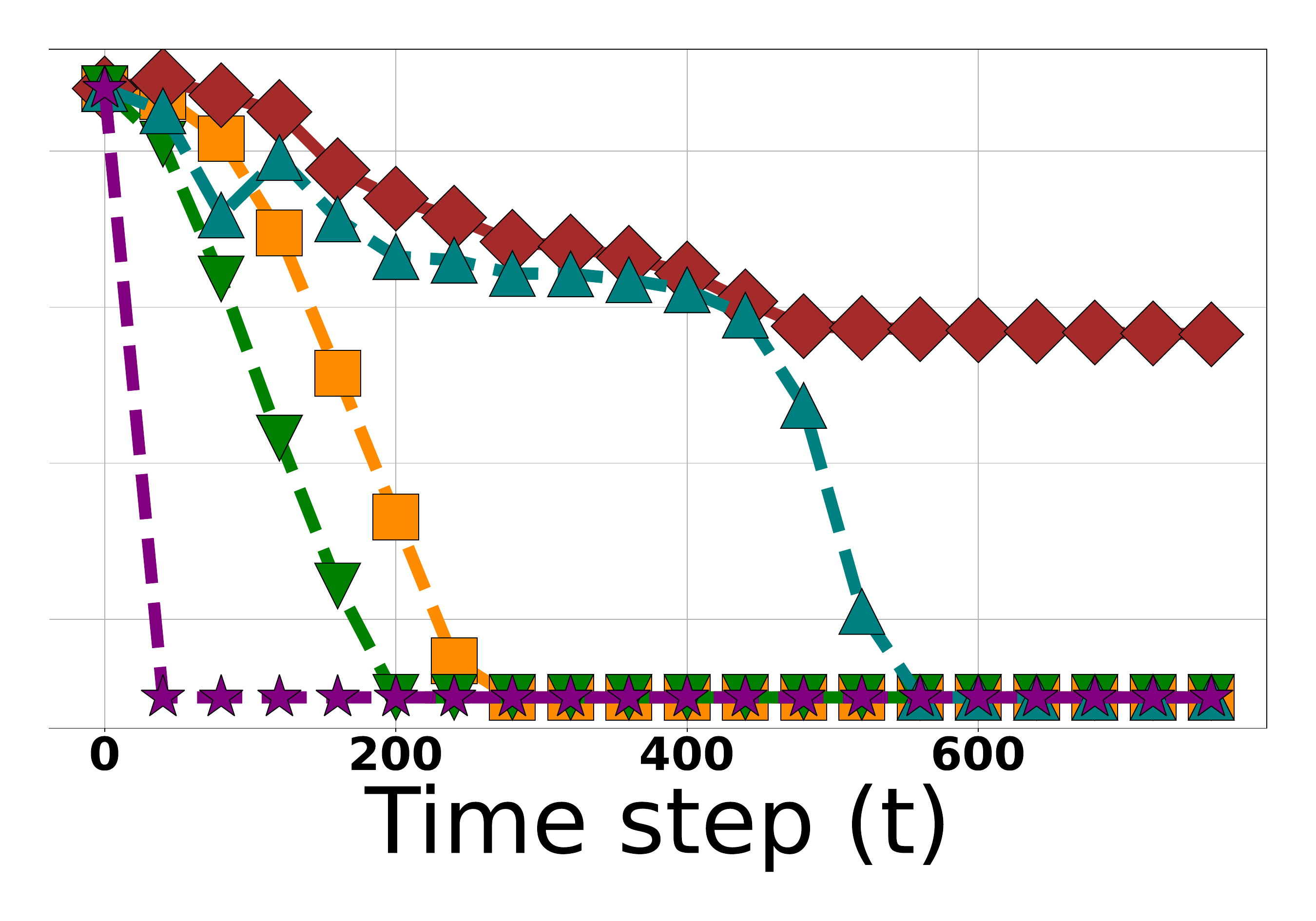}
        \caption{$\gamma = 10$}
        \label{fig:Speed_alpha1_gamma10.0_n_10}
    \end{subfigure}

    \caption{Convergence speed under homogeneous dynamics and varying payoff's heterogeneity levels.}
    \label{fig:Homogeneous_speed_subfigs}
\end{figure}

\begin{figure}[t]
    \centering
    % Legend (optional separate graphic)
    \includegraphics[width=\linewidth]{img/Non_Hybrid/Legend.pdf}

    \centering
        \begin{subfigure}{0.48\columnwidth}
        \centering
        \includegraphics[width=\linewidth]{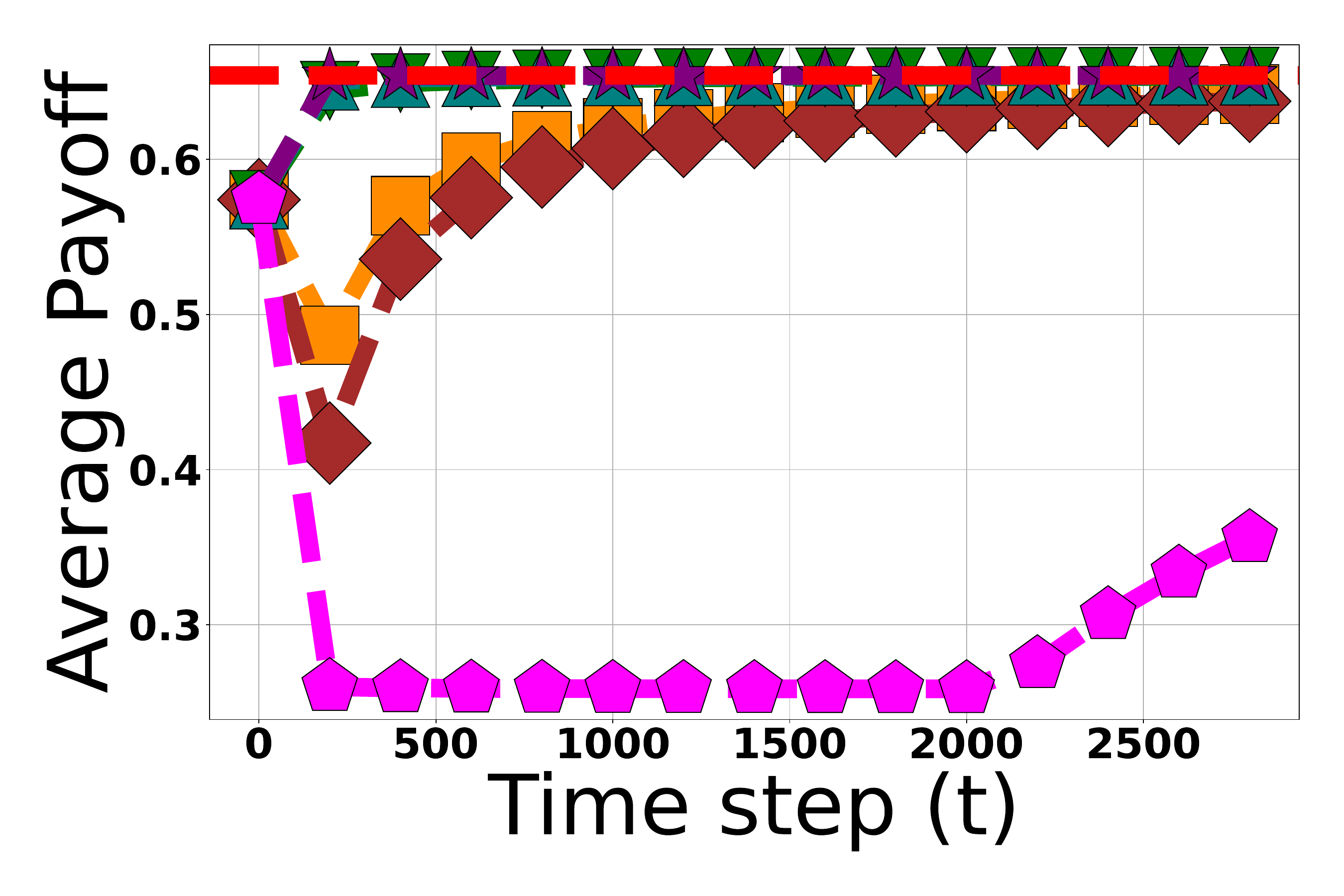}
        \caption{Time-average payoff.}\label{fig:Avg_Payoff_SBRD-DAQ_alpha1_gamma0.0_n_10}
    \label{fig:Avg_Payoff_alpha1_gamma0}
    \end{subfigure}
        \begin{subfigure}{0.48\columnwidth}
        \centering
\includegraphics[width=\linewidth]{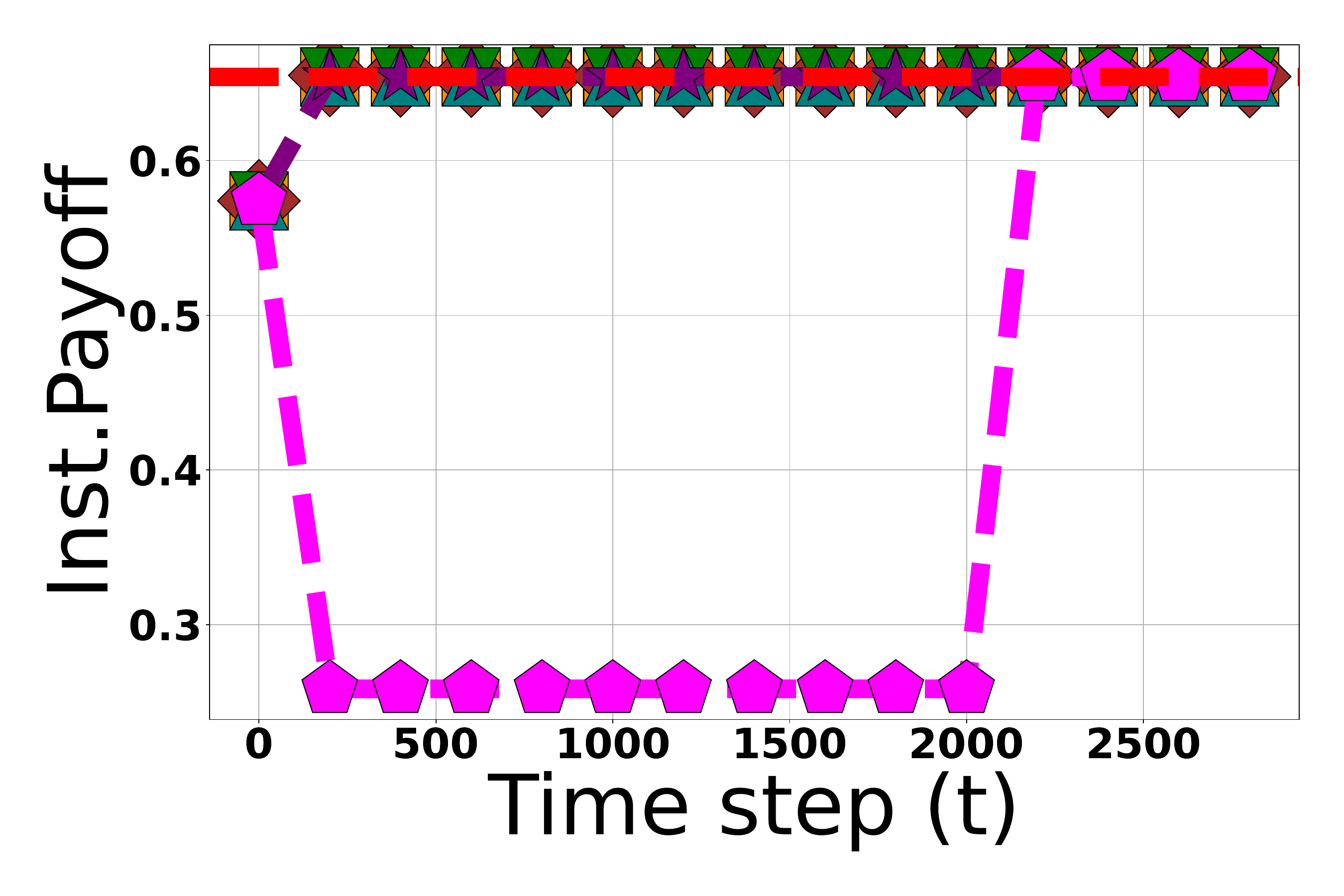}
        \caption{Instantaneous payoff.}\label{fig:Payoff_SBRD-DAQ_alpha1_gamma0.0_n_10}
    \label{fig:Payoff_alpha1_gamma0}
    \end{subfigure}

    \caption{Agent's payoff, $\gamma=0$.}
    \label{fig:non-hybrid_avg_payoff}

\end{figure}

\begin{figure*}[t]
\centering
\captionsetup{justification=centering}

% ---------------- Row 1: BR vs DAQ_F ----------------
\begin{subfigure}[t]{0.32\linewidth}
  \centering
  \includegraphics[width=\linewidth]{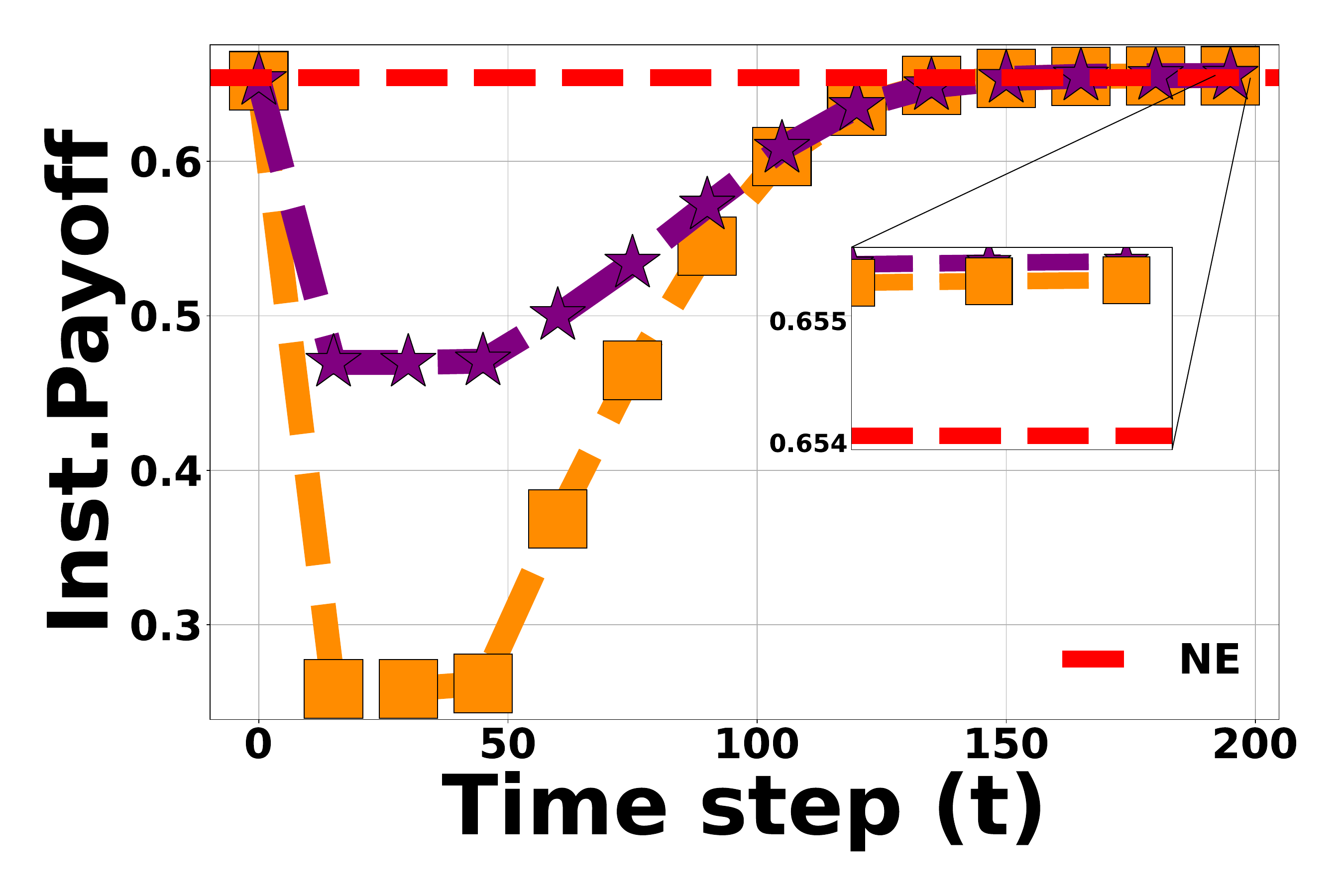}\vspace{-1mm}
  \includegraphics[width=\linewidth]{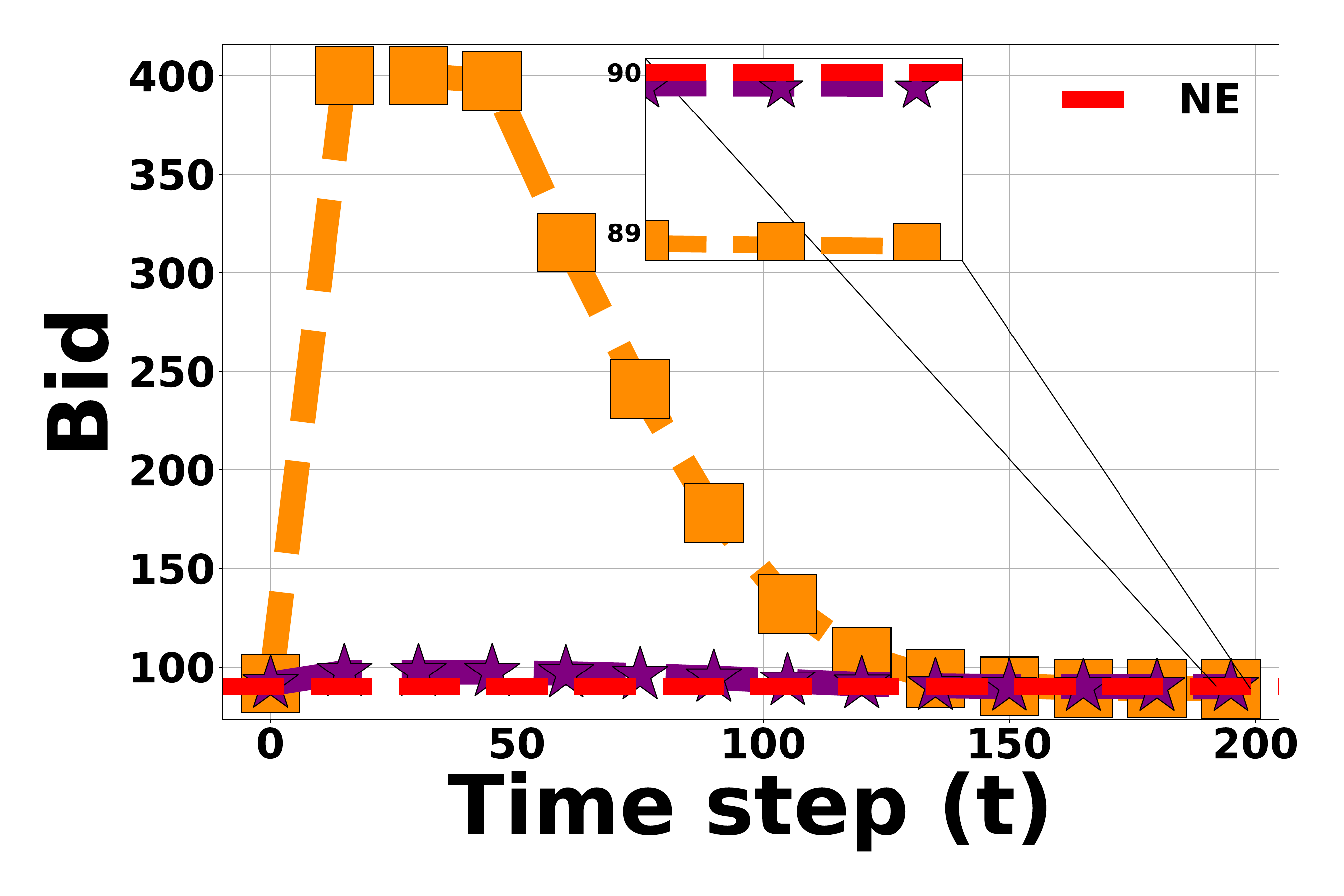}
  \caption{$\BR$ vs. $\DAQ_{\mathrm{F}}$ ($\alpha_{\BR}=10\%$)}
  \label{fig:br_daqF_10_stack}
\end{subfigure}\hfill
\begin{subfigure}[t]{0.32\linewidth}
  \centering
  \includegraphics[width=\linewidth]{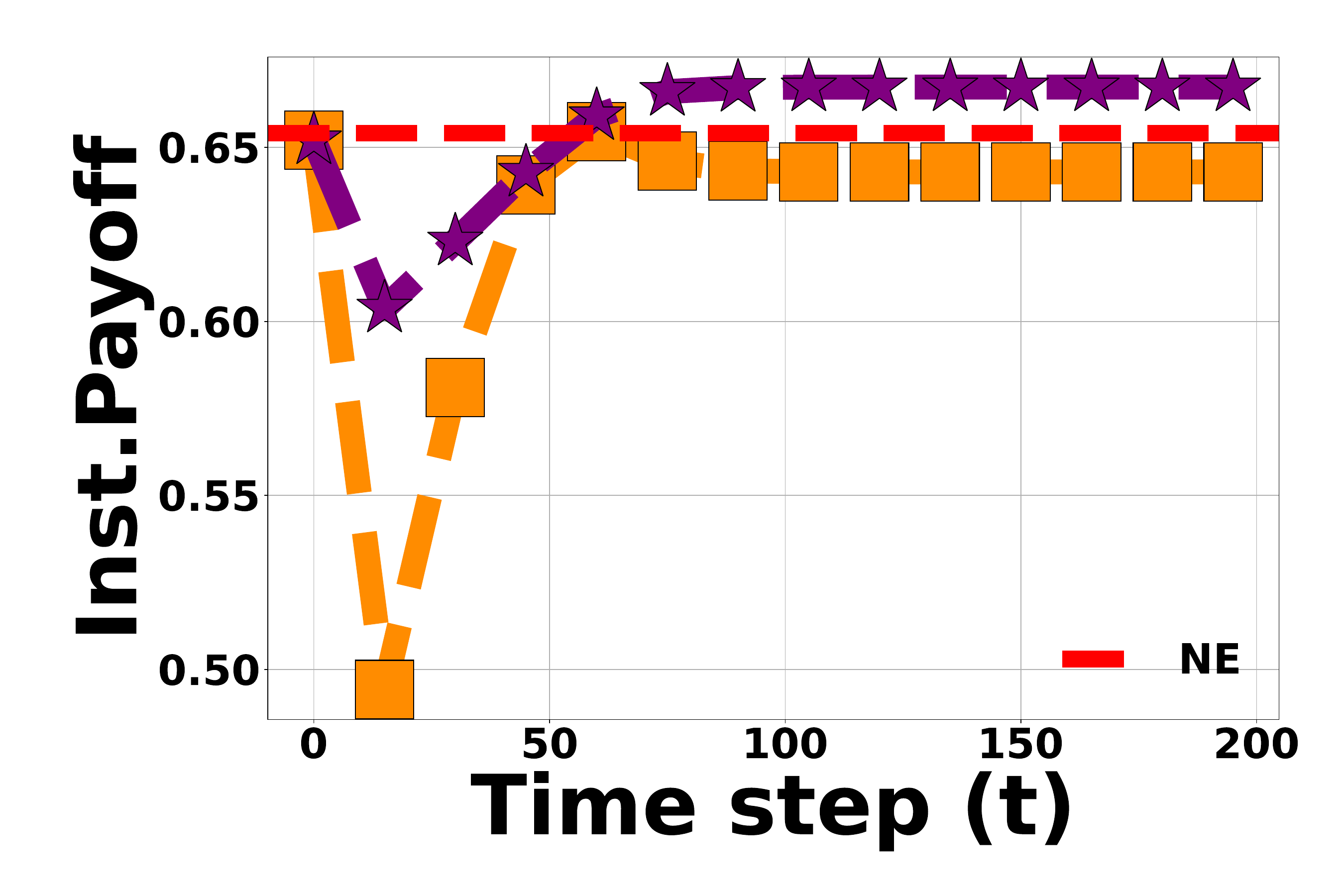}\vspace{-1mm}
  \includegraphics[width=\linewidth]{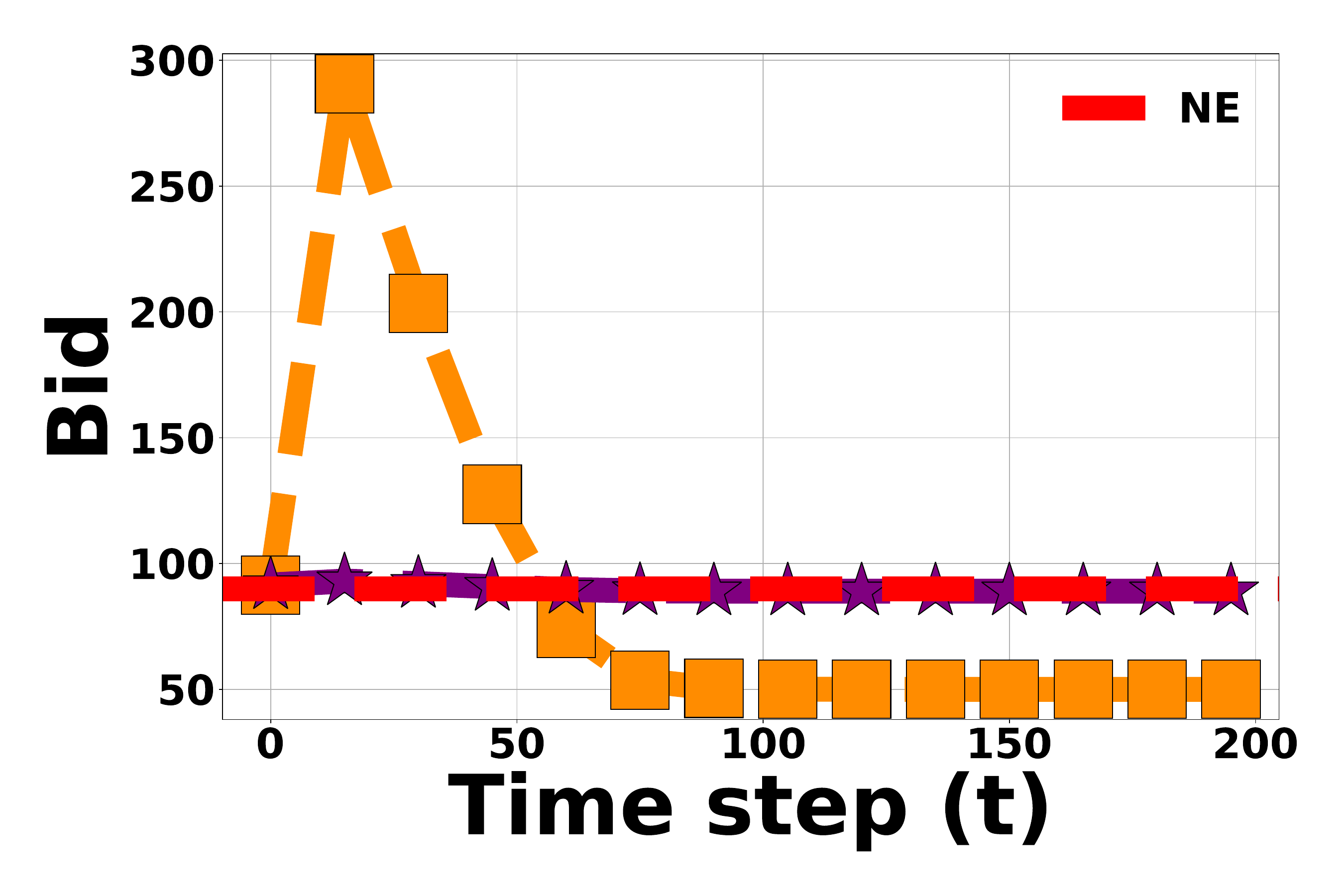}
  \caption{$\BR$ vs. $\DAQ_{\mathrm{F}}$ ($\alpha_{\BR}=80\%$)}
  \label{fig:br_daqF_80_stack}
\end{subfigure}\hfill
\begin{subfigure}[t]{0.32\linewidth}
  \centering
  \includegraphics[width=\linewidth]{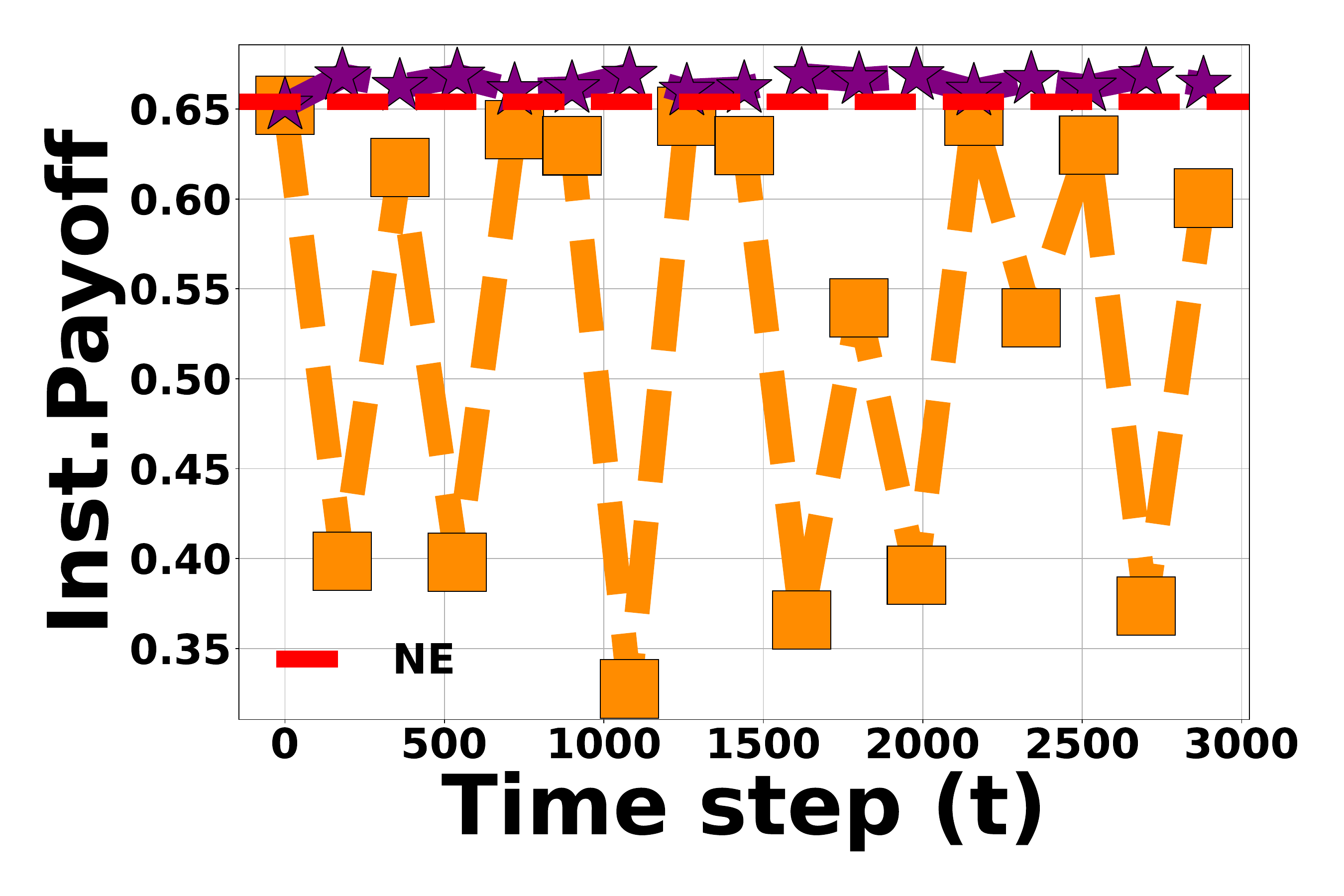}\vspace{-1mm}
  \includegraphics[width=\linewidth]{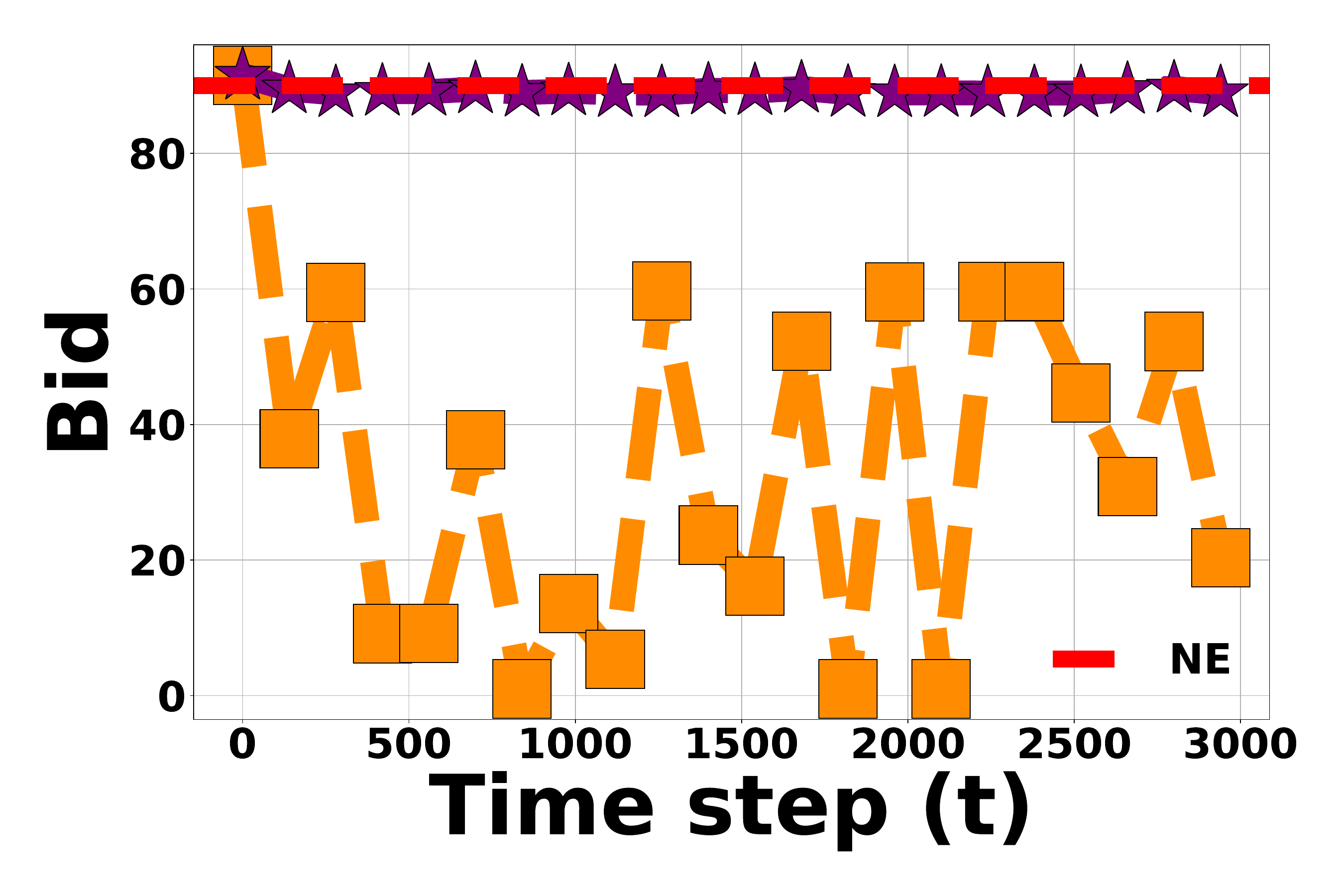}
  \caption{$\BR$ vs. $\DAQ_{\mathrm{F}}$ ($\alpha_{\BR}=90\%$)}
  \label{fig:br_daqF_90_stack}
\end{subfigure}

\vspace{2mm}

% ---------------- Row 2: BR vs OGD_F ----------------
\begin{subfigure}[t]{0.32\linewidth}
  \centering
  \includegraphics[width=\linewidth]{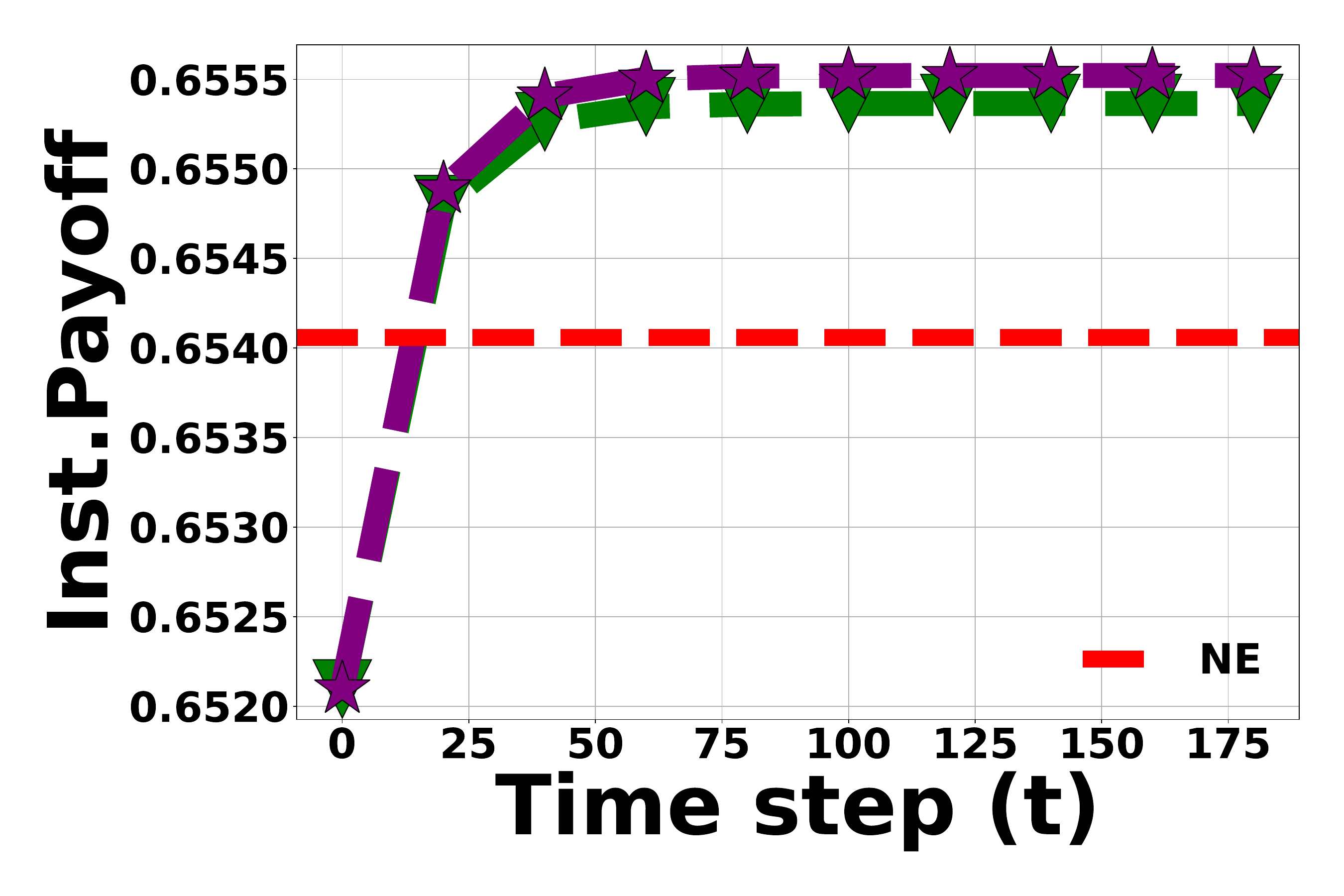}\vspace{-1mm}
  \includegraphics[width=\linewidth]{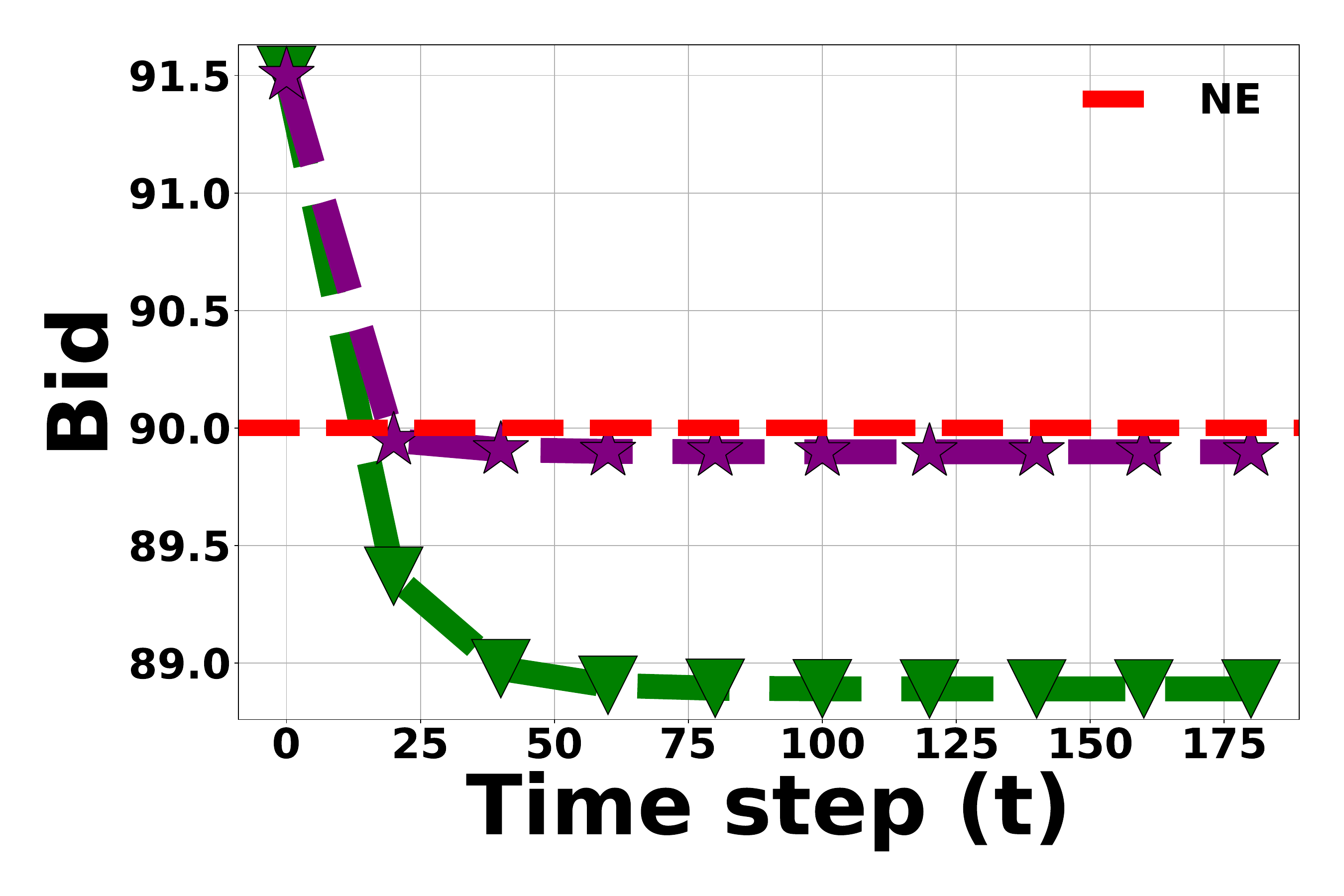}
  \caption{$\BR$ vs. $\OGD_{\mathrm{F}}$ ($\alpha_{\BR}=10\%$)}
  \label{fig:br_ogdF_10_stack}
\end{subfigure}\hfill
\begin{subfigure}[t]{0.32\linewidth}
  \centering
  \includegraphics[width=\linewidth]{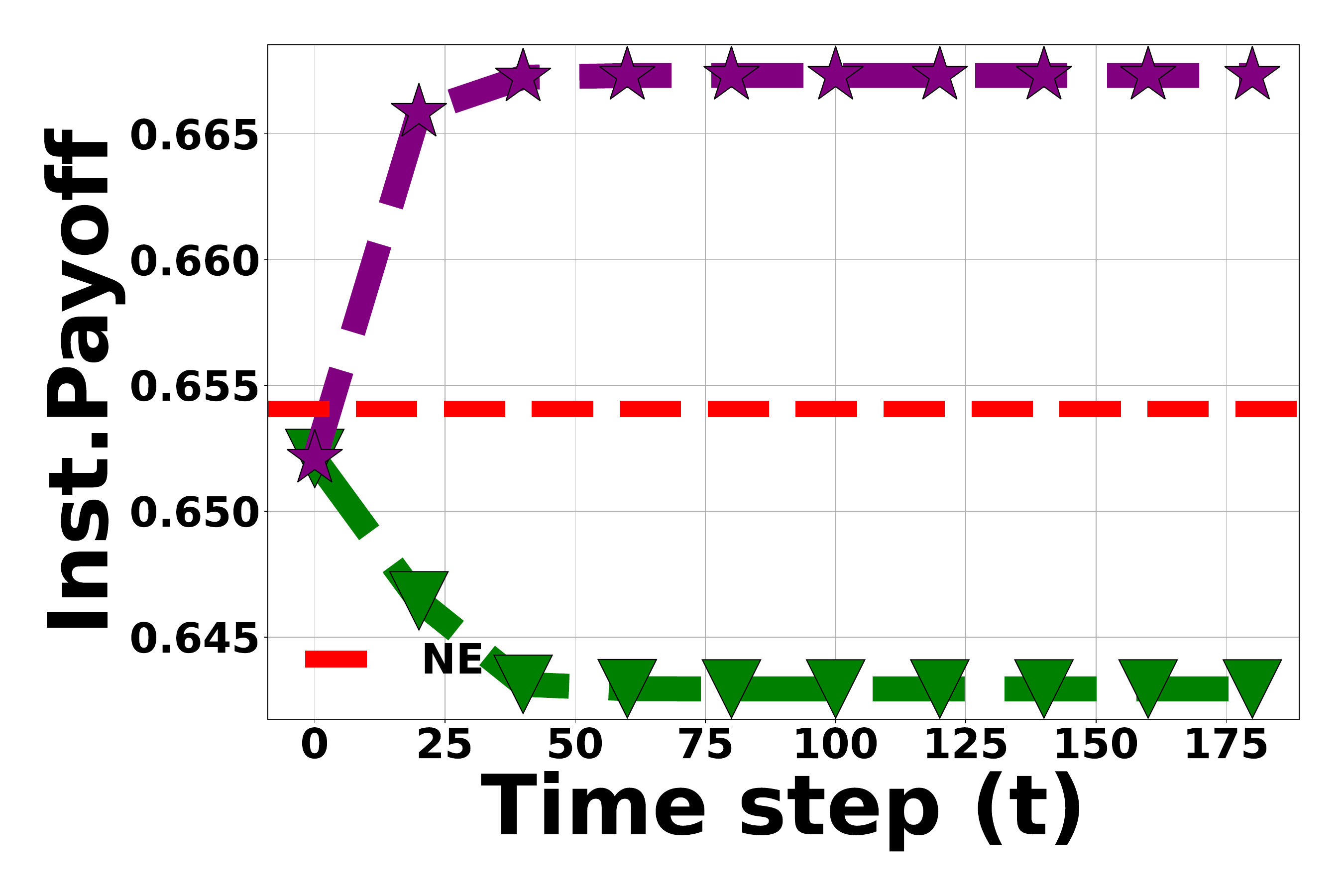}\vspace{-1mm}
  \includegraphics[width=\linewidth]{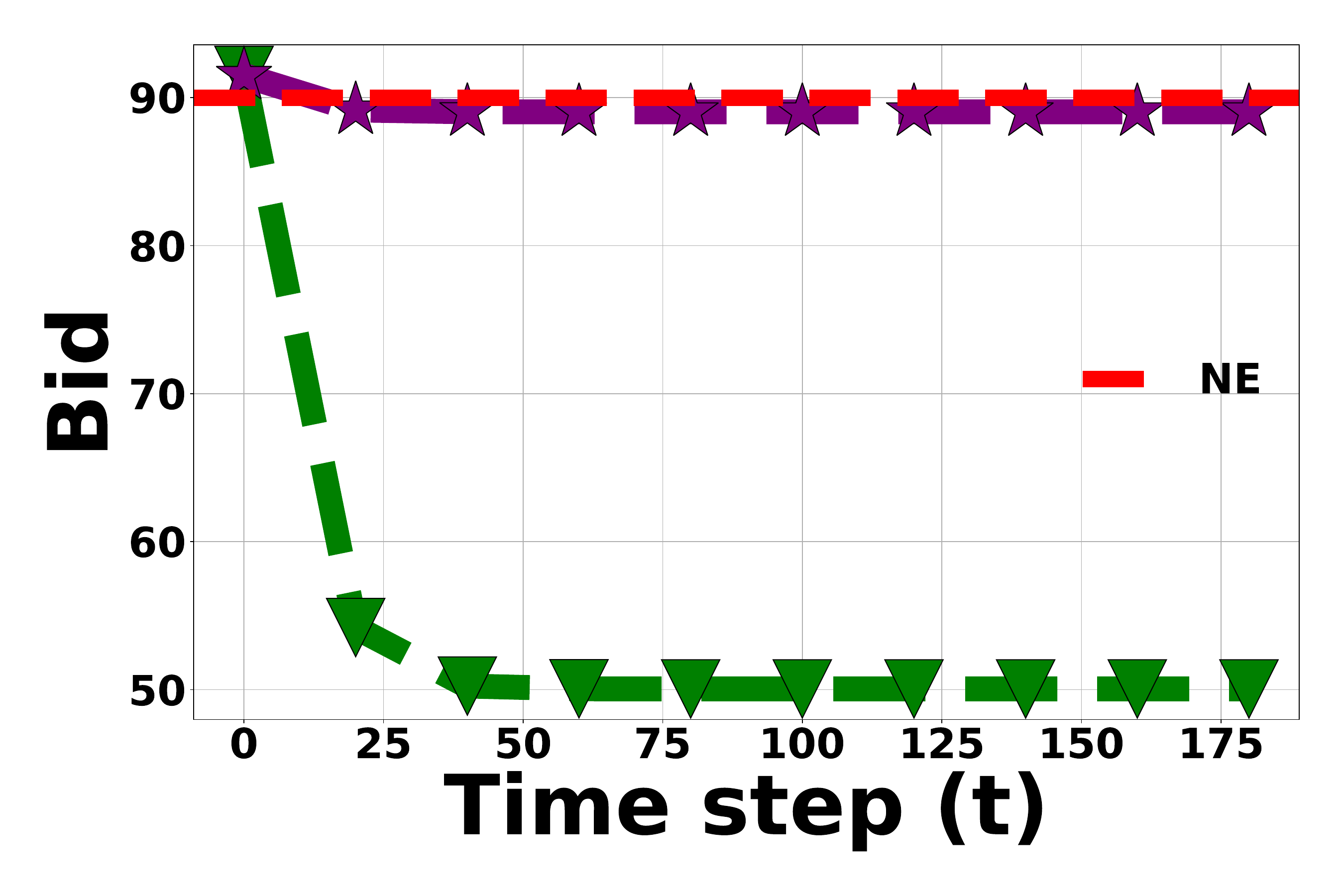}
  \caption{$\BR$ vs. $\OGD_{\mathrm{F}}$ ($\alpha_{\BR}=80\%$)}
  \label{fig:br_ogdF_80_stack}
\end{subfigure}\hfill
\begin{subfigure}[t]{0.32\linewidth}
  \centering
  \includegraphics[width=\linewidth]{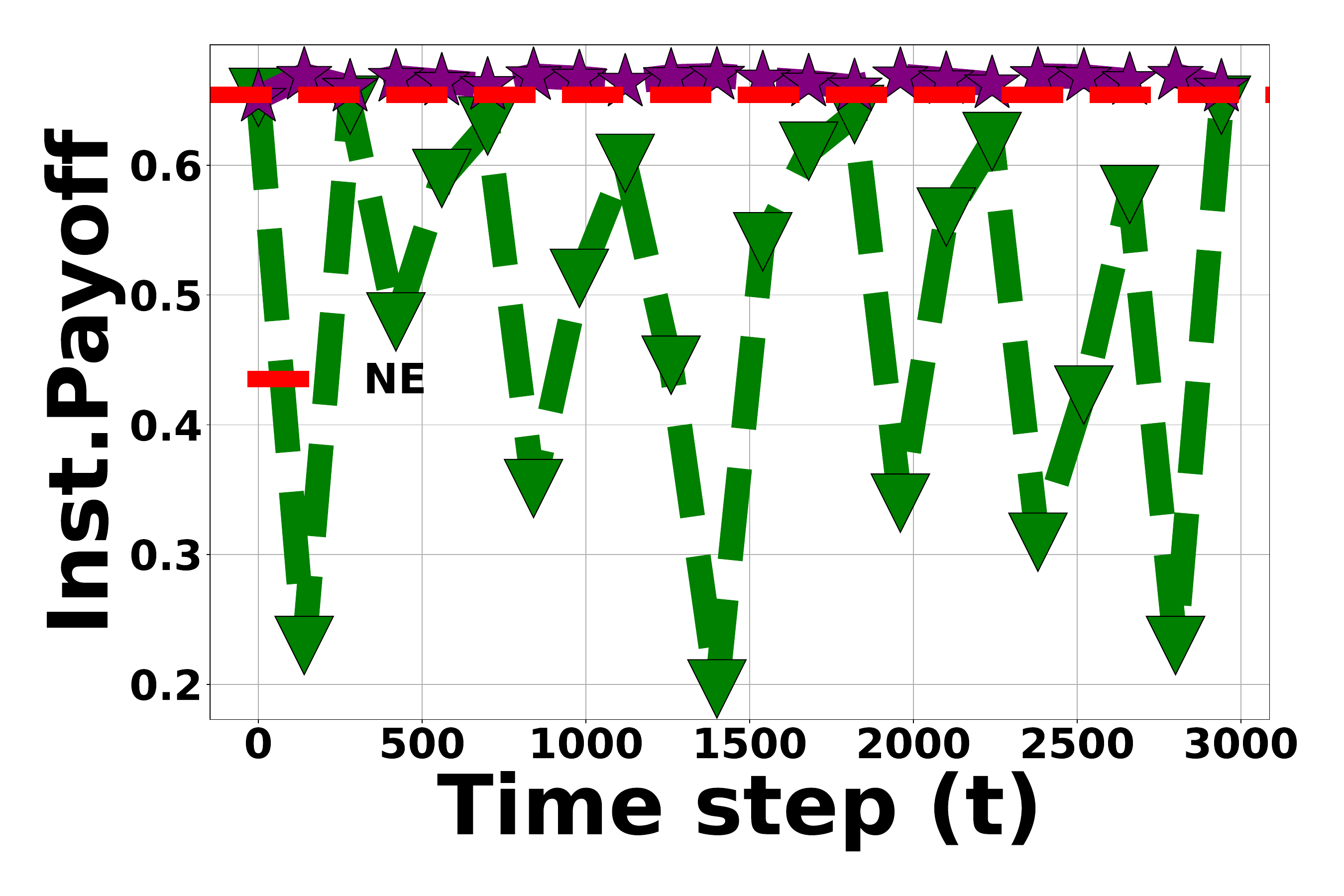}\vspace{-1mm}
  \includegraphics[width=\linewidth]{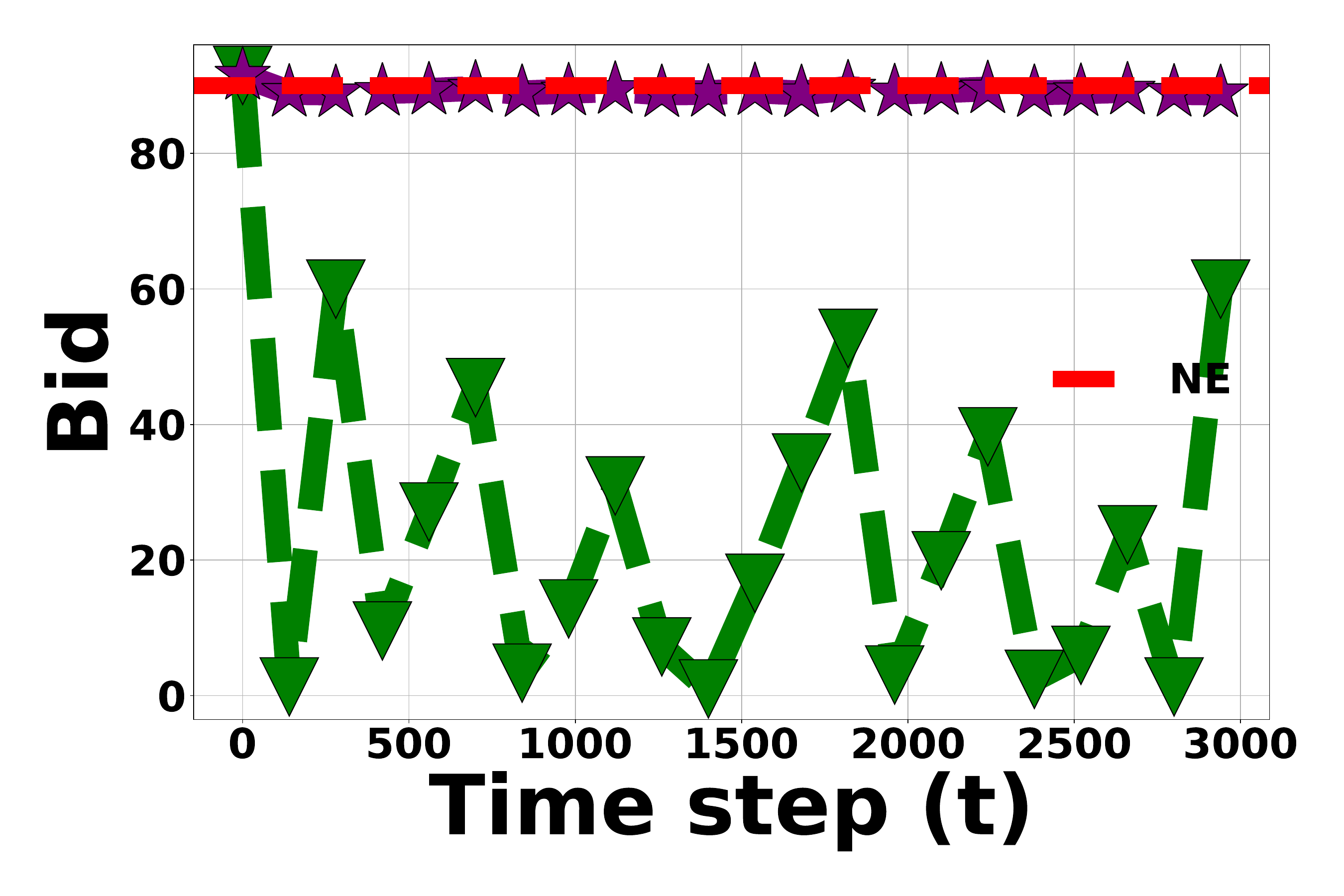}
  \caption{$\BR$ vs. $\OGD_{\mathrm{F}}$ ($\alpha_{\BR}=90\%$)}
  \label{fig:br_ogdF_90_stack}
\end{subfigure}

\caption{Heterogeneous dynamics. In each sub-figure: instantaneous payoff (top) and bids (bottom).}
\label{fig:heter_bid_payoff_time}
\end{figure*}

%%%%%%%%%%%%%%%%%%%

\begin{figure}[t]
    \centering
    %--- Global legend (column width) ---
    % \includegraphics[width=\linewidth]{img/Hybrid/Payoff/DAQ_F_legend.pdf}

    %--- One-row subfigures (fit in one column) ---
    \begin{subfigure}[b]{0.31\linewidth}
        \centering
        \includegraphics[width=\linewidth]{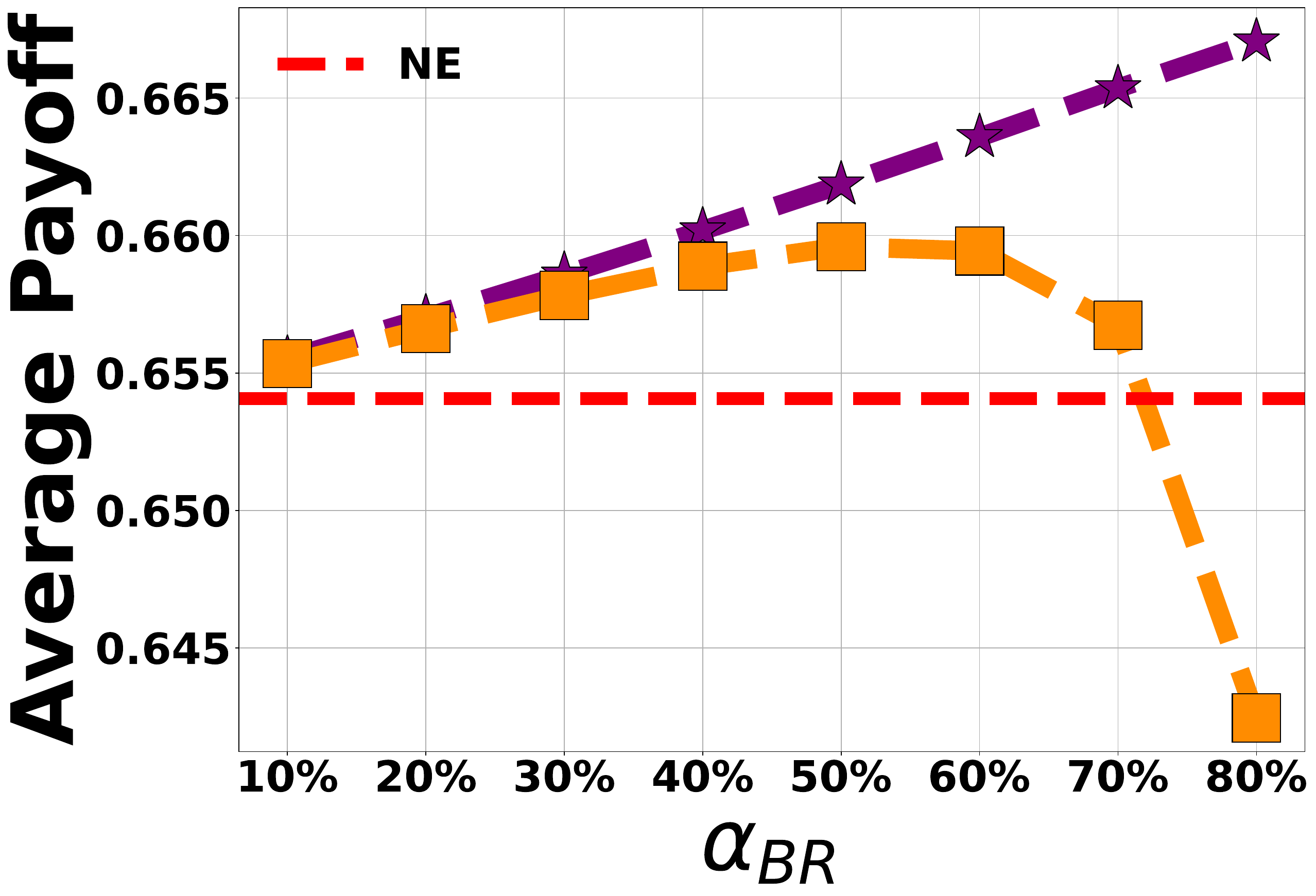}
        \caption{$\DAQ_{\text{F}}$ vs. \BR}
        \label{fig:SBRD_DAQ_F_Payoff}
    \end{subfigure}\hfill
    \begin{subfigure}[b]{0.31\linewidth}
        \centering
        \includegraphics[width=\linewidth]{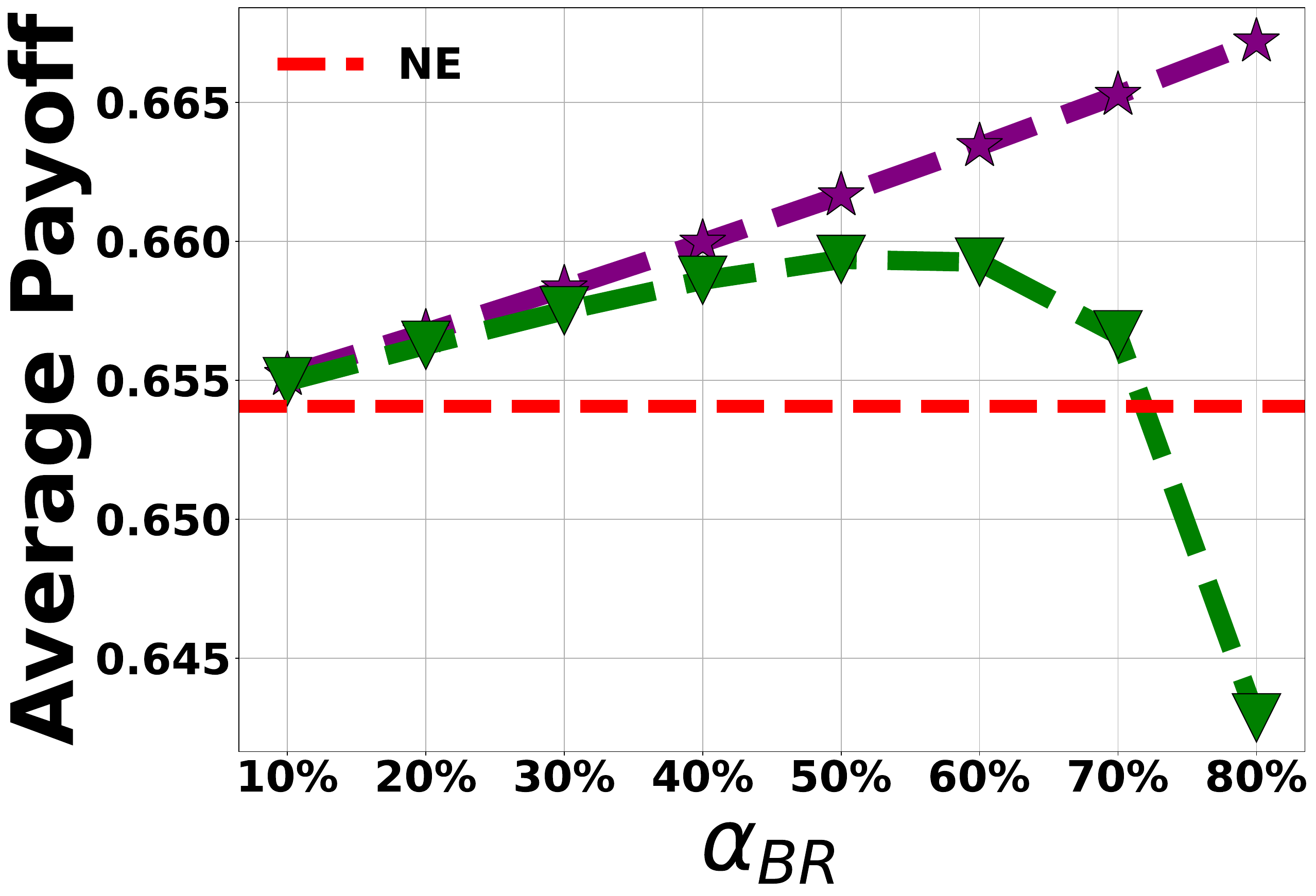}
        \caption{$\OGD_{\text{F}}$ vs. \BR}
        \label{fig:SBRD_OGD_V_Payoff}
    \end{subfigure}\hfill
    %--- Ligne de sous-figures (4 colonnes) ---
    \begin{subfigure}[b]{0.31\linewidth}
        \centering
        \includegraphics[width=\linewidth]{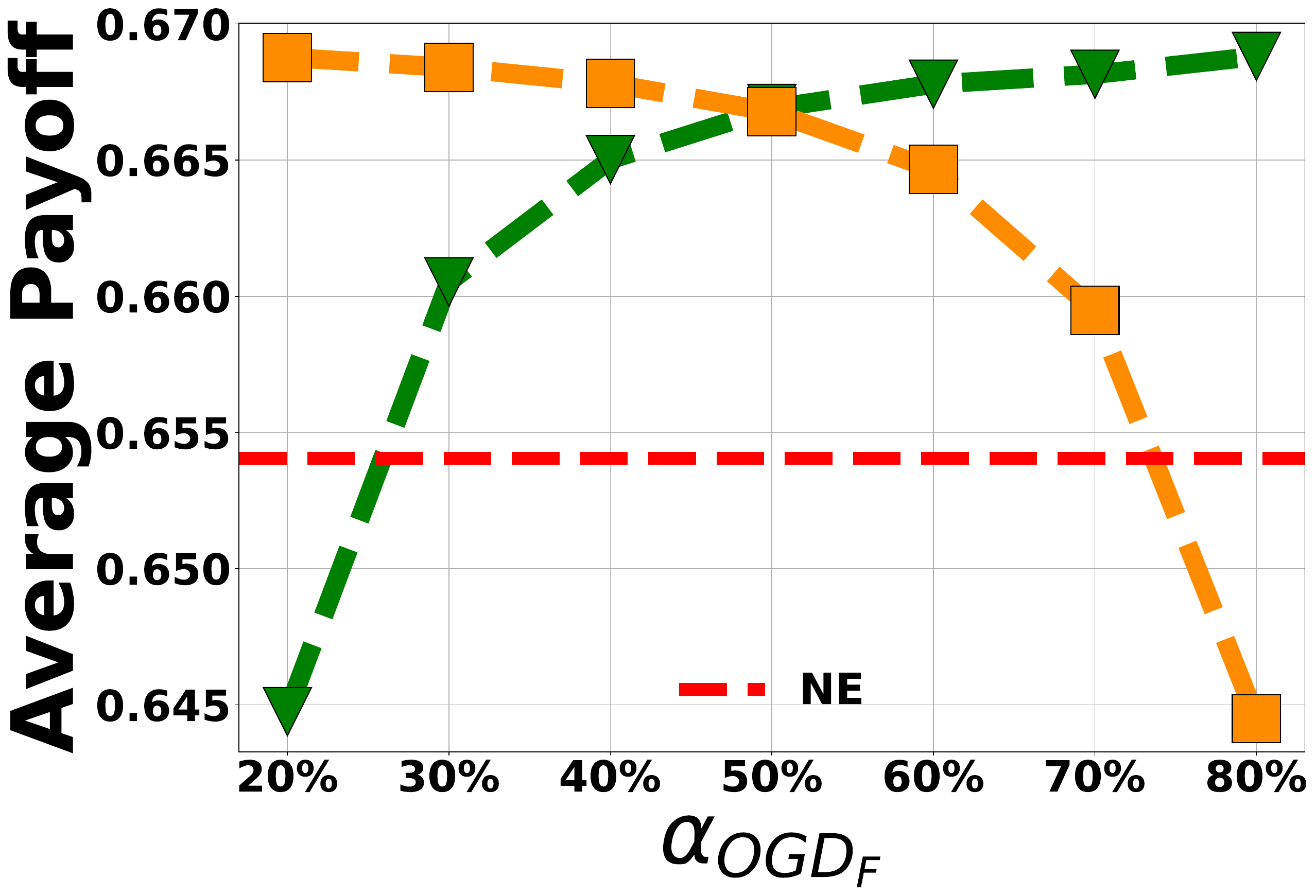}
        \caption{$\DAQ_{\text{F}}$ vs. $\OGD_{\text{F}}$}
        \label{fig:DAQ_F_OGD_V_Payoff}
    \end{subfigure}

    \caption{Heterogeneous dynamics: Average payoff}
    \label{fig:hetero_dyn_avg_payoff_alpha}
\end{figure}

We consider a set of agents~$\cI$. Each agent $i$'s valuation function is $V_i(x)=a_i\ln x$, where $a_i>0$ is an agent-specific parameter, and payment function $p_i(z)=z$. Utilities heterogeneity is determined by $\gamma$ by setting $a_i=\max(a-i\gamma,1)$, with $\gamma \in \{0,5,10\}$, $a=100$, and a budget constraint $c_i=c=400$ and $\epsilon_i = \epsilon=1$.  We set $\delta=0.1$. The number of rounds in the repeated Kelly game is $T=3000$, and results are averaged over $10$ independent runs with random %feasible 
initial bids. Agents follow one of the bidding algorithms described in Section~\ref{ss:Bidding_algorithms}. Namely, $\OGD$, $\DAQ$, $\RRM$, and $\BR$. For $\OGD$, $\DAQ$, and $\RRM$, we consider both fixed and time-varying learning rates: the fixed learning rate is $\eta^{(i)}=\frac{D_i}{G_i\sqrt{T}}$ and the time-varying one is $\eta_t^{(i)}=\frac{D_i}{G_i\sqrt{t}}$. We distinguish these variants by either adding the subscripts $F$ for fixed or $V$ for time-varying. 

Our simulations include both \textit{homogeneous dynamics} settings, where all agents follow the same update rule, and \textit{heterogeneous dynamics} settings, where a fraction $\alpha_{\mathcal{A}_1}$ of agents use algorithm $\mathcal{A}_1$ while the remaining agents use a different algorithm $\mathcal{A}_2$.

The objective is twofold: (1) to verify whether repeated play converges to the Nash equilibrium (NE) of the stage game, while measuring the convergence speed of the bidding algorithms, and (2) to compare their time-average payoff.

To measure convergence to the NE, we use the metric $r_t$, defined as, 
\begin{align}\label{eq:metric_residual}
    r_t\triangleq \|\textbf{BR}(\bm z(t))-\bm z(t)\|_{2} .
\end{align}
Indeed, the unique NE of the stage game is the fixed point of the best-response operator, i.e.,  $\textbf{BR}(\bm{z^{*}})=\bm{z}^{*}$. Accordingly, when the bid profile $\bm{z}(t)$ is near the NE, one expects $\textbf{BR}(\bm{z}(t))\approx\bm{z}(t)$.

To make payoffs comparable across agents, we normalize them to lie in $[0,1]$. Specifically, we define $\overline{\varphi_i}(\bm z)=
\frac{\varphi_i(\bm z)-\varphi_{\min}}{\varphi_{\max}-\varphi_{\min}}$ where $\varphi_{\min}$ and $\varphi_{\max}$ are the smallest and largest values of the payoff functions among all players and across all actions. We then compare the bidding algorithms using the time-average normalized payoff, i.e., $\frac{1}{T}\sum_{t=1}^T \overline{\varphi_i}(\bm z(t))$, which matches each player's objective of maximizing long-run payoff.

The rest of this section is organized as follows; Section~\ref{ss:homogeneous_dyn} addresses these questions under \textit{homogeneous dynamics}, whereas Section~\ref{ss:heterogeneous_dyn} addresses them under \textit{heterogeneous dynamics}.

% The rest of this section is organized as follows. Section~\ref{ss:homogeneous_dyn} evaluates whether the repeated play converges to the Nash equilibrium of the stage game and compares the convergence speed across the bidding algorithms used by the agents. Section~\ref{ss:heterogeneous_dyn} compares 

% \YBnote{\begin{itemize}
%     \item To save some space, we can move the experiments with the $\epsilon$ metric and social welfare to the appendix. The results with $r_t$ are easier to visualize so we can keep them. We can also have this figure for $r_t$ taking a full row of the page.
%     \end{itemize}}

\subsection{Homogeneous dynamics} 
\label{ss:homogeneous_dyn}

%Figure~\ref{fig:Homogeneous_speed_subfigs} reports the evolution of the metric $r_t$ under homogeneous dynamics, for several values of $\gamma$. Repeated play with $\OGD$, $\DAQ_{\mathrm{F}}$, or $\BR$ converges to the Nash equilibrium, which aligns with our theoretical results. In contrast, $\DAQ_{\mathrm{V}}$ fails to converge with a similar precision. In terms of convergence speed, $\BR$ is the fastest, followed by $\OGD$ and $\DAQ_{\mathrm{F}}$, while $\OGD_V$ appears to be more sensitive to player heterogeneity ($\gamma$) than $\OGD_F$.  

%Table~\ref{tab:convergence_residual} reports, for $n=|\cI|\in\{2,10,20\}$, the minimum number of iterations needed to reach the threshold $r_t\leq 10^{-5}$. The results indicate that the convergence of $\BR$ becomes faster as the number of players increases, whereas the other bidding algorithms exhibit the opposite trend. This behavior is consistent with the $\cO(n)$ convergence bound in Theorem~\ref{cor:convergence_DAQ}. The method $\RRM$ converges significantly more slowly than the others. Finally, varying the heterogeneity level through $\gamma$ has only a limited overall effect on convergence rates, though it slightly accelerates $\OGD_F$, $\DAQ_{\mathrm{F}}$, and $\RRM$.

Table~\ref{tab:convergence_residual} reports, for $n=|\cI|\in\{2,10,20\}$, the minimum number of iterations required to reach the threshold $r_t \le 10^{-5}$. 
The results show that repeated play with $\OGD$, $\DAQ_{\mathrm{F}}$, $\RRM$, and $\BR$ converges to the Nash equilibrium, in agreement with our theoretical results. In contrast, $\DAQ_{\mathrm{V}}$ may fail to converge with comparable precision. 
In terms of convergence speed, $\BR$ is the fastest, followed by $\OGD$ and $\DAQ_{\mathrm{F}}$, while $\RRM$ converges significantly more slowly. 
Moreover, the convergence of $\BR$ becomes faster as the number of players increases, whereas the other bidding algorithms exhibit the opposite trend. 
This behavior is consistent with the $\cO(n)$ convergence bound established in Theorem~\ref{cor:convergence_DAQ}. 
Finally, varying the heterogeneity level through $\gamma$ has only a limited overall effect on convergence rates, although it slightly accelerates $\OGD_F$, $\DAQ_{\mathrm{F}}$, and $\RRM$.

In the following experiments, we focus on a population of $n=10$ agents. 
Figure~\ref{fig:Homogeneous_speed_subfigs} illustrates the evolution of the fixed-point residual $r_t$ under homogeneous dynamics ($\OGD$, $\DAQ$, and $\BR$) for the same values of $\gamma$.

Figure~\ref{fig:non-hybrid_avg_payoff} reports, for a representative player, both the instantaneous payoff and the time-average payoff. Since $\gamma=0$, all agents share the same payoff function, so the plotted curves are representative of any player. %The figure shows that the ranking in payoff performance mirrors the ranking in convergence speed: best-response dynamics achieves the highest payoffs, followed by $\OGD$, then $\DAQ$, and finally $\RRM$.

The figure shows that the ranking in payoff performance mirrors the ranking in convergence speed: best-response dynamics achieves faster the limit payoffs, followed by $\OGD$, then $\DAQ$, and finally $\RRM$.

%An additional observation is that, although $\DAQ_{\text{V}}$ does not attain the high-precision threshold in the residual metric $r_t$, it still achieves time-average payoffs comparable to its fixed-learning rate counterpart $\DAQ_{\text{F}}$. Figures~\ref{fig:Payoff_SBRD-DAQ_alpha1_gamma0.0_n_10} and \ref{fig:Homogeneous_speed_subfigs} suggest that, for both $\DAQ_{\text{V}}$ and $\RRM$, reaching a moderately small $r_t$ (e.g., $r_t<1$) is already sufficient for the instantaneous payoff to be close to the Nash-equilibrium payoff; further reductions in $r_t$ bring only marginal payoff improvements. This explains why $\DAQ_{\text{V}}$ performs well in terms of payoff despite not converging to very high precision. Finally, the poor payoff of $\RRM$ indicates that it is not an attractive update rule from a selfish perspective.

It is interesting to observe that that, even though $\DAQ_{\text{V}}$ does not yet converge within the tested horizon (see Table~\ref{tab:convergence_residual}), it still achieves time-average payoffs comparable to its fixed-learning rate counterpart $\DAQ_{\text{F}}$. By comparing Figure~\ref{fig:Homogeneous_speed_subfigs} and Figure~\ref{fig:Payoff_SBRD-DAQ_alpha1_gamma0.0_n_10} we note that, for both $\DAQ_{\text{V}}$ and $\RRM$, reaching a moderately small $r_t$ (e.g., $r_t<1$) is sufficient to render the instantaneous payoff close to one attained at the Nash-equilibrium. Actually, further reductions in $r_t$ bring only marginal payoff improvements. This explains why $\DAQ_{\text{V}}$ performs well in terms of payoff despite not converging to very high precision. Finally, the poor payoff of $\RRM$ indicates that it is not an attractive update rule from a selfish perspective.

% \YBnote{\begin{itemize}
%     \item As you noticed, here I only included one figure for $r_t$ and one table for the same metric. I think we should moved them here. The figure should also take the whole row (no minipage and you add *). 

%     \item We need to figure out what to show for the heterogeneous dynamics. What are the most interesting results we observe there? I assume not converging to the Nash is important. But what else?  

% \end{itemize}}

\subsection{Heterogeneous dynamics}
\label{ss:heterogeneous_dyn}
We consider now \textit{heterogeneous dynamics} where a fraction $\alpha_{\mathcal{A}_1}$ of agents uses algorithm $\mathcal{A}_1$ while the remaining agents use $\mathcal{A}_2$ with $\gamma=0$.

Figure~\ref{fig:heter_bid_payoff_time} shows the evolution over time of the instantaneous bid and payoff of two representative agents---one using $\cA_1$ and the other using $\cA_2$---for $(\cA_1, \cA_2)\in \{ (\BR,\OGD) , (\BR,\DAQ)\}$ and for $\alpha_{\BR}\in \{10\%, 80\%, 90\% \}$. Similar results for the couple~$(\OGD, \DAQ)$ are available in the supplementary material.

Overall, heterogeneous dynamics do not appear to converge to the stage-game NE. For $\alpha_{\BR}\in \{10\%, 80\%\}$, the trajectories nevertheless appear to settle to a steady regime. In the $(\BR,\DAQ)$ regime, $\DAQ$ at first saturates its budget constraint and then exhibits a slow transient regime before stabilizing, whereas $\OGD$ adapts relatively faster. This is consistent with $\DAQ$ aggregating gradients over time and $\OGD$ responding more strongly to recent feedback.
Finally, for $\alpha_{\BR}= 90\%$, both $\OGD$ and $\DAQ$ show persistent bid oscillations, while $\BR$ bids remain comparatively stable. 
Since $\BR$ agents form the large majority, the aggregate bid changes little so best responses vary only little because of the concavity of the payoff function.

In terms of instantaneous payoff, for $\alpha_{\BR}\in \{10\%, 80\%\}$ the $(\BR,\DAQ)$ configuration displays periods of low payoff for $\BR$, but even lower for $\DAQ$, due to the $\DAQ$ agents budget saturation. In contrast, with $(\BR,\OGD)$ the system stabilizes more quickly. In both cases, the $\BR$ agent's payoff converges to a value close (and sometimes slightly above) the stage-game NE value, whereas the $\DAQ$ and $\OGD$ agents can fall below the NE payoff at $\alpha_{\BR}=80\%$. Finally, for $\alpha_{\BR}= 90\%$, the oscillatory bids of $\OGD$ and $\DAQ$ player translate into large payoff fluctuations, with a payoff repeatedly dropping to significantly lower values than their NE payoff. 

Finally, Figure~\ref{fig:hetero_dyn_avg_payoff_alpha} reports the average payoff of two representative agents---one using $\cA_1$ and the other using $\cA_2$---for $(\cA_1, \cA_2)\in \{ (\BR,\OGD) , (\BR,\DAQ), (\OGD,\DAQ) \}$ and multiple values of $\alpha_{\cA_1}$. Overall, heterogeneous play can yield time-average payoffs that are slightly above or below the stage-game NE payoff, sometimes benefiting one group more than the other. However, these deviations remain small: across all mixtures, observed payoffs lie in $[0.64, 0.67]$, while the NE payoff is roughly $0.655$. Both $\OGD$ and $\DAQ$ exhibit similar trends: when they represent less than $30\%$ of the population, their time-average payoff falls below the NE payoff, whereas for larger fractions it can exceed it. In contrast, $\BR$ achieves payoffs consistently above the NE, with peak values for large $\alpha_{\BR}$. Overall, while payoff differences across policies remain small in these experiments, the $\BR$ dynamics appear slightly preferable, despite lacking no-regret guarantees.

% They also indicate that our homogeneous convergence guarantees to the stage-game NE do not directly extend to heterogeneous dynamics. 

% \YBnote{\begin{itemize}
%     \item We need first a figure showing the evolution of the bids, and perhaps also the associated instantaneous payoff. For this figure, we will mainly conclude that convergence to the NE does not occur with the same precision as in the homogeneous dynamics. We also deduce, if the bids converge, that also the time-average payoff will converge to that value. We also deduce how different are the utilities with respect to their value at the NE. 
% \end{itemize}}

\section{Conclusion}
\label{sec:conclusion}

In this paper, we study the game induced by the competition among agents in a proportional allocation auction. We derive a sufficient condition under which the game satisfies Rosen’s \textit{diagonal strict concavity} (DSC). The condition holds in particular when agents have logarithmic utilities in their allocated share. We further relate these utilities to bandwidth allocation problems in which the objective is to balance fairness and throughput. 

We then consider the repeated version of the game, where all agents update their bids using either best-response dynamics or classical no-regret learning algorithms, namely Online Gradient Descent and Dual Averaging. Under DSC, these homogeneous dynamics are proved to converge to the stage-game Nash equilibrium.

Several extensions of this work can be considered. First, it would be desirable to develop a theoretical framework to capture the impact of heterogeneous update rules on system dynamics. Another direction is to extend the analysis to more general utility functions, such as $\alpha$-fair utilities, or to settings in which players bid for multiple heterogeneous resources.

\bibliographystyle{ieeetr}
\bibliography{main.bib} 

\onecolumn

\section{Supplementary Material}

\subsection{Proof of Theorem~\ref{thm:SDSC_log}}\label{sec:SDSC}

Define the functions $f_i$, $g_i$, and $\psi_{\bm{r},\bm{V}}$ as
\begin{align}
        &f_i(x) = (1 - x)^2 \ddot{V}_i(x) - 2(1 - x)\dot{V}_i(x) , \\ 
        &g_i(x) = -  x(1 - x) \ddot{V}_i(x)+  (2x - 1) \dot{V}_i(x),  \\
        & \psi_{\bm{r},\bm{V}}(\bm{x}) \triangleq \left( \sum_{i\in \cI}\frac{r_ig_i(x_i)^2}{k_i(x_i)} \right) \left(\sum_{i\in \cI}\frac{1}{r_ik_i(x_i)} \right),
\end{align}
where $k_i(x) =  g_i(x) - f_i(x) + \delta^2 L_i$ and $L_i= \min_{z_i \in \mathcal{R}_i} \ddot{p}_i(z_i)$. Further define the set $\Delta = \{ \bm{x}>\bm{0}: \;\sum_{i\in \cI} x_{i}\leq \frac{\sum_{k\in \cI} c_k}{\sum_{k\in \cI} c_k + \delta}  \}$.

We prove that, 
\begin{align}\label{e:suff_cond_SDSC_supp}
               \exists \bm{r} > \bm{0} : \; \sup_{\bm{x}\in \Delta} \psi_{\bm{r},\bm{V}}(\bm{x}) < 1 \implies \mathcal{G} \text{ is } \SDSC(\bm{r}),
\end{align}
or equivalently we prove that the matrix~$\bm{H}_{\bm{r}}(\bm{z})$ is negative definite whenever~$\sup_{\bm{x}\in \Delta} \psi_{\bm{r},\bm{V}}(\bm{x}) < 1$. 

The second-order partial derivatives of the payoff functions $\partial_{i,j}^{2} \varphi_i$ can be written as
\begin{align}
    \partial_{i,j}^{2} \varphi_i(\bm{z}) = \partial_{i,j}^{2} U_i(\bm{z}) - \ddot{p}_{i}(z_i),
\end{align}
where $U_i(\bm{z}) = V_i(x_i(\bm{z}))$, and $\ddot{p}_{i}$ is the second derivative of $p_i$. Note that $p_i$ is convex, which implies it contributes to the matrix $\bm{H}_{\bm{r}}(\bm{z})$ as a diagonal matrix whose entries are $-2r_i \ddot{p}_i(z_i)\leq 0$  and thus is negative semi-definite. Consequently, any failure of $\bm{H}_{\bm{r}}(\bm{z})$ to be negative definite would stem from the matrix formed by $\partial_{i,j}^{2} U_i(\bm{z})$. Lemma~\ref{lem:H_computation} shows that this matrix has a particular structure. 

\begin{lemma}\label{lem:H_computation}
    The partial derivatives of $U_i$ verify,  
    \begin{align}
       m(z)  \partial_{i,j}^{2} U_i(\bm{z}) = 
        \begin{cases} 
            f_i(x_i(\bm{z})) , & \text{if } i = j, \\ 
            g_i(x_i(\bm{z})), & \text{otherwise.} 
        \end{cases}
    \end{align}
   
\end{lemma}
\begin{IEEEproof}
The calculation of the derivative of $U_i$ with respect to $z_i$ writes 
        % \noindent \( \mathcal{J}_{\bm{1}}(z) \):
        \begin{equation}
        \partial_i U_i(z_i) = \frac{\sum_{j\ne i}^N z_j +\delta}{\left(\sum_{l=1}^N z_l +\delta\right)^2}\dot{V}_i\left(\frac{z_i}{\sum_{l=1}^N z_l +\delta}\right).        \end{equation}
The elements on the diagonal can be written as  
        \begin{align}  \nonumber
             \partial^{2}_{i,i} U_{i} (\bm{z})
&=  \frac{-2\sum_{l\ne i}^N z_i +\delta}{(\sum_{l=1}^N z_l +\delta)^3}\dot{V}_i\left(\frac{z_i}{\sum_{l=1}^N z_l +\delta}\right) \\ \nonumber
& \quad + \left( \frac{\sum_{l \neq i}^N z_l +\delta}{(\sum_{l=1}^N z_l +\delta)^2}\right)^2 \ddot{V}_i\left(\frac{z_i}{\sum_{l=1}^N z_l +\delta}\right)\\ \nonumber
&=\frac{\left [(1-x_i(\bm{z}))^2\ddot{V}_i(x_i(\bm{z}))-2{(1-x_i(\bm{z}))\dot V}_i(x_i(\bm{z})) \right ]}{ \left(\sum_{i=1}^{N} z_j + \delta \right)^{2}}  \\ \label{eq:pseudohessian1}
&= \frac{{f}_i(x_i(\bm{z}))}{m(z)^{2}}. 
        \end{align}
In a similar fashion we obtain the off-diagonal entries 
        \begin{align} \nonumber
\partial^{2}_{i,j} U_i(\bm{z})
&=\frac{\sum_{l=1}^N z_l +\delta - 2\left(\sum_{l \neq i}^N z_l +\delta\right)}{\left(\sum_{l=1}^N z_l +\delta\right)^3}\dot{V}_i\left(\frac{z_i}{\sum_{l=1}^N z_l +\delta}\right)\\ \nonumber 
&\quad \quad -\frac{\sum_{l \neq i}^N z_l +\delta}{\left(\sum_{l=1}^N z_l +\delta\right)^4}\cdot z_i\cdot \ddot{V}_i\left(\frac{z_i}{\sum_{l=1}^N z_l +\delta}\right)\\ \nonumber
&=\frac{\left [ -\ddot{V}_i(x_i(\bm{z})) x_i(\bm{z})(1-x_i(\bm{z}))+{\dot V}_i(x_i(\bm{z})) (2x_i(\bm{z})-1) \right ]}{ \left(\sum_{i=1}^{N} z_j + \delta \right)^{2}}  \\ \label{eq:pseudohessian2}
&=\frac{{g}_i(x_i(\bm{z}))}{ m(z)^{2}}
        \end{align}
which concludes the calculation.  
  \end{IEEEproof}
Lemma~\ref{lem:H_computation} shows that the matrix of second partial derivatives of the functions $U_i$ can be rewritten using a change of variable from $\bm{z}\in \mathcal{R}$ to $\bm{x}(\bm{z})$. Moreover, this reformulated matrix can be decomposed into the sum of a diagonal matrix, with entries $f_i(x_i) - g_i(x_i)$, and a rank-one matrix, whose entries are given by $g_i(x_i)$. Lemma~\ref{thm:SDSC} leverages this particular structure to prove~\eqref{e:suff_cond_SDSC_supp}. 
\begin{lemma}\label{thm:SDSC}
$ \sup_{\bm{x}\in \Delta} \psi_{\bm{r},\bm{V}}(\bm{x}) < 1 \implies \bm{H}_{\bm{r}}(\bm{z})\text{ is negative definite for any $\bm{z}$.} $
\end{lemma}
\begin{IEEEproof}
Let $\bm{r}>\bm{0}$. We assume that $ \sup_{\bm{x}\in \Delta} \psi_{\bm{r},\bm{V}}(\bm{x}) < 1$ and prove that the matrix $ \bm{H}_{\bm{r}}(\bm{z}) $ is negative definite for every $ \bm{z} \in \mathcal{R} $. To this end, we introduce the following notation. We define the vectors $ \bm{g} = ( g_i(x_i(\bm{z})) )_{i=1}^{n} $ and $ \bm{f} = ( f_i(x_i(\bm{z})) )_{i=1}^{n} $, and set $ \tilde{k}_i(\bm{z}) = g_i(x_i(\bm{z})) - f_i(x_i(\bm{z})) + m(\bm{z})\ddot{p}_i(z_i) $, so that $ \tilde{\boldsymbol{k}} = ( \tilde{k}_i(\bm{z}) )_{i=1}^{n} $. Here, $ \odot $ denotes the Hadamard (elementwise) product, $ \mathbf{1} $ is the vector of ones, and $ \Lambda(\bm{a}) $ represents the diagonal matrix whose diagonal entries are the components of $\bm{a}$.

Thanks to Lemma~\ref{lem:H_computation}, the matrix can be expressed as
\begin{align}
   \hskip-3mm m(\bm{z})\cdot (\bm{H}_r(\bm{z}))_{i,j} =
    \begin{cases}
        2 r_i \Bigl( f_i(x_i(\bm{z})) - m(\bm{z})\, \ddot{p}_i(z_i) \Bigr), & \text{if } i = j, \\
        r_i\, g_i(x_i(\bm{z})) + r_j\, g_j(x_j(\bm{z})), &  \text{if } i \neq j.
    \end{cases}
\end{align}
Thus, the matrix $ \bm{H}_r(\bm{z}) $ can be written compactly as
\begin{equation}\label{eq:H}
    \bm{H}_{\bm{r}}(\bm{z}) = -2\, \Lambda\Bigl( \tilde{\boldsymbol{k}} \odot \bm{r} \Bigr) + \Bigl( \bm{g} \odot \bm{r} \Bigr) \mathbf{1}^{\intercal} + \mathbf{1} \Bigl( \bm{g} \odot \bm{r} \Bigr)^{\intercal}.
\end{equation}

Notice that $ \tilde{k}_i(\bm{z}) \geq k_i(x_i(\bm{z})) $. Moreover, replacing by the expressions of $g_i$ and $f_i$ in \eqref{e:g_expression} and \eqref{e:f_expression}, we obtain, 
\begin{align}
k_i(x_i) = (x_i-1)\,\ddot{V}_i(x_i) + \dot{V}_i(x_i) + \delta^2\, \min_{z_i \in \mathcal{R}_i} \ddot{p}_i(z_i).
\end{align}
By Assumption~\ref{assum:V_properties}, we have $ \dot{V}_i>0 $, $ \ddot{V}_i<0 $, and $ \ddot{p}_i>0 $. Moreover, since $ x_i(\bm{z})<1 $ (as $ \bm{x}\in \Delta $), it follows that $ k_i(x_i(\bm{z}))>0 $. Define also the vector $\bm{k} = \bigl( k_i(x_i(\bm{z})) \bigr)_{i=1}^{N}$.

Let $ \bm{v} \in \mathbb{R}^{N} $. Using the expression of $\bm{H}_r(\bm{z})$, we have
\begin{align}\nonumber 
    \bm{v}^T \bm{H}_{\bm{r}}(\bm{z}) \bm{v} &= -2\left(\sum_{i=1}^{n} v_i^2\, r_i\, \tilde{k}_i(\bm{z}) - \langle \bm{v},\, \bm{g} \odot \bm{r} \rangle\, \langle \bm{v},\, \mathbf{1} \rangle \right) \\ \label{e:04/02}
    &\le -2\left(\sum_{i=1}^{n} v_i^2\, r_i\, k_i(x_i(\bm{z})) - \langle \bm{v},\, \bm{g} \odot \bm{r} \rangle\, \langle \bm{v},\, \mathbf{1} \rangle \right),
\end{align}
where $ \langle \cdot,\cdot \rangle $ denotes the inner product.

We now rewrite the inner products in \eqref{e:04/02} by noting that, for any vector $ \bm{a} $, the notation $ \sqrt{\bm{a}} $ represents the vector obtained by taking the square root of each component of $\bm{a}$, and the ratio of two vectors is computed elementwise. In particular,
\begin{align}
    \langle \bm{v},\, \bm{g} \odot \bm{r} \rangle &= \left\langle \bm{v}\odot \sqrt{\bm{k}\odot \bm{r}},\, \frac{\sqrt{\bm{r}}\odot \bm{g}}{\sqrt{\bm{k}}} \right\rangle, \\[1mm]
    \langle \bm{v},\, \mathbf{1} \rangle &= \left\langle \bm{v}\odot \sqrt{\bm{k}\odot \bm{r}},\, \frac{\mathbf{1}}{\sqrt{\bm{k}\odot \bm{r}}} \right\rangle.
\end{align}
Applying the Cauchy-Schwarz inequality to these expressions and substituting into \eqref{e:04/02} yields
\begin{align}\nonumber
%   & \bm{v}^T \bm{H}_r(\bm{z}) \bm{v} \\
%&\le -2\sum_{i=1}^{n} v_i^2\, r_i\, k_i(x_i(\bm{z})) \\ \nonumber
& \bm{v}^T \bm{H}_{\bm{r}}(\bm{z}) \bm{v} \le -2\sum_{i=1}^{n} v_i^2\, r_i\, k_i(x_i(\bm{z})) \\ \nonumber
& + 2\left(\sum_{i=1}^{n} v_i^2 r_i k_i(x_i(\bm{z}))\right) \sqrt{\sum_{i=1}^{n} \frac{r_i\, g_i(x_i(\bm{z}))^2}{k_i(x_i(\bm{z}))}} \sqrt{\sum_{i=1}^{n} \frac{1}{r_i k_i(x_i(\bm{z}))}} \\\nonumber
 &= 2\underbrace{\sum_{i=1}^{n} v_i^2 r_i k_i(x_i(\bm{z}))}_{>0} \left( -1 + \sqrt{\sum_{i=1}^{n} \frac{r_i g_i(x_i(\bm{z}))^2}{k_i(x_i(\bm{z}))} \sum_{i=1}^{n} \frac{1}{r_i k_i(x_i(\bm{z}))}} \right),
\end{align}
%Given that $\mathcal{G}\in \mathcal{F}_{\bm{r}}$, we deduce that $\bm{H}_{\bm{r}}(\bm{z})$ is negative definite.  
which concludes the proof.
\end{IEEEproof}

% \footnote{Cauchy-Schwarz inequality :$\langle \bm a, \bm b \rangle \leq \|\bm a\| \cdot \|\bm b\| = \sqrt{\sum_{i=1}^{N} \bm a_i^2}\sqrt{\sum_{i=1}^{N} \bm b_i^2}$}, 

% By Lemma~\ref{lem:H_computation}, 
% \begin{align}
%    m(\bm z)^{2} (\bm{H}_r(\bm{z}))_{i,j} =
%     \begin{cases}
%         2 r_i \left(f_i(x_i(\bm{z})) - m(z)^2 \ddot{p}_i(z_i)\right), & \text{if } i = j, \\
%         r_i g_i(x_i(\bm{z})) + r_j g_j(x_j(\bm{z})), & \text{otherwise.}
%     \end{cases}
% \end{align}

   Assume that $V_i(x) = a_i \log(x)+ d_i$, $\delta>0$, and $r_i =a_i^{-1}$. Straightforward calculations show that $g_i(x)= a_i$ and $k_i(x) = \nicefrac{a_i}{x^2} + \delta^2 L_i$. Thus $\psi_{\bm{r},\bm{V}}(\bm{x})$ verifies, 
    \begin{align}
        \psi_{\bm{r},\bm{V}}(\bm{x})
&= \left( \sum_{i=1}^{n} \frac{a_i}{\nicefrac{a_i}{x_i^2} + \delta^2 L_i) }   \right) \left( \sum_{i=1}^{n} \frac{1}{\nicefrac{1}{x_i^2} + \delta^2 a_i^{-1}L_i} \right) \\ 
& \leq   \left (\sum_{i=1}^{n} x_i^2 \right)^{2}  \leq \left (\sum_{i=1}^{n} x_i \right)^{4} \leq \left( \frac{C}{C+\delta} \right)^4 < 1. 
    \end{align}
Thus the condition $\sup_{\bm{x}\in \Delta} \psi_{\bm{r},\bm{V}}(\bm{x}) < 1$ is satisfied for logarithmic utilities. This finishes the proof.

    % following valuation functions $ V_i(\cdot) $ with a linear cost function $ p_i(\cdot) $ for $i\in[N]$. 
    % \begin{enumerate}
    %     \item $ V_i(x_i) = \log(x_i) $, for $r_i= 1$ .
    %     \item $ V_i(x_i) = a_i \log(x_i) + b_i $, for $ r_i=1/ a_i$, where $ a_i, b_i $ are constants.
    % \end{enumerate}
    % The corresponding game $\mathcal{G}$ belong to $\mathcal{F}_{\bm r}$

% \begin{table}[h!]
% \centering
% \begin{tabular}{c|l}
% \hline 
% \textbf{Notation} & \textbf{Description} \\
% \hline\hline
% $m(\bm z)$ & $\sum_{j=1}^Nz_j + \delta$\\
% &\\
% $\Lambda(\bm a) $ & $\diag(a_i)$\\
% &\\ 
% $f(\bm a)$ & $\left ( f(a_i) \right)_{i=1}^N,$ $\quad$  e.g: $\sqrt{\bm a}=\left(\sqrt{a_i} \right)_{i=1}^N$\\
% &\\
% $\bm a\odot \bm b$ & $\left ( a_i b_i \right)_{i=1}^N, \quad $ e.g: $ \sqrt{\bm r}\odot \bm g = \left(\sqrt{r_i}g_i \right)_{i=1}^N$ 
% \end{tabular}
% \caption{Notations }
% \label{tab:notations}
% \end{table}

%We further investigate a sufficient conditions that satisfies Proposition \ref{prop:sufDSC}

\subsection{Additional experiments} 
\label{appendix:hetero_dyn_daq_br}

%%%%%%%%%%%%%%%%%%%%%% OGD vs DAQ %%%%%%%%%%%%%%%%%%
%\YBnote{We need to add some text here explaining briefly what these figures are showing.}

\begin{figure*}[t]
    \centering

    % --- Optional global legend ---
    \includegraphics[width=0.5\linewidth]{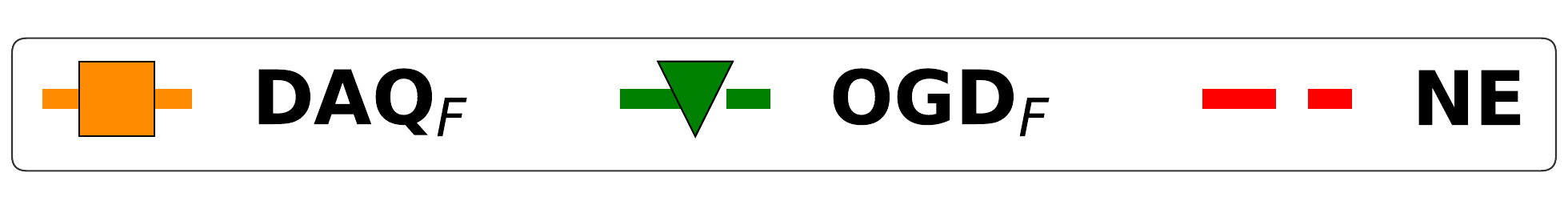}

    % ---------------- Row 1 (20, 50, 80) ----------------
    \begin{subfigure}[t]{0.32\linewidth}
        \centering
        \includegraphics[width=\linewidth]{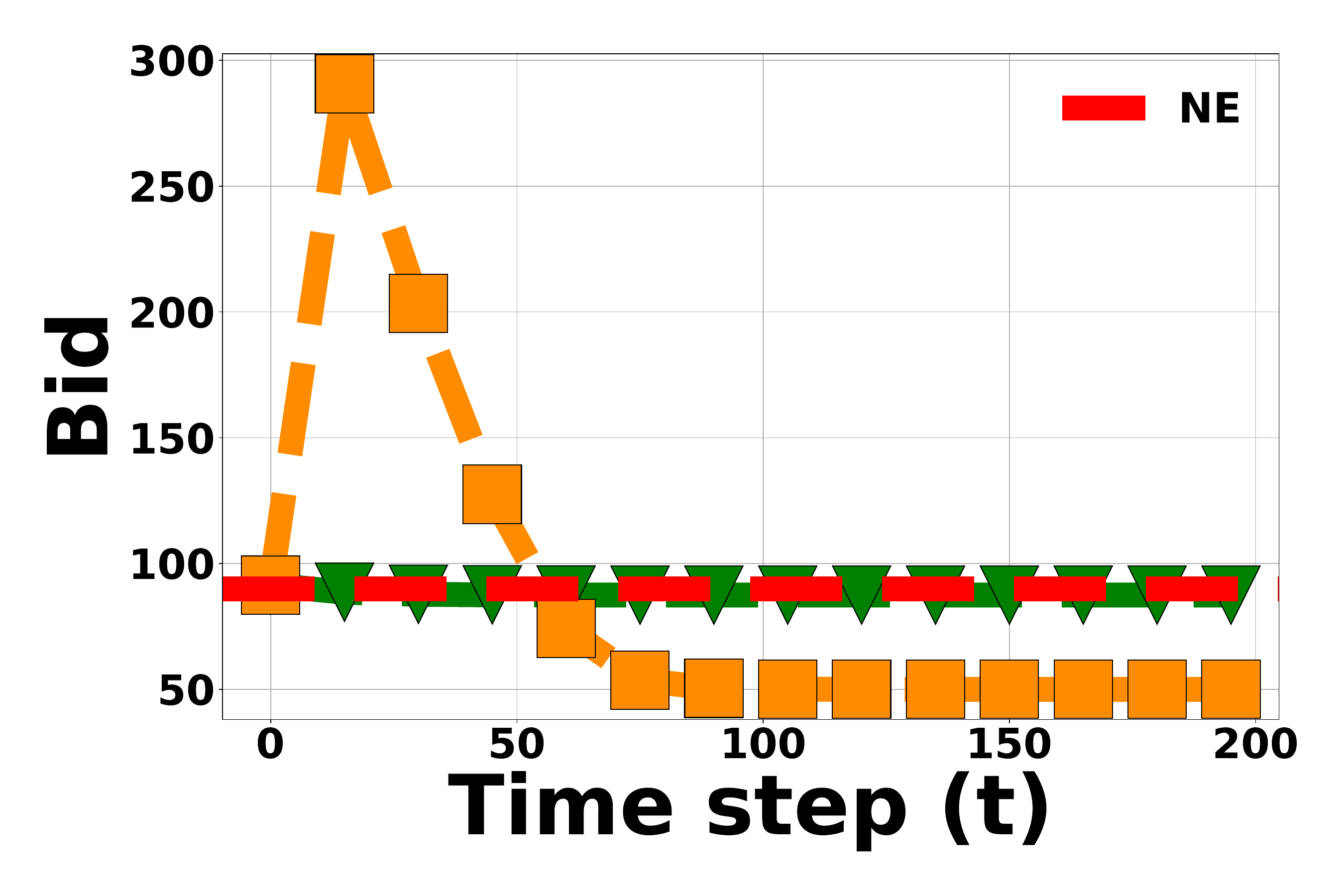}\vspace{-0.2em}
        \includegraphics[width=\linewidth]{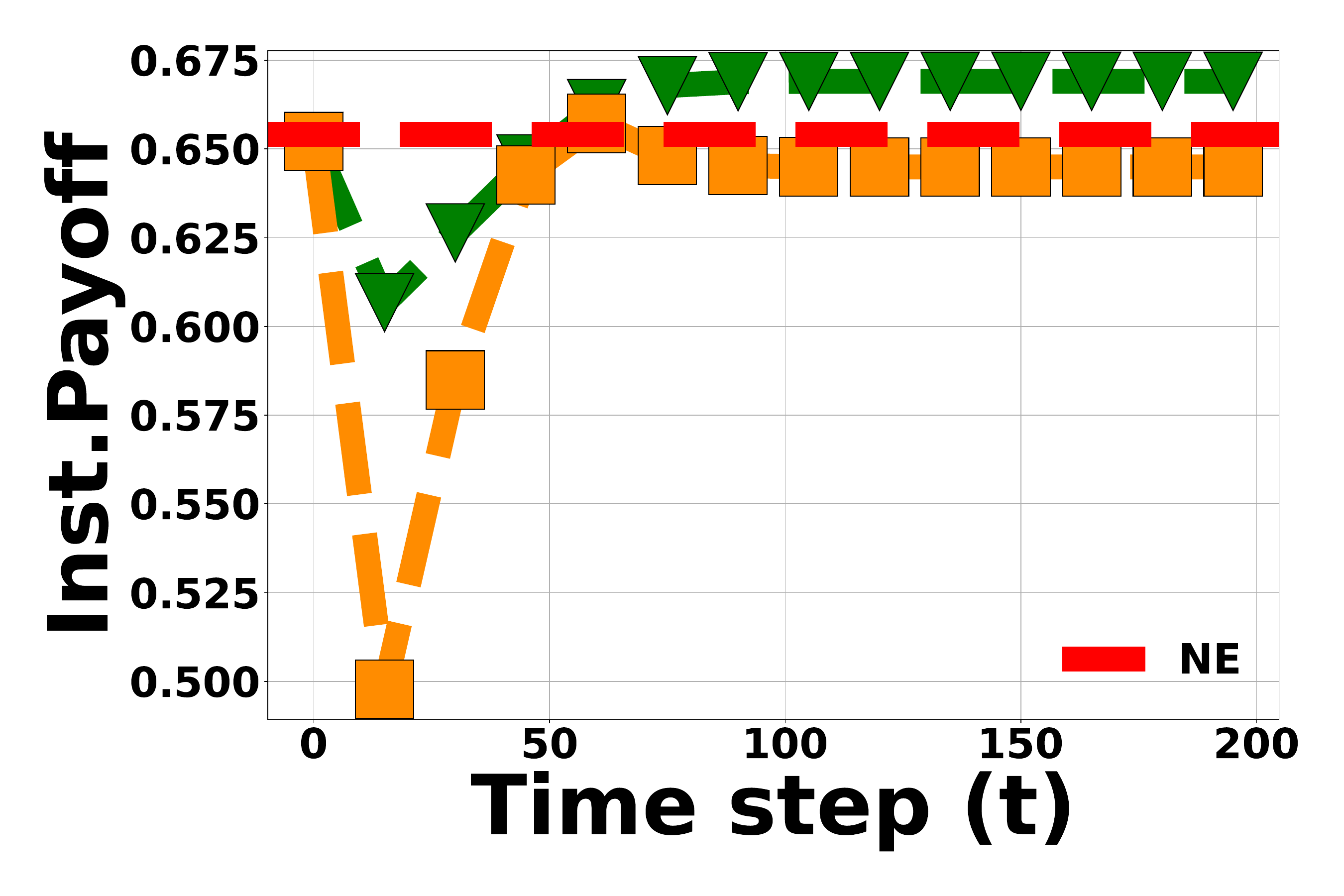}
        \caption{$\alpha_{\DAQ_{\mathrm{F}}}=20\%$}
        \label{fig:DAQ_OGD_20}
    \end{subfigure}\hfill
    \begin{subfigure}[t]{0.32\linewidth}
        \centering
        \includegraphics[width=\linewidth]{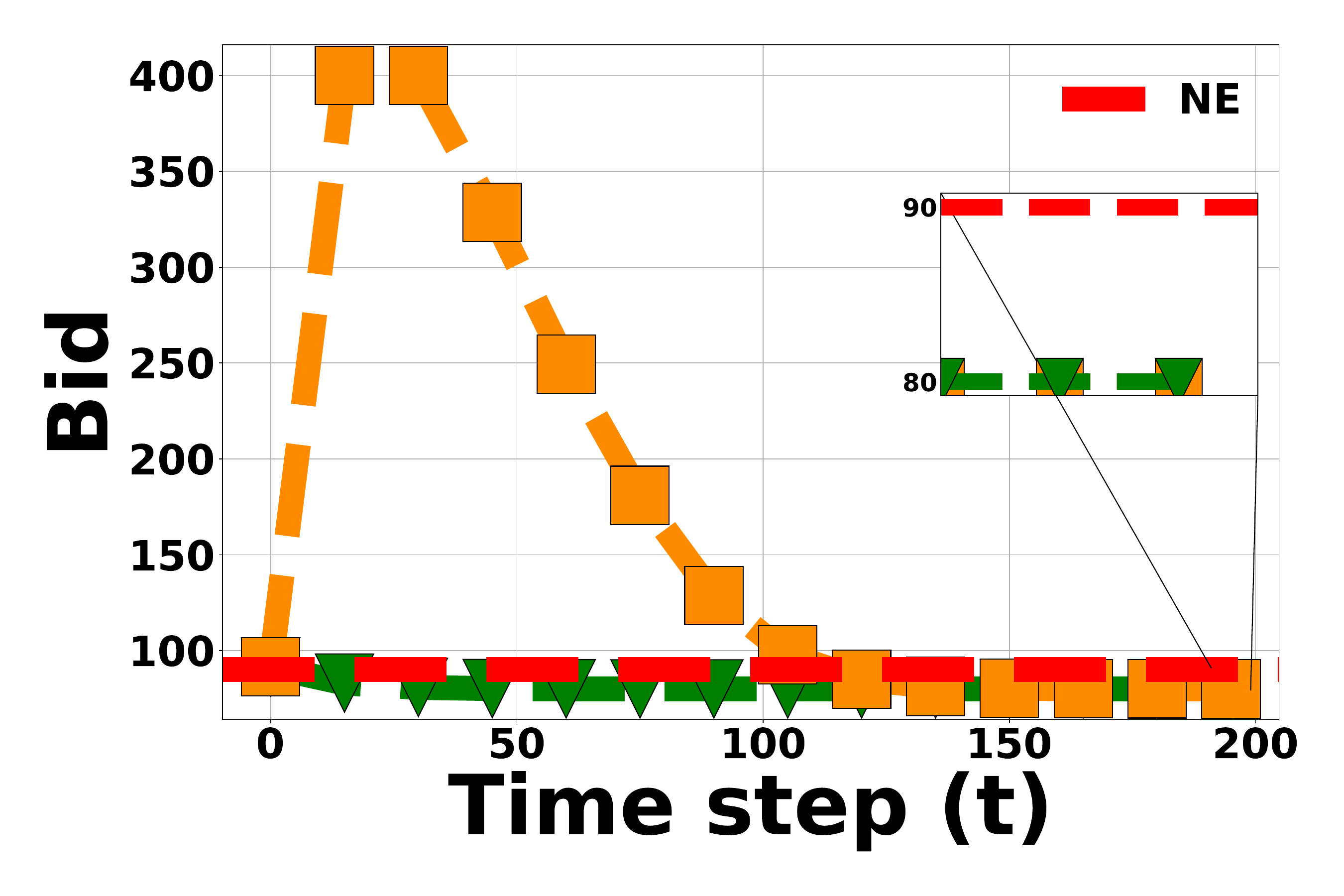}\vspace{-0.2em}
        \includegraphics[width=\linewidth]{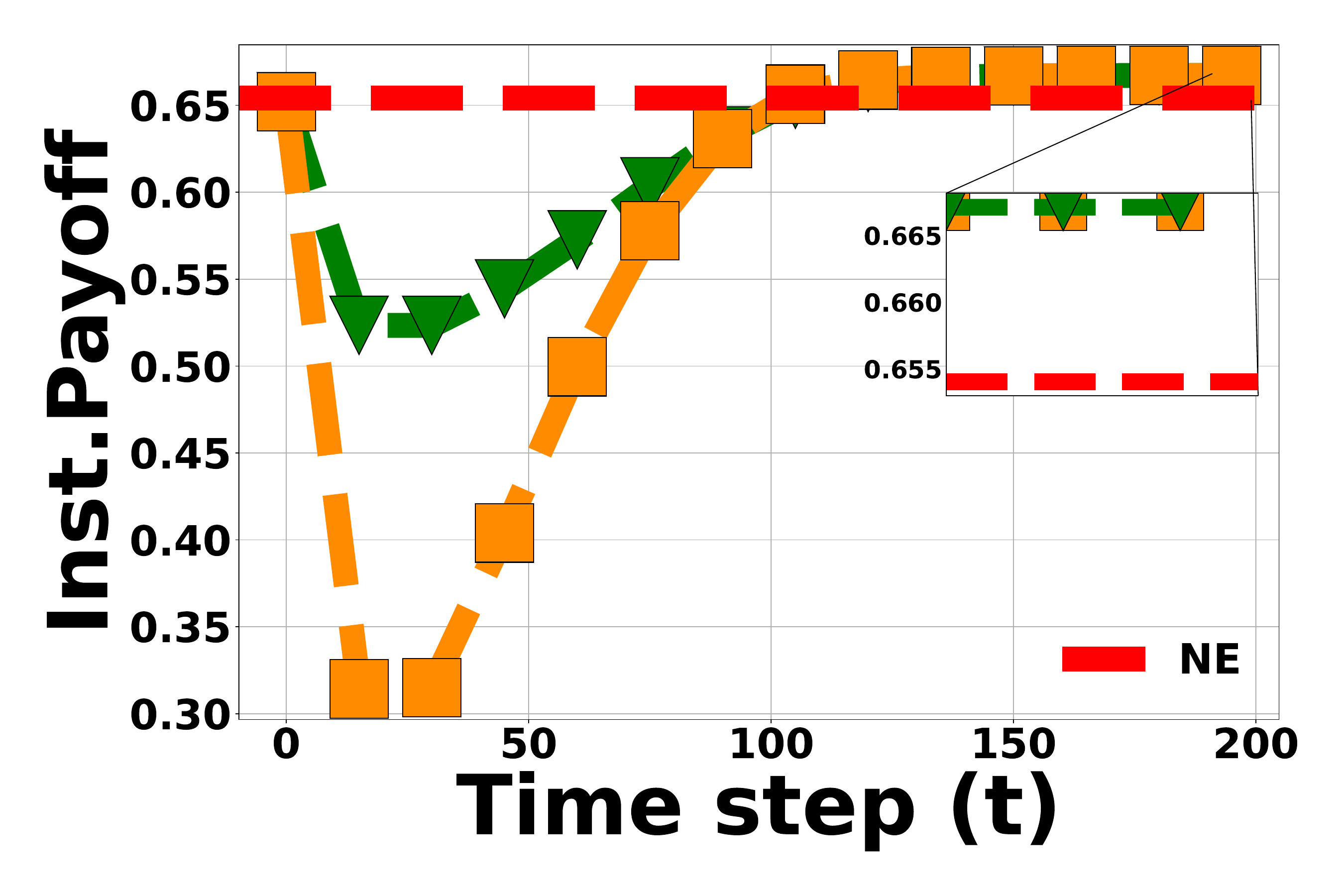}
        \caption{$\alpha_{\DAQ_{\mathrm{F}}}=50\%$}
        \label{fig:DAQ_OGD_50}
    \end{subfigure}\hfill
    \begin{subfigure}[t]{0.32\linewidth}
        \centering
        \includegraphics[width=\linewidth]{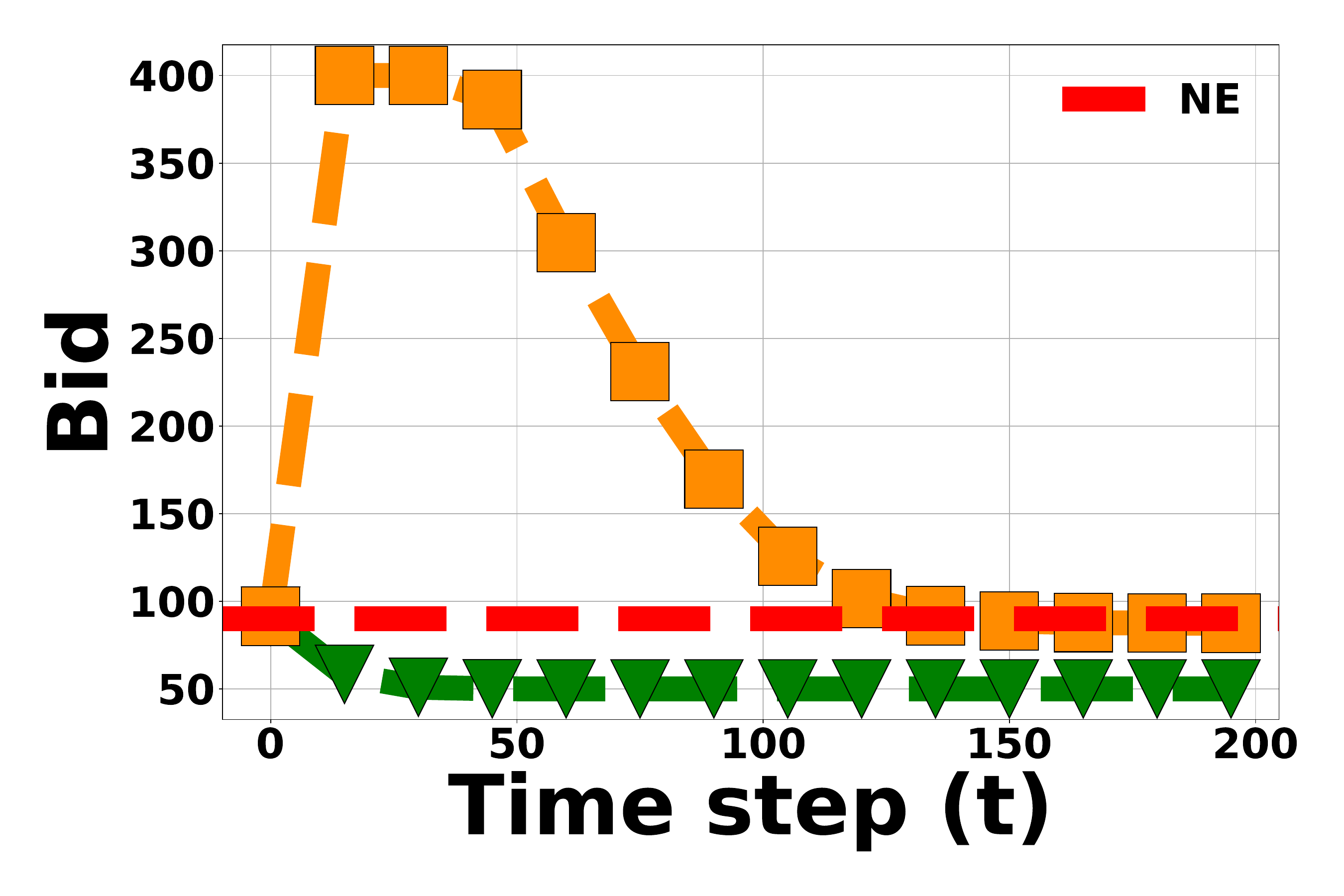}\vspace{-0.2em}
        \includegraphics[width=\linewidth]{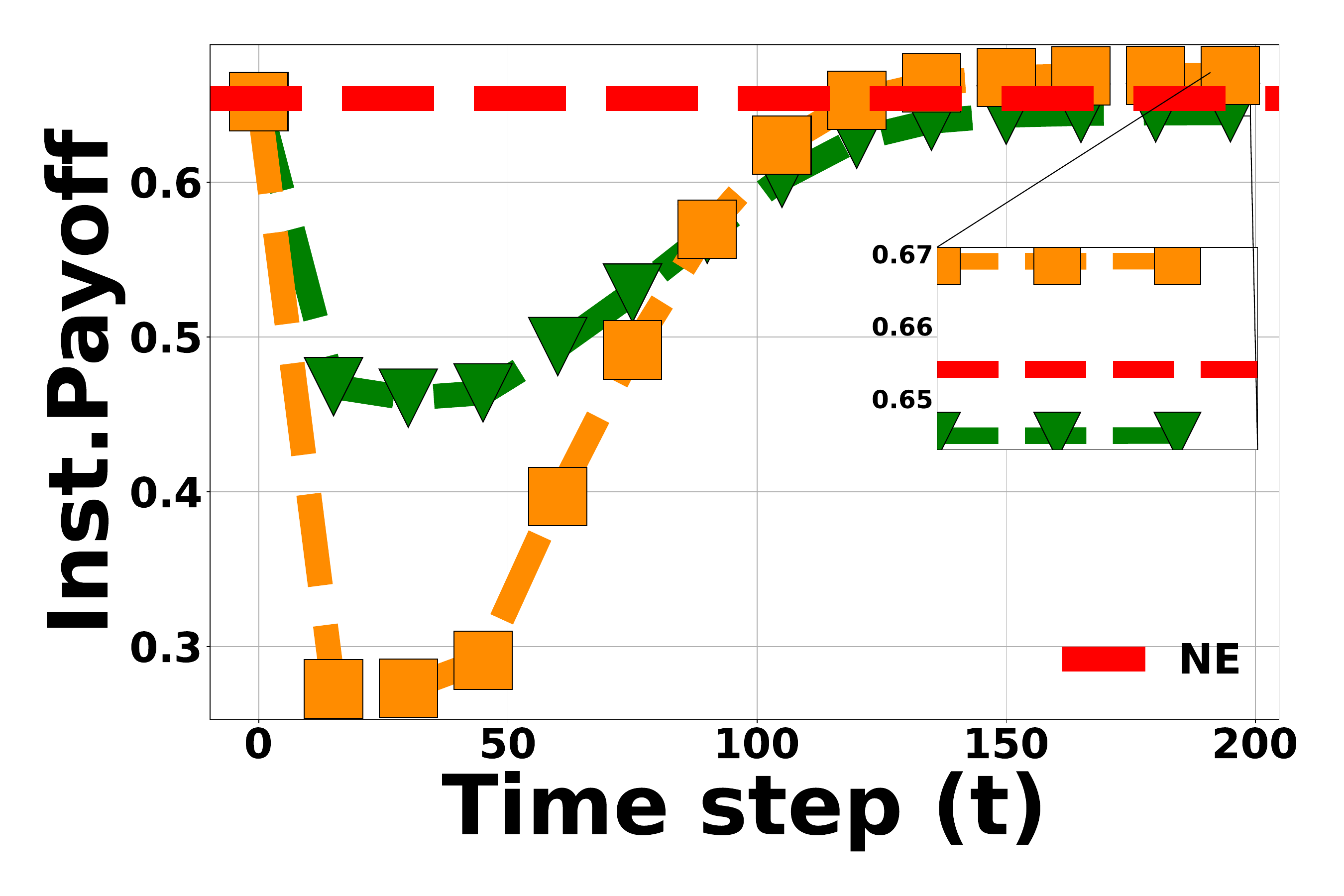}
        \caption{$\alpha_{\DAQ_{\mathrm{F}}}=80\%$}
        \label{fig:DAQ_OGD_80}
    \end{subfigure}

    \vspace{0.6em}

    % ---------------- Row 2 (10, 90) ----------------
    \begin{subfigure}[t]{0.49\linewidth}
        \centering
        \includegraphics[width=\linewidth]{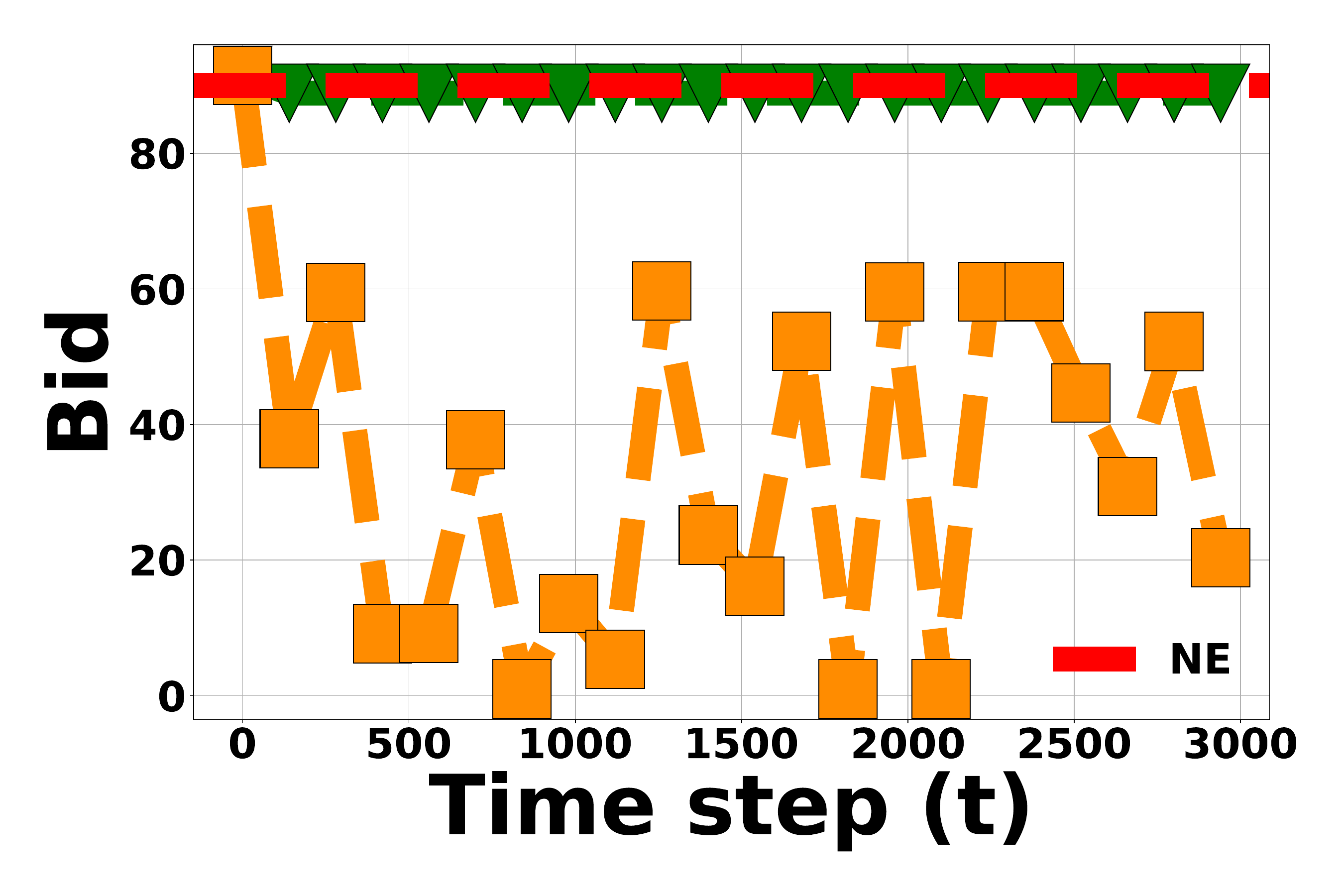}\vspace{-0.2em}
        \includegraphics[width=\linewidth]{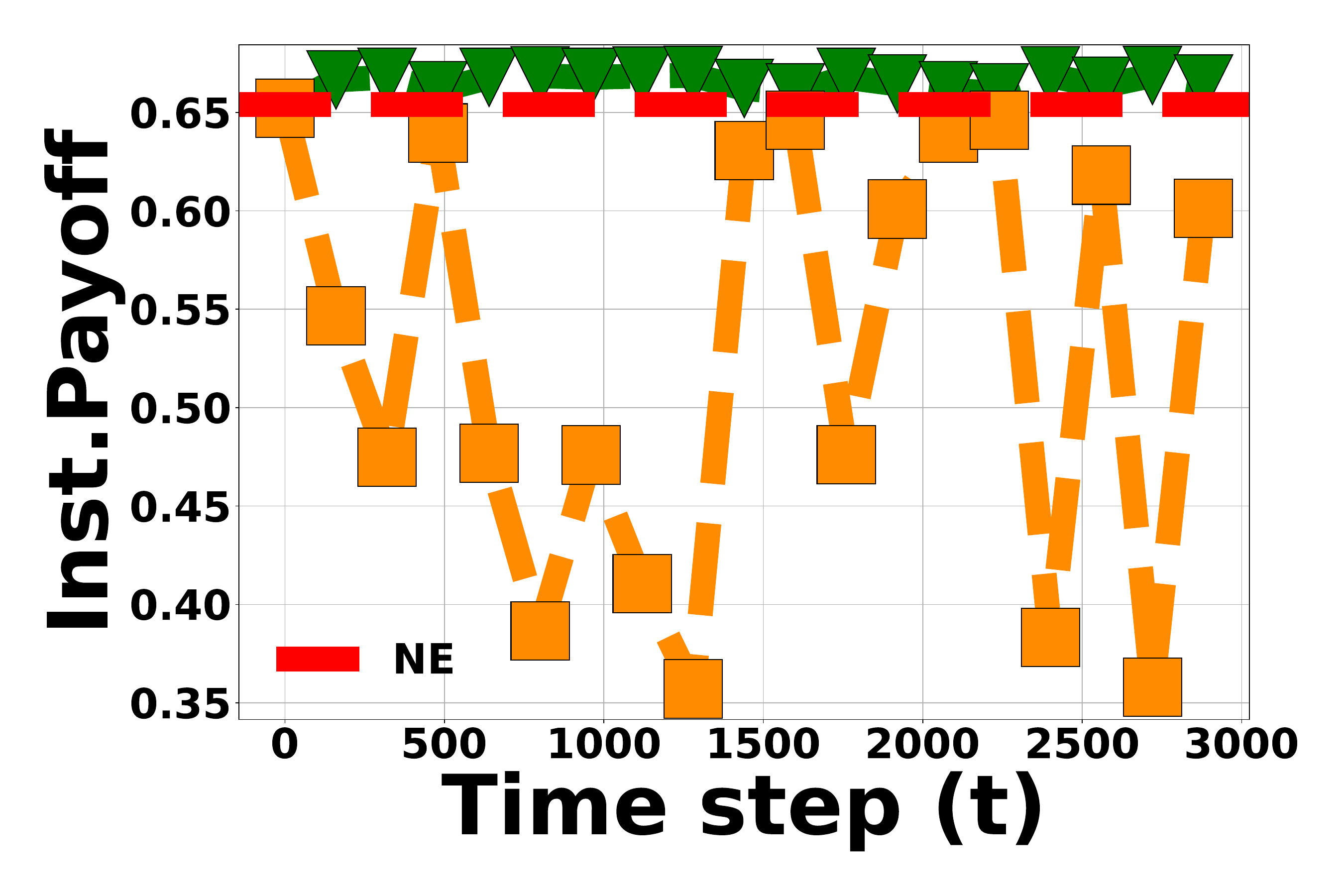}
        \caption{$\alpha_{\DAQ_{\mathrm{F}}}=10\%$}
        \label{fig:DAQ_OGD_10}
    \end{subfigure}\hfill
    \begin{subfigure}[t]{0.49\linewidth}
        \centering
        \includegraphics[width=\linewidth]{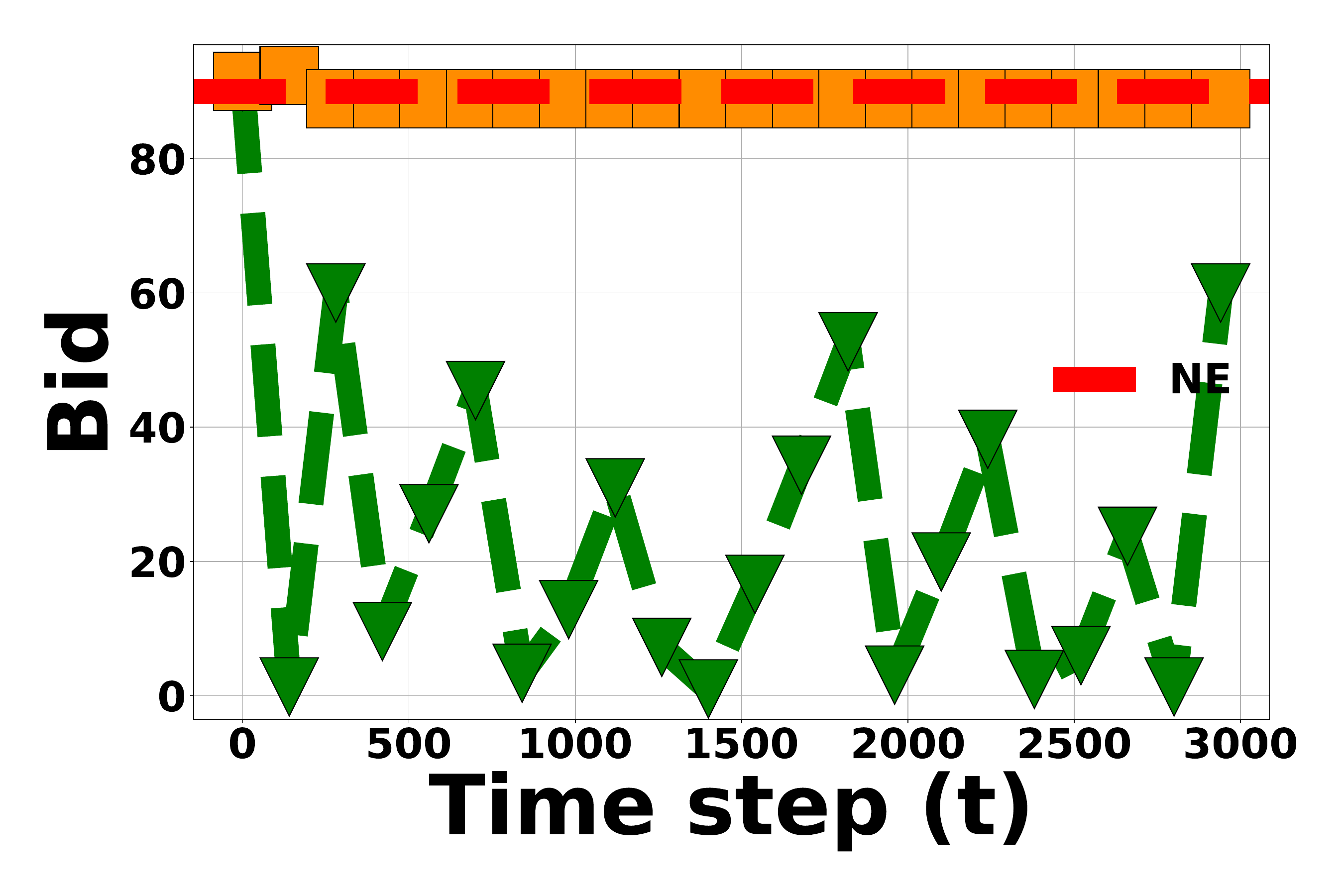}\vspace{-0.2em}
        \includegraphics[width=\linewidth]{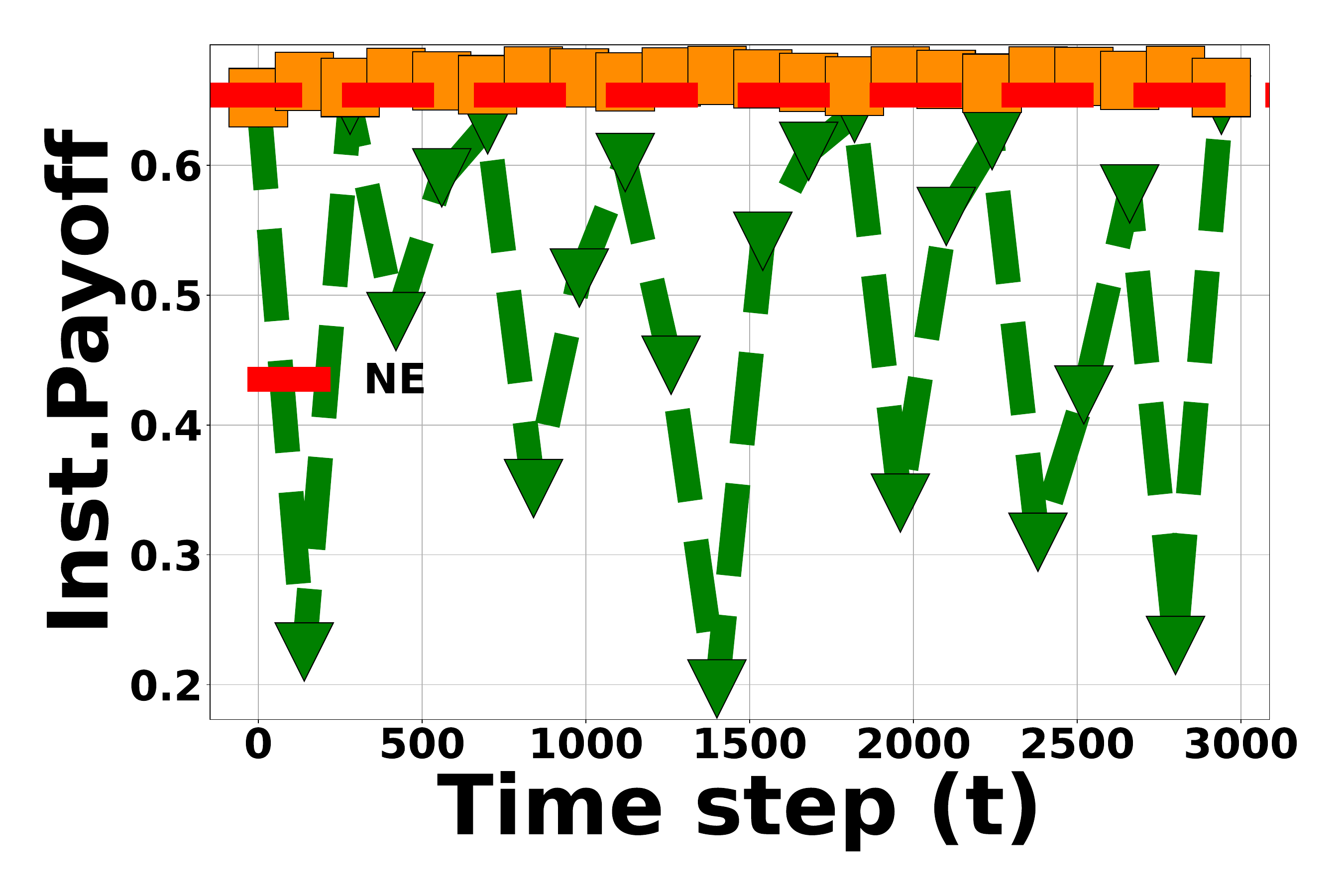}
        \caption{$\alpha_{\DAQ_{\mathrm{F}}}=90\%$}
        \label{fig:DAQ_OGD_90}
    \end{subfigure}

    \caption{Heterogeneous dynamics: $\DAQ_{\mathrm{F}}$ vs.\ $\OGD_{\mathrm{F}}$. In each subfigure: bids (top) and instantaneous payoff (bottom).}
    \label{fig:DAQ_OGD_hetero_all_alpha}
\end{figure*}

We consider now \textit{heterogeneous dynamics} where a fraction $\alpha_{\DAQ}$ of agents uses $\DAQ$ while the remaining agents use $\OGD$ with $\gamma=0$. Figures~\ref{fig:DAQ_OGD_hetero_all_alpha} shows the evolution over time of the instantaneous bid and payoff of two representative agents---one using $\OGD$ and the other using $\DAQ$---for $\alpha_{\DAQ}\in \{10\%, 20\%, 50\%, 80\%, 90\% \}$. Convergence is observed only for $\alpha_{\DAQ}\in[20\%,80\%]$. In this regime, the algorithm used by the majority of agents achieves a utility above the NE-utility, while the minority attains a lower payoff. The dynamics are nearly symmetric for $\alpha_{\DAQ}=50\%$. On the other hand, for extreme splits $\alpha_{\mathcal A_1}\in\{10\%,90\%\}$, we observe oscillations.

% \input{journal/appendix}
% \label{sec:appendix}
\begin{comment}
\newpage 

\subsection{Draft experiments}
\input{journal/draft_experiments}
\end{comment}

%\color{blue}
%\section{$\alpha$-fair functions}
%\input{journal/alpha_fair_functions}
%\color{black}

%\subsection{Best Response Dynamic: Convergence of $\alpha$-fair functions} 
%\label{sec:BR_alpha_fair} 
%\input{journal/BR_alpha_fair}

%\section{Best Response Dynamic: Convergence} 
%\color{blue}
%\label{sec:BRe} 
%\input{journal/BRe}

%%%%%%%%%%%%%%%%%%%%%%%%%%%%%%%%%%%%%%%%%%%%%%%%%%%%%%%%%%%%
%%%%%%%%%%%%%%%%%%%%%%%%%%%%%%%%%%%%%%%%%%%%%%%%%%%%%%%%%%%%%%%%%%%%%%%%%%%%%%%%%%%%%%%%%%%%%%%%%%%%

%%%%%%%%%%%%%%%%%%%%%%%%%%%%%%%%%%%%%%%%%%%%%%%%%%%%%%%%%%%%%%%%%%%%%%%%%%%%%%%%%%%%%%%%%%%%%%%%%%%%
%%%%%%%%%%%%%%%%%%%%%%%%%%%%%%%%%%%%%%%%%%%%%%%%%%%%%%%%%%%%

% \section{Related Work} 
% \label{sec:Related_work}
% \input{journal/related_work}

% \color{red}

% \section{Related work}\label{sec:relatedwork}
% \input{related_work}
\end{document}